\pdfoutput=1
\RequirePackage{ifpdf}
\ifpdf 
\documentclass[pdftex]{sigma}
\else
\documentclass{sigma}
\fi

\usepackage{booktabs} 
\usepackage{wasysym} 
\usepackage{bm}
\usepackage{tikz-cd}
\usetikzlibrary{matrix,arrows,calc,positioning,shapes}
\usepackage{calc}
\usepackage{mathrsfs}
\usepackage{mathtools}
\usepackage{lscape}
\usepackage{rotating}

\mathtoolsset{showonlyrefs=true} 

\numberwithin{equation}{section}

\newtheorem{thm}{Theorem}[section]
\newtheorem{lem}[thm]{Lemma}
\newtheorem{prp}[thm]{Proposition}
\newtheorem{cor}[thm]{Corollary}

\theoremstyle{definition}

\newtheorem{rmk}[thm]{Remark}

\makeatletter
\renewcommand*\env@matrix[1][*\c@MaxMatrixCols c]{%
 \hskip -\arraycolsep
 \let\@ifnextchar\new@ifnextchar
 \array{#1}}
\makeatother

\makeatletter
\DeclareRobustCommand{\rvdots}{%
 \vbox{
 \baselineskip4\p@\lineskiplimit\z@
 \kern-\p@
 \vspace{1pt}\hbox{.}\hbox{.}\hbox{.}
 }}
\makeatother

\usepackage{etoolbox} 
\AtBeginEnvironment{pmatrix}{\everymath{\displaystyle}}
\AtBeginEnvironment{bmatrix}{\everymath{\displaystyle}}

\tikzset{ampersand replacement=\&}

\allowdisplaybreaks[1] 

 \newcommand{\oo}{\infty}
 \newcommand{\A}{\mathcal{A}}
 \newcommand{\B}{\mathcal{B}}
\renewcommand{\d}{\mathrm{d}}
 \newcommand{\C}{\mathcal{C}}
 \newcommand{\CC}{\mathbb{C}}
 \newcommand{\D}{\mathcal{D}}
 \newcommand{\DD}{\mathfrak{D}}

 \newcommand{\K}{\mathcal{K}}
 \newcommand{\M}{\mathcal{M}}
 \newcommand{\BL}{\mathrm{BL}}
 \newcommand{\MP}{\mathrm{MP}}
 \newcommand{\coker}{\operatorname{coker}}
 \newcommand{\im}{\operatorname{im}}
 \newcommand{\dalf}{{}^4\square}
 \newcommand{\del}{\partial}
 
 \newcommand{\gf}{{}^4g}
 \newcommand{\grf}{{}^4\nabla}
 \newcommand{\id}{\mathrm{id}}
 \newcommand{\ka}{\kappa}
 \newcommand{\la}{\lambda}
 \newcommand{\sqf}{{}^4\square}
 \newcommand{\eps}{\varepsilon}
 \newcommand{\Res}{\operatorname{Res}}

 \newcommand{\tr}{\mathrm{tr}}

\makeatletter
\newcommand{\superimpose}[2]{%
 {\ooalign{$#1\@firstoftwo#2$\cr\hfil$#1\@secondoftwo#2$\hfil\cr}}}
\makeatother

\newsavebox{\sqzerbox}
\savebox{\sqzerbox}{$\mathpalette\superimpose{{\square}{\raisebox{.2ex}{\scalebox{.7}{0}}}}$}
\newcommand{\sqzer}{\usebox{\sqzerbox}}
\newsavebox{\sqonebox}
\savebox{\sqonebox}{$\mathpalette\superimpose{{\square}{\raisebox{.2ex}{\scalebox{.7}{1}}}}$}
\newcommand{\sqone}{\usebox{\sqonebox}}
\newsavebox{\sqtwobox}
\savebox{\sqtwobox}{$\mathpalette\superimpose{{\square}{\raisebox{.2ex}{\scalebox{.7}{2}}}}$}
\newcommand{\sqtwo}{\usebox{\sqtwobox}}

\newsavebox{\cronebox}
\savebox{\cronebox}{$\mathpalette\superimpose{{\bigcirc}{\raisebox{.1ex}{\scalebox{.7}{1}}}}$}

\newsavebox{\crtwobox}
\savebox{\crtwobox}{$\mathpalette\superimpose{{\bigcirc}{\raisebox{.1ex}{\scalebox{.7}{2}}}}$}

\begin{document}
\allowdisplaybreaks

\newcommand{\arXivNumber}{2004.09651}

\renewcommand{\PaperNumber}{011}

\FirstPageHeading

\ShortArticleName{Explicit Triangular Decoupling of the Separated Lichnerowicz Tensor Wave Equation}

\ArticleName{Explicit Triangular Decoupling of the Separated\\ Lichnerowicz Tensor Wave Equation on Schwarzschild\\ into Scalar Regge--Wheeler Equations}

\Author{Igor KHAVKINE~$^{\rm ab}$}

\AuthorNameForHeading{I.~Khavkine}

\Address{$^{\rm a)}$~Institute of Mathematics of the Czech Academy of Sciences,\\
\hphantom{$^{\rm a)}$}~\v{Z}itn{\'a} 25, 115 67 Praha 1, Czech Republic}

\Address{$^{\rm b)}$~Charles University in Prague, Faculty of Mathematics and Physics,\\
\hphantom{$^{\rm b)}$}~Sokolovsk\'a 83, 186 75 Praha 8, Czech Republic}
\EmailD{\href{mailto:khavkine@math.cas.cz}{khavkine@math.cas.cz}}

\ArticleDates{Received March 16, 2021, in final form January 22, 2022; Published online February 04, 2022}

\Abstract{We consider the vector and the Lichnerowicz wave equations on the Schwarz\-schild spacetime, which correspond to the Maxwell and linearized Einstein equations in harmonic gauges (or, respectively, in Lorenz and de~Donder gauges). After a complete separation of variables, the radial mode equations form complicated systems of coupled linear ODEs. We outline a precise abstract strategy to decouple these systems into sparse triangular form, where the diagonal blocks consist of spin-$s$ scalar Regge--Wheeler equations (for spins $s=0,1,2$). Building on the example of the vector wave equation, which we have treated previously, we complete a successful implementation of our strategy for the Lichnerowicz wave equation. Our results go a step further than previous more ad-hoc attempts in the literature by presenting a full and maximally simplified final triangular form. These results have important applications to the quantum field theory of and the classical stability analysis of electromagnetic and gravitational perturbations of the Schwarzschild black hole in harmonic gauges.}

\Keywords{Schwarzschild black hole; linearized gravity; harmonic gauge; Regge--Wheeler equation; rational ODE; computer algebra; rational solution; decoupling}

\Classification{35Q75; 34L99; 34L05; 68W30}

\section{Introduction}

It is textbook knowledge that a standard self-adjoint Sturm--Liouville
spectral problem $(p \phi')' - q\phi + \omega^2 w \phi = 0$, with
$p,q,w>0$, has ($i$) a real and positive $\omega^2$-spectrum (or just real
$\omega$-spectrum) on a natural function space, with ($ii$) the generalized
eigenfunctions providing a resolution of the identity. In some
applications (which we describe below) it is crucial to provably know
that properties ($i$) and ($ii$) hold (or do not hold) for some equations
that do not quite fit into this standard class. Unfortunately, in those
cases, the standard Sturm--Liouville theory is no longer applicable and
each case has to be studied individually. This work is dedicated to a
deep study of a specific ordinary differential equation (ODE) system
that naturally appears in the context of linearized gravity on the
spherically symmetric Schwarzschild black hole, which is non-standard in
the above sense: the coefficients $p$, $q$, $w$ are now matrices, either no
longer self-adjoint or no longer positive definite, and the spectral
parameter $\omega$ also appears linearly in the~$q$ coefficient. Despite
these complications and its superficially very unstructured form, this
system turns out to be highly special. Our main result is an explicit
transformation of this complicated ODE system into a much simpler
(sparse, upper triangular) form, with standard self-adjoint
Sturm--Liouville operators on the diagonal. From this simplified form, it
becomes obvious that properties~($i$) and~($ii$) can be proven to hold by
standard theory, while as a byproduct a number of interesting geometric
features of the equations of linearized gravity Schwarzschild (under the
harmonic gauge condition) are discovered, including an alternative
variational formulation, non-trivial generalized symmetries and the
existence of Debye potentials (all of these results are concisely
collected in Section~\ref{sec:appl}). Our methods are rather special to
this ODE system, but having synthesized previous ad-hoc results from the
literature into an abstract strategy based on ideas from homological
algebra, the geometry and algebra of differential equations, and the
theory of rational ODEs, they may also be applicable in other cases.
Last but not least, we know of no other tools that can produce any of
the same results for the given ODE system.

Since the above topics do not often appear together, let alone applied to
problems motivated by linearized gravity, we have aimed for the
presentation to be self-contained, though as concise as possible. That
balance and the intrinsic complexity of the ODE system under study bear
responsibility for the length of this work.

{\bf Physical motivation.}
The study of linear metric perturbations around a Schwarzschild black
hole was initiated in~\cite{regge-wheeler}. Since then, a rich
literature has developed on this topic, including extensions to other
linear fields and more general kinds of black holes~\cite{berti-qnm, chandrasekhar,
frolov-novikov}, with important applications, including the
modeling of gravitational wave forms~\cite{sasaki-lrr}, rigorous linear
stability of black holes~\cite{dr-lect, tataru-lindecay} (geared towards
the nonlinear stability problem), and the study of quantum effects
around black holes~\cite{frolov-novikov}.

Due to diffeomorphism invariance, the linearized Einstein equations
require gauge fixing to get well-posed initial value or inhomogeneous
problems. Different gauge choices have different merits. The
mathematical literature on linear metric perturbations of Schwarzschild
has been dominated by the question of classical stability of the
classical Cauchy problem. Rather decisive, positive results have already
been obtained in a gauge-invariant formalism~\cite{dotti-prl}, as well
as directly for the metric in a special double null-foliation
gauge~\cite{dhr-schw}, neither approach using mode decompositions. So
why bother studying a different gauge, using mode decomposition, as we
do below? The answer is that there is a wider class of interesting
questions that we can ask about black hole perturbations, where such
information is useful. For instance, applications in quantum field
theory (QFT) require an explicit construction of the linear Green
function, which is difficult to obtain other than by a mode
decomposition. Also, QFT favors specific choices of gauge, with mode
stability in those gauges being a necessary condition for the existence
of a~reasonable vacuum state, while stability in one gauge (favored by
the classical Cauchy problem), does not imply stability in a different
gauge (favored by QFT). See more precise remarks at the end of
Section~\ref{sec:appl}.

In general theoretical treatments of linear metric perturbations, a
common choice is the so-called \emph{harmonic}\footnote{For any background metric $\bar{g}$, the perturbed metric
 $g$ satisfies the non-linear \emph{harmonic} (or \emph{wave map})
 condition when $\bar{\nabla}_\mu (\epsilon_g g^{\mu\nu}) = 0$, where
 $\epsilon_g$ is the perturbed volume form, but $\bar{\nabla}$ is
 the background covariant derivative. This choice of gauge also reduces
 the full Einstein equations to wave-like form and linearizes to our
 choice of gauge. The coordinate based condition $\square_g x^\mu = 0$
 is only equivalent when the background metric is Minkowski~\cite{bicak-katz}.} %
gauge (also known as
\emph{de~Donder} gauge, \emph{wave $($map/coordinate$)$} gauge, or by
analogy with electrodynamics as \emph{Lorenz} gauge) $\nabla^\nu
\overline{p}_{\mu\nu} = 0$, with $\overline{p}_{\mu\nu} = p_{\mu\nu} -
\frac{1}{2} p^\la_\la g_{\mu\nu}$ the trace reversed metric
perturbation, preferred for its local regularity
properties~\cite{barack-ori-gauge} and in applications to quantum field
theory~\cite{bdm, bfr-qg, fewster}. In this gauge, the linearized
Einstein equations take the form of a tensor wave equation, the
\emph{Lichnerowicz wave equation} $\square p_{\mu\nu} - 2\,
R_{(\mu}{}^{\la\ka}{}_{\nu)} p_{\la\ka} = 0$, which is both covariant
and hyperbolic, with both properties of significant theoretical
importance. However, due to its
algebraic complexity, until recently, relatively little work has been
done in this gauge for black hole perturbations, compared to
Regge-Wheeler gauge and its variations.
Notable early uses of harmonic gauge include~\cite{barack-lousto,
berndtson, gpy}. Following~\cite{barack-lousto}, this gauge (there mostly
called \emph{Lorenz}) has found important applications in the gravitational
self-force literature~\cite{barack-pound}. More recently, the
work~\cite{hintz-vasy} has used harmonic gauge in the proof of global
non-linear stability of Kerr-de~Sitter black holes, while
in~\cite{johnson-thesis, johnson-wave}\footnote{Actually, \cite{johnson-thesis,johnson-wave} use a \emph{generalized} harmonic/wave
 gauge, $\nabla^\nu\overline{p}_{\mu\nu} = f_\mu[p]$. In the earlier
 preprint~\cite{johnson-thesis} $f_\mu[p]$ was chosen to be a local
 zeroth order expression, while in the later published
 work~\cite{johnson-wave} $f_\mu[p]$ was chosen to be non-local with
 respect to the orbits of spherical symmetry of Schwarzschild.} %
and \cite{hung-harmonic-odd, hung-harmonic-even} the vector field method
was used, directly in real-space, to study stability and decay of
perturbations of Schwarzschild, with~\cite{johnson-cqg} an addendum that
implicitly considers an upper triangular decoupling of the Lichnerowicz
wave equation under the transverse-traceless conditions, $\nabla^\nu
\overline{p}_{\mu\nu}$ and $p_\lambda^\lambda = 0$, without referring to
or comparing with the fully diagonal decoupling of the same system
by~\cite{berndtson}.

{\bf Mathematical problem.}
To date, despite the strong motivations listed above and the known
complete separability of the Lichnerowicz d'Alembertian, there still
does not exist an explicit mode-level construction of a harmonic gauge
Green function for metric perturbations on Schwarz\-schild. A necessary
step in such a construction would be a determination of the spectrum of
the separated radial mode equation and a proof of completeness of its
generalized eigenfunctions. In this work, we take a significant step in
that direction. Namely, we perform a highly non-trivial simplification
that allows this spectral problem to be studied by standard
Sturm--Liouville theory.

The main obstacle is that the separated radial mode equation on
Schwarzschild is a rather complicated system of ODEs. While it can be
put into matrix Sturm--Liouville form, it presents a highly non-standard
spectral problem. For one, it is naturally self-adjoint only with
respect to an indefinite functional inner product. So either it becomes
a self-adjoint spectral problem on a~\emph{Krein space} (analogous to a
Hilbert space, but with an indefinite inner product), or a
non-self-adjoint problem on a Hilbert space (where we artificially positivise
the natural inner product). Further, it presents a \emph{quadratic
eigenvalue problem} in the frequency $\omega$, rather than a regular
eigenvalue problem with respect to $\omega^2$. In either case, the
standard spectral theory of self-adjoint operators on Hilbert space
becomes totally inapplicable, leaving us without any obvious way to
prove (or disprove) the reality of the spectrum (absence of modes
exponentially unstable in time) or even to check the completeness of
(generalized) eigenfunctions. The indefiniteness of the functional inner
product is ultimately due to the Lorentzian signature of the metric,
meaning that this complication does not occur in the analogous problem
in Riemannian geometry~\cite{gpy}.

{\bf Methodology.}
In this work, building on the strategy outlined and successfully
implemented in~\cite{kh-vwtriang, kh-rwtriang} for the simpler example
of the vector wave equation $\square v_\mu = 0$, we prove that the
harmonic gauge radial mode equation can be decoupled into a triangular
system, where the diagonal blocks are the \emph{spin-$s$ Regge--Wheeler}
equations, with $s=0$, $1$ or $2$. These Regge--Wheeler equations are
then of standard scalar Sturm--Liouville form, with very well-understood
spectral properties~\cite{chandrasekhar, dss}. Such a decoupling
essentially reduces the non-standard spectral problem of the radial mode
equation to the standard and well-understood spectral problem of the
Regge--Wheeler equation. The details of such a spectral analysis are left
for to future work.

To our knowledge, prior to~\cite{kh-vwtriang}, which was our warm-up for
this work, such a triangular decoupling has never been explicitly
discussed in the literature on linearized gravity. Indirect hints of it
have previously appeared only in~\cite{berndtson} (Remark~\ref{rmk:full-triang}) and independently
in~\cite{rosa-dolan} (though only for the vector wave equation in the
latter). Very recently, \cite{hung-harmonic-odd, hung-harmonic-even} has
used some of the resulting formulas from~\cite{berndtson}.
Unfortunately, though pioneering, the original works~\cite{berndtson,
rosa-dolan} were based on rather ad-hoc, extensive, explicit
computations and the full details of how the original radial mode
equations transform into the decoupled form and back are not easy to
understand from these references. We have strived to distill their
overall strategy into conceptually clear terms\footnote{It is by no means a simple task to extract a guiding
 strategy from~\cite{berndtson}, as it only becomes apparent as the
 common pattern in the detailed and explicit calculations done for a
 sequence of examples of increasing complexity. Still, that work should
 be credited as (to our knowledge) the first to carry out this strategy
 explicitly for the Lichnerowicz wave equation, and also incidentally
 for the vector wave equation on Schwarzschild.} %
and to understand why certain key ad-hoc steps were successful, as
recorded and implemented in~\cite{kh-vwtriang, kh-rwtriang} for the
vector wave equation. In the end, our results also go a step further
than~\cite{berndtson, rosa-dolan}, by reducing the resulting triangular
form as much as possible and proving that no further reduction is
possible (in the context of rational ODEs).

{\bf Contents overview.}
In Section~\ref{sec:formal}, we start by reviewing a precise notion of
\emph{morphism} between differential equations, as well as of
\emph{equivalence up to homotopy} between morphisms or between equations
(or \emph{isomorphism}), where we synthesize some elementary notions from
homological algebra~\cite{weibel}, $D$-modules~\cite[Section~10.5]{seiler}
and the categorical approach to differential
equations~\cite[Section~VII.5]{vinogradov}. Essentially, a morphism is a
differential operator mapping solutions to solutions, but also equipped
with extra data preserving the structure of the equations.\footnote{Those familiar with the $SE = OT$ identity that is used
 prominently in Wald's \emph{adjoint method}~\cite{wald-adjoint} might
 be comforted to know that such an identity is in fact an example of a
 morphism. Another example is Chandrasekhar's transformation
 (Section~\ref{sec:rwz}). Many examples of similar transformations in
 black hole perturbation theory are reviewed in~\cite{gjk-darboux}.} %
Then, we review how a hierarchical separation a differential equations
into \emph{gauge modes}, \emph{gauge invariant modes} and
\emph{constraint violating modes} can lead to an equivalence with an
equation in block upper triangular form. Our general decoupling strategy
is to use this step recursively until a full triangular form is reached.
The second step in the strategy is to simplify this initial triangular
form further by transforming as many as possible off-diagonal components
to zero or to some non-zero canonical form by the methods to be
presented in Section~\ref{sec:rw}. It is remarkable that all these
equations and transformations between them involve only differential
operators with rational coefficients. This review is a condensed version
of the more detailed discussion given in~\cite{kh-vwtriang}.

In Section~\ref{sec:rw}, we review some tools from the study of rational
solutions of \emph{rational ODEs} (ODEs with rational coefficients) and
apply them to the spin-$s$ Regge--Wheeler equations. Namely, we reduce
the problem of deciding when a triangularly coupled Regge--Wheeler system
can be made diagonal to a finite dimensional linear problem that can be
easily solved with computer algebra. Incidentally, for spins $s=0,1,2$
that are of interest to us, we prove for the first time that the only
non-trivial rational morphisms (in the sense of
Section~\ref{sec:formal}) between spin-$s$ Regge--Wheeler equations are
the identity morphisms for equal $s$. (In other words, there do not
exist rational \emph{spin raising} and \emph{lowering
operators}~\cite{whiting-shah-spheroidal} at least between radial modes
of spins $s=0,1,2$, and also likely any other spins.) These methods were
first applied in~\cite{kh-rwtriang}. In this work, they are further
improved and refined by new results on the characterization of the
cokernel of a rational ODE. Finally, in Section~\ref{sec:rwz}, we also
discuss the Chandrasekhar transformation as an equivalence between the
$s=2$ Regge--Wheeler and Zerilli equations, with special attention to its
dependence on the frequency $\omega$.

In Section~\ref{sec:schw}, we first review the complete separation of
variables of the scalar, vector and Lichnerowicz wave equations into
spherical and time harmonic modes. Our notation and conventions are
compared with the literature in Appendix~\ref{sec:spherical}. Then we
present the resulting radial mode equations, related differential
operators and identities between them that are needed to implement the
decoupling strategy from Section~\ref{sec:formal}. In
Section~\ref{sec:vw} we recall from~\cite{kh-vwtriang} and slightly
improve the explicit full triangular decoupling for the radial vector
wave equation. In Section~\ref{sec:lich} we give for the first time a
explicit full triangular decoupling for the radial Lichnerowicz wave
equation, which follows from the decoupling strategy of
Section~\ref{sec:formal} and builds on the results of
Section~\ref{sec:vw}. Since some of the explicit formulas needed for the
even sector of the Lichnerowicz equation are quite lengthy, they are
relegated to Appendix~\ref{sec:lichev-formulas}. In both the vector and
Lichnerowicz wave equation cases, we improve on the original results
of~\cite{berndtson} by giving a more complete and simpler triangular
form.

In Section~\ref{sec:appl} we formulate our main results as a concise
theorem. Then we state several immediate qualitative corollaries, which
indicate important applications of the explicit knowledge of our
decoupling formulas. Finally, these applications and other potential
further developments are discussed in Section~\ref{sec:discussion}.

\section{Results and applications} \label{sec:appl}

Here we concisely state our main result (Theorem~\ref{thm:rw-equiv}) and
then discuss several important applications in the form of corollaries
that follow straightforwardly from the main theorem. The commentary
given along the way can be seen as a guide to reading the rest of the
paper. There are of course further potential applications that require
less straightforward consequences of Theorem~\ref{thm:rw-equiv}. They
will be pursued in future works.

We must start by introducing the minimum of the notations and
definitions needed to state our results. We may be informal here, but
will give pointers to the formal details in the body of the paper.

The central objects that enter the hypothesis of our main theorem are
specific linear ordinary second order differential operators
(equivalently, differential equations). They have matrix coefficients
(of given size) whose entries are rational functions depending on
parameters, which are $M$ (\emph{mass}), $l$ (\emph{angular momentum
quantum number}) and $\omega$ (\emph{frequency}). They are $\sqzer$
($1\times 1$), $\sqone_o$ ($1\times 1$), $\sqone_e$ ($3\times 3$),
$\sqtwo_o$ ($3\times 3$) and $\sqtwo_e$ ($7\times 7$), where notation
reflects that they are related to the mode separation of the tensor wave
operator $\square = \nabla_\mu \nabla^\nu$ on the static spherically
symmetric Schwarzschild black hole (full details in
Section~\ref{sec:schw}) for different \emph{spins} ($s=0,1,2$) and
belong either to the \emph{even} or \emph{odd} sectors under antipodal
reflection of the spherical orbits of symmetry. Each operator is defined
on the interval $r\in (2M,\oo)$, with a regular singularity at $r=2M$
and an irregular one at $r=\oo$. The dependence on all the parameters is
polynomial, with $\omega$ considered the \emph{spectral parameter}.
Since the explicit form of some of these operators is rather complicated, we
write them only in the sections where they are derived, namely
Sections~\ref{sec:sw},~\ref{sec:vw} and~\ref{sec:lich}. However, each of
them has the following generic form (see Section~\ref{sec:eq-mor} for
our conventions regarding differential operators):
\begin{equation*}
 E = \del_r P(r) \del_r + Q(r) + {\rm i}\omega A(r) + \omega^2 W(r),
\end{equation*}
where $P(r)$, $Q(r)$, ${\rm i}A(r)$ and $W(r)$ are self-adjoint matrices
(supposing that $M$ and $\ell$ have real values) of appropriate size,
meaning that $E^* = E$ is \emph{formally self-adjoint}
(Section~\ref{sec:adjoint}). As~dis\-cussed in the Introduction, the
matrix $W(r)$ is of indefinite signature (except for the $1\times 1$
cases), meaning that it cannot be used to define a Hilbert space (only a Krein space) on
which~$E$ could define a self-adjoint unbounded operator. Even if that
were possible, the dependence of~$E$ on both $\omega$ and $\omega^2$
implies that we would have needed to consider it as a \emph{quadratic
eigenvalue problem}, where verifying the reality of the
$\omega$-spectrum would have already been non-trivial.

In our main result, each of these complicated differential operators
will be transformed to simplified upper diagonal form, which involves
the family of \emph{spin-$s$ Regge--Wheeler} operators
(Section~\ref{sec:rw})
\begin{gather}
 \D_s \phi := \del_r f \del_r \phi
 - \frac{1}{r^2} \big[\B_l + \big(1-s^2\big)f_1\big] \phi + \frac{\omega^2}{f} \phi,
\end{gather}
where for convenience
\begin{gather}
 f(r) := 1 - \frac{2M}{r}, \qquad
 f_1(r) := r \del_r f(r) = \frac{2M}{r}.
\end{gather}
The following convenient combinations of the parameters often accompany
the Regge--Wheeler operators:
\begin{equation}
 \B_l := l(l+1), \qquad
 \A_l := (l-1)l(l+1)(l+2) = \B_l (\B_l-2), \qquad
 \alpha := (12 M \omega)^2 + \A_l^2.
\end{equation}
When $s=2$, the alternative \emph{Zerilli} operator is often used
\begin{equation}
 \D_2^+ \phi := \del_r f \del_r \phi
 - \frac{\A_l + 3(\B_l-2)f_1 (1+3f_1) + 9 f_1^3}
 {r^2(\B_l-2+3f_1)^2}
 + \frac{\omega^2}{f} \phi.
\end{equation}
In Section~\ref{sec:rwz} we study in detail the equivalence of $\D_2$
and $\D_2^+$, and explain why we prefer not to use the Zerilli operator.
The Regge--Wheeler operators are in fact of standard Sturm--Liouville
type~\cite{dss}.

\subsection{Main result}

To state Theorem~\ref{thm:rw-equiv} below, besides the notations we have
just introduced above, one also needs the notion of \emph{equivalence}
between two differential equations. Informally, for our purposes, two
differential equations are equivalent when there exist differential
operators that map solutions to solutions, in both directions, and are
moreover mutually inverse on the solutions. More precisely, our notion
of equivalence involves the existence of a number of auxiliary
differential operators satisfying some compositional identities, which
are conveniently summarized in an \emph{equivalence diagram} such
as~\eqref{eq:E-equiv} below. The full details of the definition and the
rationale for it are given in Section~\ref{sec:eq-mor}.

To prove that such an equivalence exists, it is of course sufficient to
write down all the differential operators and verify the required
identities between them. That is precisely our proof strategy, with full
details presented in Section~\ref{sec:schw}, with
Sections~\ref{sec:formal} and~\ref{sec:rw} building up preliminary
material for it.

The main content of the theorem is that each of our differential
equations of interest is equivalent to a much simpler upper triangular
form. To appreciate the magnitude of the simplification, it is
sufficient to compare the upper triangular forms
in~\eqref{eq:E-uptriang} with the original forms of $\sqzer$
\eqref{eq:sw-coord}, $\sqone_o$ \eqref{eq:vwo-coord}, $\sqone_e$
\eqref{eq:vwe-coord}, $\sqtwo_o$ \eqref{eq:radial-od} and $\sqtwo_e$
\eqref{eq:radial-ev}.

\begin{thm} \label{thm:rw-equiv}
Let $E$ be one of the $\sqzer$, $\sqone_o$, $\sqone_e$, $\sqtwo_o$ or
$\sqtwo_e$ operators, representing the systems of radial mode equations
for the scalar $($Section~$\ref{sec:sw})$, vector $($Section~$\ref{sec:vw})$ or
Lichnerowicz $($Section~$\ref{sec:lich})$ wave operators on the
Schwarzschild black hole of mass $M$, mode separated as described in
Section~$\ref{sec:schw}$, in either the odd or even sector. Then the
following is true:
\begin{enumerate}\itemsep=0pt
\renewcommand{\theenumi}{$\roman{enumi}$}
\renewcommand{\labelenumi}{$(\theenumi)$}
\item \label{subthm:Erat}
 With respect to the Schwarzschild radial coordinate $r$, $E$ is a
 rational ODE, with singular points only at $r=0, 2M, \oo$. $E$ also
 depends polynomially on the frequency $\omega$ and angular momentum
 $\B_l=l(l+1)$ spectral parameters, as well as the mass parameter $M$.
\item \label{subthm:Eequiv}
 There exists a rational ODE system $\tilde{E}$ in upper triangular
 form and an equivalence diagram $($Section~$\ref{sec:formal})$
 \begin{equation} \label{eq:E-equiv}
 \begin{tikzcd}[column sep=huge,row sep=huge]
 \bullet
 \ar[swap]{d}[description]{E}
 \ar[shift left]{r}{k}
 \&
 \bullet
 \ar[shift left]{l}{\bar{k}}
 \ar{d}[description]{\tilde{E}}
 \\
 \bullet
 \ar[shift left]{r}{k'}
 \ar[dashed,bend left=40]{u}{h}
 \&
 \bullet
 \ar[shift left]{l}{\bar{k}'}
 \ar[dashed,bend right=40]{u}[swap]{\tilde{h}}
 \end{tikzcd},
 \end{equation}
 such that the diagonal of $\tilde{E}$ consists of Regge--Wheeler
 operators $\D_s$ $($up to constant, $\omega$- and $\B_l$-dependent
 factors$)$, while the equivalence maps are rational differential
 operators of order at most $1$ and poles only at $r=0, 2M$.
\item \label{subthm:equiv-poles}
 The equivalence $\big(k,k',\bar{k},\bar{k}'\big)$ and homotopy $\big(h,\bar{h}\big)$
 operators in~\eqref{eq:E-equiv} depend rationally on~$\omega$ and~$\B_l$, with poles only at $\omega=0$ and $\B_l=0,2$, respectively,
 with one exception. The exception is the case $E=\sqtwo_e$, where some
 of these operators have additional poles at the algebraically special
 frequencies $\omega = \pm {\rm i}\frac{\A_l}{12M}$. In that case, using the
 Chandrasekhar transformation $($Section~$\ref{sec:rwz})$ to replace the
 Regge--Wheeler operator $\D_2$ by the Zerilli ope\-ra\-tor~$\D_2^+$ in~$\tilde{E}$ removes the poles at the algebraically special frequencies
 at the price of introducing factors of $(\B_l-2+3f_1)$ in some of the
 denominators in the equivalence diagram.
\item \label{subthm:libM0}
 The equivalence diagram~\eqref{eq:E-equiv} has finite limit as $M\to 0$.
 $\tilde{E}$ becomes diagonal in that limit.
\item \label{subthm:Eoffdiag}
 The non-vanishing off-diagonal elements of $\tilde{E}$ can all be put
 into the form $\frac{\Delta_0(r)}{f} + \Delta_1(r) \del_r$, with
 $\Delta_0(r)$ and $\Delta_1(r)$ uniformly bounded on $r\in (2M,\oo)$.
\item \label{subthm:sym}
 The rows and columns of $\tilde{E}$ can be permuted such that it
 remains upper triangular and block diagonal with respect to blocks of
 equal $s$ of the Regge--Wheeler operators $\D_s$. The most general
 rational automorphism of $\tilde{E}$ takes the form $\tilde{E} A = A
 \tilde{E}$, where $A$ is a constant matrix that is block diagonal with
 respect to equal spin blocks and is upper triangular within each such
 block, with the single exception of the $s=0$ block when $E =
 \sqtwo_e$.
\item \label{subthm:selfadj}
 There exists a constant matrix $\Sigma$, with $\Sigma^2 = \id$ and
 $\Sigma^* = \Sigma$ $($Section~$\ref{sec:adjoint})$, such that $\Sigma
 \tilde{E}$ is formally self-adjoint, $\big(\Sigma \tilde{E}\big)^* = \Sigma
 \tilde{E}$.
\end{enumerate}
In each of these cases, the upper triangular Regge--Wheeler system has
the following form:
 \noeqref{eq:E-uptriang0}%
 \noeqref{eq:E-uptriang1o}%
 \noeqref{eq:E-uptriang1e}%
 \noeqref{eq:E-uptriang2o}%
 \noeqref{eq:E-uptriang2e}%
\begin{subequations} \label{eq:E-uptriang}
\begin{enumerate}\itemsep=0pt
\renewcommand{\theenumi}{\alph{enumi}}
\renewcommand{\labelenumi}{$(\theenumi)$}
\item when $E = \sqzer$, then
 \begin{equation} \label{eq:E-uptriang0}
 \tilde{E} = \frac{1}{\omega^2} \D_0,
 \end{equation}
\item when $E = \sqone_o$, then
 \begin{equation} \label{eq:E-uptriang1o}
 \tilde{E} = \frac{\B_l}{\omega^2} \D_1,
 \end{equation}
\item when $E = \sqone_e$, then
 \begin{equation} \label{eq:E-uptriang1e}
 \tilde{E} = {\frac{1}{\omega^2}} \begin{bmatrix}
 \D_0 & 0 & -\tfrac{f_1}{r^2} \left(\B_l+\tfrac{1}{2}f_1\right) \\
 0 & \B_l \D_1 & 0 \\
 0 & 0 & \D_0
 \end{bmatrix}\!,
 \end{equation}
\item when $E = \sqtwo_o$, then
 \begin{equation} \label{eq:E-uptriang2o}
 \tilde{E} = {-\frac{2}{\omega^2}} \begin{bmatrix}
 \B_l \D_1 & 0 & \tfrac{1}{3}\B_l^2 \tfrac{f_1}{r^2} \\
 0 & \A_l \D_2 & 0 \\
 0 & 0 & \B_l \D_1
 \end{bmatrix}\!,
 \end{equation}
\item when $E = \sqtwo_e$, then
 \begin{gather} \label{eq:E-uptriang2e}
 \tilde{E} = {-\frac{2}{\omega^2}} \begin{bmatrix}[c@{~}c@{\,}c@{\!\!}c@{\!\!}c@{\!}c@{\!}c]
 \D_0 & 0 & -\tfrac{f_1}{r^2}\left(\B_l \!+\! \tfrac{1}{2}f_1\right) &
 0 & \tfrac{f_1}{r^2}\left(\B_l \!+\! \tfrac{1}{2}f_1\right) & 0 &
 \tfrac{f_1^2}{8r^2} (7\B_l \!+\! 2) \\
 0 & \B_l\D_1 & 0 & 0 & 0 & -\tfrac{5}{3} \B_l^2 \tfrac{f_1}{r^2} & 0 \\
 0 & 0 & \D_0 & 0 & 0 & 0 & \tfrac{f_1}{r^2}\left(\B_l \!+\! \tfrac{1}{2}f_1\right) \\
 0 & 0 & 0 & \alpha\A_l \D_2 & 0 & 0 & 0 \\
 0 & 0 & 0 & 0 & \D_0 & 0 & -\tfrac{f_1}{r^2}\left(\B_l \!+\! \tfrac{1}{2}f_1\right) \\
 0 & 0 & 0 & 0 & 0 & \B_l\D_1 & 0 \\
 0 & 0 & 0 & 0 & 0 & 0 & \D_0
 \end{bmatrix}\!.
 \end{gather}
\end{enumerate}
\end{subequations}
\end{thm}

\begin{proof}
Section~\ref{sec:schw} gives a detailed presentation of the explicit
formulas for the operators in the equivalence diagram~\eqref{eq:E-equiv},
thus proving~($ii$), for each case being considered:
$\eqref{eq:sw-equiv}{\implies}(a)$,
$\eqref{eq:vw1o-equiv}{\implies}(b)$,
$\eqref{eq:vw1e-equiv}{\implies}(c)$,
$\eqref{eq:lich2o-equiv}{\implies}(d)$,
$\eqref{eq:lich2e-equiv}{\implies}(e)$. The rest of the theorem is
simply a~summary, collected here in a way convenient for future
reference, of properties of these ope\-ra\-tors that can be gleaned from
their explicit formulas.
\end{proof}

\subsection{Spectral problem}

One might notice that the explicit form~\eqref{eq:E-uptriang} of the
triangular Regge--Wheeler systems in Theorem~\ref{thm:rw-equiv} gives
much more detailed information about the off-diagonal elements than
point (\ref{subthm:Eoffdiag}) of Theorem~\ref{thm:rw-equiv}. But the
reason we have phrased it like that is to draw attention to the following
almost immediate
\begin{cor} \label{cor:relbound}
Let $\tilde{E}$ be one of $\sqzer$, $\sqone_o$, $\sqone_e$, $\sqtwo_o$
and $\sqtwo_e$ from Theorem~$\ref{thm:rw-equiv}$. Then, sup\-po\-sing that~$\tilde{E}$ is an $n\times n$ operator matrix, the bounded inverse
$\tilde{E}^{-1}\colon L^2(2M,\oo; f\,\d{r})^{\oplus n} \to L^2\big(2M,\oo;
\frac{\d r}{f}\big)^{\oplus n}$ exists for any value of $\omega$ for which
the bounded inverse $\D_s^{-1} \colon L^2(2M,\oo; f\,\d{r}) \allowbreak\to
L^2\big(2M,\oo; \frac{\d r}{f}\big)$ also exists.
\end{cor}
\begin{proof}
The corollary follows from a direct generalization to larger operator
matrices of Remark~\ref{rmk:relbound}, about the relative boundedness of
off-diagonal matrix elements with respect to the Regge--Wheeler operators
on the diagonal, which here follows from
Theorem~\ref{thm:rw-equiv}(\ref{subthm:Eoffdiag}).
\end{proof}

As point out in Remark~\ref{rmk:relbound}, the novel results of
Section~\ref{sec:rw} concerning the \emph{cokernel} of the differential
operator in~\eqref{eq:rw-decoupling} are crucial for guaranteeing the
possibility of choosing the off-diagonal elements to be relatively
bounded while simplifying an upper triangular Regge--Wheeler system. In
practice, we have found that at the initial stages of this
simplification, which is part of the general strategy of
Section~\ref{sec:formal}, the off-diagonal elements were generally not
relatively bounded. Using only the results from~\cite{kh-rwtriang} about
only the \emph{solutions} of~\eqref{eq:rw-decoupling}, or other off-the-shelf
tools from the literature, it would have been impossible to decide a
priori whether those off-diagonal elements that cannot be chosen to be
zero could be chosen to be relatively bounded.

From Corollary~\ref{cor:relbound} one can essentially conclude that all
the properties of the resolvent $\tilde{E}^{-1}$ (like its dependence on
$\omega$), and by equivalence hence also of $E^{-1}$, can be deduced
from the properties of the resolvent $\D_s^{-1}$, which has been very
well studied. In particular, we can conclude that the $\omega$-spectrum
of each $E$ in Theorem~\ref{thm:rw-equiv} is purely real, as it is for
$\D_s$~\cite{dss}, once the relevant function spaces are chosen to
respect the equivalences that we have constructed. This information is
an important starting point for the integral representation of the
solutions of the corresponding Cauchy problem and of the large time
asymptotics of such solutions~\cite{dss}. It is also crucial for the
explicit construction of propagators in quantum field
theory~\cite{bdm, bfr-qg, fewster}.

We now make the above reasoning slightly more precise. When $E$ is one
of $\sqzer$, $\sqone_o$, $\sqone_e$, $\sqtwo_o$ or $\sqtwo_e$, both it
and its decoupled triangular form $\tilde{E}$ are rational ODEs with a
regular singular point at $r=2M$ and an irregular singular point at
$r=\oo$. Hence, for each system, and at each singular point, using the
methods of asymptotic analysis of ODEs~\cite{wasow}, each solution can
be uniquely identified by its asymptotics. Being bijective on solutions,
the equivalence mor\-phisms~$k$,~$\bar{k}$ of Theorem~\ref{thm:rw-equiv}
must hence be bijective on these asymptotics. Moreover, since these
morphisms are differential operators, they can be applied to the
asymptotics directly, without knowledge of the full solution. This means
that if we know the \emph{connection coefficients} between the singular points
at $r=2M$ and $r=\oo$ (essentially the \emph{transmission} and
\emph{reflection} coefficients for incoming and outgoing waves) at
frequency $\omega$ for $\tilde{E}$, the equivalence morphisms transfer
them to the corresponding connection coefficients for $E$. This argument
of course works only when the equivalence morphisms are well-defined,
which excludes $l=0,1$ in some cases and $\omega = 0$, as well as the
algebraically special $\omega = \pm {\rm i} \frac{\A_l}{12M}$ in the case
$E=\sqtwo_e$. But these excluded cases can be studied individually. For
instance, for the algebraically special frequencies, it is sufficient to
replace the Regge--Wheeler operator $\D_2$ with the Zerilli operator
$\D_2^+$ in $\tilde{\sqtwo}_e$, as pointed out in
Theorem~\ref{thm:rw-equiv}(\ref{subthm:equiv-poles}). For $\omega=0$ and
$l=0,1$ one can get more information by studying the coefficients of the
Laurent expansion of the equivalence morphisms. Thus, we arrive at
\begin{cor} \label{cor:rw-equiv-asymp}
For $\omega \ne 0$ and $l\ge 2$, the equivalence morphisms $k$,
$\bar{k}$ from Theorem~$\ref{thm:rw-equiv}$ are bijective on solution
asymptotics of each $E$, $\tilde{E}$ pair near $r=2M$ and $r=\oo$.
Hence, after bijectively identifying their solutions, the connection
coefficients for $E$ and $\tilde{E}$ are equal.
\end{cor}

\subsection{Symmetries and potentials}

Each of Section~\ref{sec:sw}, \ref{sec:vw} and~\ref{sec:lich} starts
with separation of variables on the static spherically symmetric
Schwarzschild black hole background, which turns spacetime partial
differential operators into ordinary differential operators acting on
radial mode functions. The time derivative ${\rm i}\del_t$ is replaced by the
frequency parameter $\omega$, while the spherical Laplacian $D_A D^A$ is
replaced by the angular momentum parameter $-\B_l$ (shifted by a constant
depending on the tensor rank of its argument), among other
substitutions. Not every operator on radial mode functions comes from
such a separation of variables, but it is straightforward to work out
the rules for recognizing those operators that do. Thus, directly from
the statement of Theorem~\ref{thm:rw-equiv}, we can also draw
conclusions about spacetime partial differential operators.

The first conclusion is about symmetries. A \emph{symmetry} of a
differential equation $E[u]=0$ is a pair of differential operators $S$,
$S'$ such that $ES = S'E$ (Section~\ref{sec:eq-mor}). Such symmetries are
extremely important in understanding the structure of a differential
equation~\cite{olver-lie}.

The second conclusion is about a variational
formulation~\cite{olver-lie}. A formally self-adjoint differential
operator $E=E^*$ (Section~\ref{sec:adjoint}) is the Euler--Lagrange
equation of a variational principle $\int \langle u, E[u] \rangle \,
\d{x}$. Conversely, the Euler--Lagrange equation of a variational
principle is always formally self-adjoint. For linear differential
equations the two notions are essentially equivalent. Having a
variational formulation leads to the possibility of using Noether's
theorem to convert symmetries to conservation laws. Even more
interesting is when the same equation has more than one variational
principle. This is the case for a self-adjoint equation $E[u]=0$ when
there exists another self-adjoint equation $\tilde{E}[v]=0$, which are
equivalent in a non-trivial way ($E$ and~$\tilde{E}$ are not just
multiples of each other). Multiple variational formulations can lead to
interesting integrability properties and bi-Hamiltonian structures.

\begin{cor} \label{cor:rw-equiv-spacetime}\quad
\begin{enumerate}
\renewcommand{\theenumi}{\alph{enumi}}
\renewcommand{\labelenumi}{$(\theenumi)$}
\item It is easy to see that, reversing the mode separation performed in
Section~$\ref{sec:schw}$, each the triangular systems $\tilde{\sqone}_o$,
$\tilde{\sqone}_e$, $\tilde{\sqtwo}_o$ and $\tilde{\sqtwo}_e$
constitutes the radial mode equation of a triangularly coupled system of
scalar spacetime differential equations. The diagonal components of this
system consist of the spacetime version of Regge--Wheeler wave equations,
possibly multiplied by a polynomial in the spherical Laplacian $D_A D^A$
and ${\rm i}\del_t$ with constant coefficients.

\item By similar arguments, the morphisms $k$, $k'$ and $\bar{k}$, $\bar{k}'$
from~$\eqref{eq:E-equiv}$ can be promoted to spacetime differential
operators, possibly after being multiplied by a constant coefficient
polynomial in $\B_l$ and $\omega$ of sufficiently high degree. These
spacetime differential operators will effect a morphism between the
original harmonic gauge tensor wave equation and the triangular system
of Regge--Wheeler wave equations from $(a)$.

\item The systems of Regge--Wheeler wave equations from $(a)$ are
self-adjoint on spacetime after composition with the $\Sigma$ operator
from Theorem~$\ref{thm:rw-equiv}${\rm (\ref{subthm:selfadj})}. Hence they
possess a variational formulation.
\end{enumerate}
\end{cor}

\begin{rmk}
Note that Corollary~\ref{cor:rw-equiv-spacetime} cannot hold for the
Zerilli form of $\tilde{\sqtwo}_e$ mentioned in
Theorem~\ref{thm:rw-equiv}(\ref{subthm:equiv-poles}), because no
spacetime differential operator has a radial mode decomposition with
$\B_l$-dependent factors like $(\B_l-2+3f_1)$ in the denominator.
\end{rmk}

A \emph{Debye potential} is a differential operator that maps solutions
of a simple auxiliary PDE equation (typically a scalar wave equation
with a potential) to solutions of another PDE of interest (typically a
more complicated tensor wave equation), with the usual additional
requirement that the image solutions are non-trivial (typically not pure
gauge, in an appropriate sense). In the terminology of
Section~\ref{sec:formal}, a Debye potential a special kind of morphism
between PDEs. Debye potentials for Maxwell and linearized Einstein
equations, on both flat and some curved backgrounds (including black
hole backgrounds), have been known for some time, but have also gained
more attention recently~\cite{aab-debye, ab-adjoint, stewart-debye, wald-adjoint}. In the case of black hole backgrounds, the Debye
potentials for Maxwell and linearized Einstein equations typically give
solutions in some version of \emph{radiation gauge}~\cite{stewart-debye,
wald-adjoint}, which is adapted to the algebraically special nature of
these backgrounds, but may not be optimal for other purposes. The recent
work~\cite{johnson-wave} produced a Debye potential in a (non-local)
generalized harmonic gauge. Of course, a solution in one gauge can be
(at least locally) transformed into any other gauge, but it is by no
means obvious when such a transformation can be done by a differential
operator. With that in mind, we draw attention to the next corollary
(first implicitly obtained in~\cite{berndtson}):
\begin{cor} \label{cor:rw-equiv-debye}
On the exterior spacetime of a Schwarzschild black hole, there exist
spacetime Debye potentials generating solutions of $(a)$ Maxwell equations
in harmonic $($Lorenz$)$ gauge and $(b)$ linearized Einstein equations in
traceless harmonic $($de~Donder$)$ gauge, starting from solutions of
Regge--Wheeler wave equations of appropriate spin.
\end{cor}
\begin{proof}
From Theorem~\ref{thm:rw-equiv}, it is clear that the spin $s=1$
components decouple from the other components in $\tilde{\sqone}_o$ and
$\tilde{\sqone}_e$, so they satisfy independent Regge--Wheeler equations.
Hence, the corresponding columns of the $\bar{k}_{1o}$ and
$\bar{k}_{1e}$ operators, once promoted to spacetime differential
operators as in Corollary~\ref{cor:rw-equiv-spacetime}, are Debye
potentials for the vector wave equation $\sqone$. However, the
composition properties of $k_{1o}$, $k_{1e}$ with $T_{1o}$, $T_{2e}$ (as
computed in Sections~\ref{sec:vwod} and~\ref{sec:vwev}) imply that the
image of the Debye potential is annihilated by the spacetime operator
$T_1$. In other words, the Debye potentials actually generate solutions
of Maxwell equations in harmonic gauge. For linearized Einstein
equations, the discussion is completely analogous, starting with the
Debye potentials for the Lichnerowicz wave equation generated by
solutions of Regge--Wheeler equations of spin $s=2$. In that case, the
composition properties with $T_2$ and $\tr$ operators (as computed in
Sections~\ref{sec:lichod} and~\ref{sec:lichev}) imply that the image of
these Debye potentials is both harmonic and trace free.
\end{proof}

\section{Formal properties of differential equations and operators}\label{sec:formal}

\subsection{Equations and morphisms} \label{sec:eq-mor}

For the purposes of this work, a (\emph{partial}) \emph{differential operator},
say $E$, is always linear, with smooth coefficients, which we will write
as $E[u]$, where $u$ is a possibly vector valued function. Later on we
will even specialize to ordinary differential operators with rational
coefficients, but all abstract statements expressed in terms of operator
identities or commutative diagrams will be valid at the greater level of
generality. Compositions of differential operators will be written
$E_1\circ E_2$ or just $E_1 E_2$ when no confusion could arise.

An operator $E$ can always be thought of as a matrix of scalar
differential operators acting on the components of $u$. By a
\emph{scalar differential operator}, we mean an operator that takes
possibly vector valued functions into scalar valued functions
(corresponding to a matrix with a~single row). Scalars are real or
complex numbers (for the sake of generality we will generally allow
complex scalars). Differential operators could be of any order,
including order zero, which just corresponds to multiplication by some
matrix valued function. While the operators can be thought of as acting
on smooth (vector valued) functions, we will be mostly concerned with
composition identities among operators, and so we will not bother
specifying precisely the domain or codomain of each operator. This
information can always be deduced from the context (e.g.,\ the size of
the matrix). Also, for the abstract discussion below, it is also not
necessary to fix the number of independent variables, but in the rest of
this work we will be mostly concerned with applications to ordinary
differential operators (i.e.,\ acting on functions of a single
independent variable).

Quick example: the operator $E = r^{-2} \del_r r^2$ should be
interpreted as $E[u] = r^{-2} \big(\del_r \big(r^2 u\big)\big) = \del_r u + 2u/r$, so
alternatively it can be rewritten as $E = \del_r + 2/r$.

Below, we state some basic definitions and results concerning
differential equations and \emph{morphisms} between them, which essentially
correspond to differential operators that map solutions to solutions.
The presentation is logically self-contained and does not extend far
beyond what is needed in the rest of the paper. We generalize and
streamline slightly the presentation previously given
in~\cite[Section~2]{kh-vwtriang}, which can be consulted for a more
pedagogical exposition. For proper context, these ideas can be seen as
simple special cases of concepts coming from the more general frameworks of
$D$-modules~\cite[Section~10.5]{seiler}, the category of differential
equations~\cite[Section~VII.5]{vinogradov} or homological
algebra~\cite{weibel}. We will not delve into the precise connection
with these larger frameworks, but it is useful to mention that all
arrows/morphisms need to be reversed when our statements are interpreted
in the language of $D$-modules.

Central objects of our attention are \emph{differential equations}, like
$E[u] = 0$. Sometimes we will refer just to the operator $E$ as the
equation itself and vice versa. This should not lead to any confusion.
Technically, it becomes convenient to consider not just an equation $E^{(0)}$,
but also a set of \emph{Noether $($or compatibility$)$ identities} that
comes with it, namely an operator $E'$ such that $E' \circ E^{(0)} = 0$. One
can equally consider higher stage Noether identities, for example an
operator $E''$ such that $E'' \circ E' = 0$, etc. Such a sequence of
operators $E^{(0)}, E', E'', \ldots, E^{(n)}$, is also called a
\emph{complex} (a term from \emph{homological algebra}~\cite{weibel}),
but we might as well still refer to it as a differential equation,
albeit equipped with the extra data of Noether and higher stage Noether
identities. We get back to a notion of \emph{solutions} via the idea of
\emph{cohomology} $H^{(n)}(E) := \ker E^{(n)} / \im E^{(n-1)}$ of the
complex, which may be non-vanishing in any location and can be
interpreted as the space of solutions of $E^{(n)}[u] = 0$ modulo a gauge
symmetry generated by $u \sim u + E^{(n-1)}[v]$. The complex is said to
be \emph{exact} (in location $n$) if $H^{(n)}(E) = 0$. When $E^{(n-1)}
= 0$ or is absent, by definition $\im E^{(n-1)} = 0$, meaning that
$H^{(n)}(E) = \ker E^{(n)}$ coincides with the usual notion of
solutions. Note however that the precise spaces of solutions and
cohomologies depend on the choice of the function space on which our
differential operators act. The various transformations and
constructions involving complexes that we introduce below are designed
to preserve cohomologies, hence also solutions, once the function spaces
have been chosen, which is the main reason that we abstract from
solutions themselves and work only at the level of complexes of
differential operators.

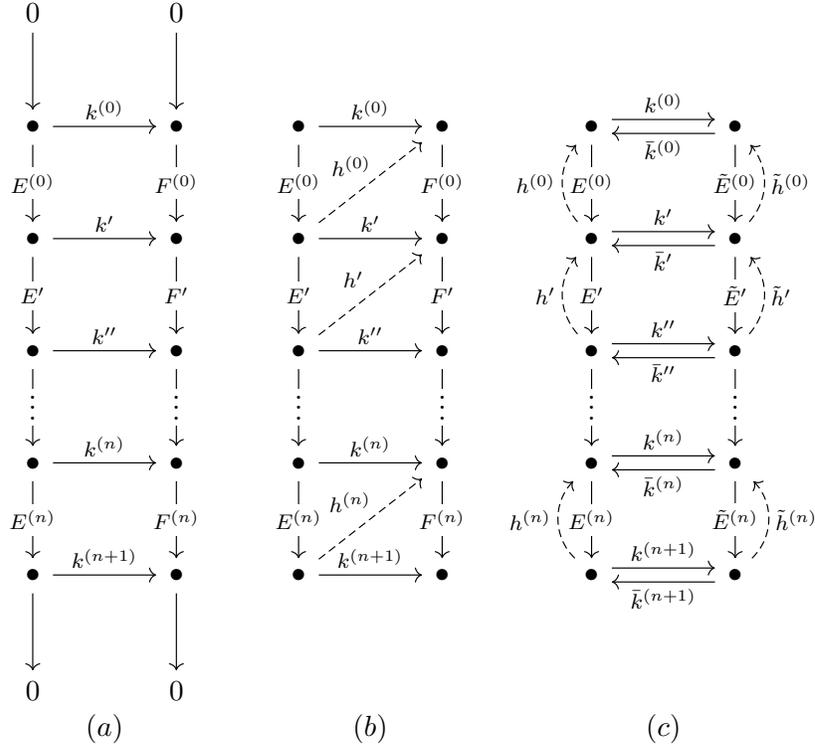
\begin{figure}
\begin{center}
\begin{tabular}{cccc}
 $\begin{tikzcd}[column sep=large,row sep=large]
 0 \ar{d} \& 0 \ar{d} \\
 \bullet \ar{d}[description]{E^{(0)}} \ar{r}{k^{(0)}} \&
 \bullet \ar{d}[description]{F^{(0)}} \\
 \bullet \ar{d}[description]{E'} \ar{r}{k'} \&
 \bullet \ar{d}[description]{F'} \\
 \bullet \ar{d}[description]{{\rvdots}} \ar{r}{k''} \&
 \bullet \ar{d}[description]{{\rvdots}} \\
 \bullet \ar{d}[description]{E^{(n)}} \ar{r}{k^{(n)}} \&
 \bullet \ar{d}[description]{F^{(n)}} \\
 \bullet \ar{d} \ar{r}{k^{(n+1)}} \&
 \bullet \ar{d} \\
 0 \& 0
 \end{tikzcd}$
 &&
 $\begin{tikzcd}[column sep=large,row sep=large]
 \bullet \ar{d}[description]{E^{(0)}} \ar{r}{k^{(0)}} \&
 \bullet \ar{d}[description]{F^{(0)}} \\
 \bullet \ar{d}[description]{E'} \ar{r}{k'} \ar[dashed]{ru}{h^{(0)}\!\!} \&
 \bullet \ar{d}[description]{F'} \\
 \bullet \ar{d}[description]{{\rvdots}} \ar{r}{k''} \ar[dashed]{ru}{h'} \&
 \bullet \ar{d}[description]{{\rvdots}} \\
 \bullet \ar{d}[description]{E^{(n)}} \ar{r}{k^{(n)}} \&
 \bullet \ar{d}[description]{F^{(n)}} \\
 \bullet \ar{r}{k^{(n+1)}} \ar[dashed]{ru}{h^{(n)}\!\!} \&
 \bullet
 \end{tikzcd}$
 &
 $\begin{tikzcd}[column sep=large,row sep=large]
 \bullet
 \ar[swap]{d}[description]{E^{(0)}}
 \ar[shift left]{r}{k^{(0)}}
 \&
 \bullet
 \ar[shift left]{l}{\bar{k}^{(0)}}
 \ar{d}[description]{\tilde{E}^{(0)}}
 \\
 \bullet
 \ar[shift left]{r}{k'}
 \ar[dashed,bend left=40]{u}{h^{(0)}}
 \ar{d}[description]{E'}
 \&
 \bullet
 \ar[shift left]{l}{\bar{k}'}
 \ar[dashed,bend right=40]{u}[swap]{\tilde{h}^{(0)}}
 \ar{d}[description]{\tilde{E}'}
 \\
 \bullet
 \ar[shift left]{r}{k''}
 \ar[dashed,bend left=40]{u}{h'}
 \ar{d}[description]{{\rvdots}}
 \&
 \bullet
 \ar[shift left]{l}{\bar{k}''}
 \ar[dashed,bend right=40]{u}[swap]{\tilde{h}'}
 \ar{d}[description]{\rvdots}
 \\
 \bullet
 \ar[shift left]{r}{k^{(n)}}
 \ar{d}[description]{E^{(n)}}
 \&
 \bullet
 \ar[shift left]{l}{\bar{k}^{(n)}}
 \ar{d}[description]{\tilde{E}^{(n)}}
 \\
 \bullet
 \ar[shift left]{r}{k^{(n+1)}}
 \ar[dashed,bend left=45]{u}{h^{(n)}}
 \&
 \bullet
 \ar[shift left]{l}{\bar{k}^{(n+1)}}
 \ar[dashed,bend right=45]{u}[swap]{\tilde{h}^{(n)}}
 \end{tikzcd}$
 \\
 ($a$) && ($b$) & ($c$)
\end{tabular}
\end{center}
\caption{Diagrams of morphisms, homotopies and equivalences (see text).}
 \label{fig:morph}
\end{figure}

The next most important concept is that of a \emph{morphism} between two
differential equations. Referring again to terminology from homological
algebra, a morphism for us will be a \emph{cochain map} between the
complexes corresponding to the two equations. In more detail, given two
complexes $E^{(0)},E',\ldots$ and $F^{(0)},F',\ldots$, a cochain map consists of
differential operators $k^{(0)},k',\ldots$ with domains and codomains
prescribed by the diagram in Figure~\ref{fig:morph}$(a)$,
which also needs to be \emph{commutative} (any sequence of operators
starting and ending at the same point must compose to the same thing).
We have denoted by bullets the appropriate function spaces, which are
less important to us than the structure and properties of the
differential operators acting between them. When needed, e.g.,\ for the
purposes of defining the cohomology or for defining a cochain map
between complexes of different lengths, any complex can be freely
extended by zero maps in both directions. This explains the extra zero
maps appended and prepended to the complexes in~Figure~\ref{fig:morph}$(a)$.
But since this can be done implicitly, we
will not show such extensions and rather use more compact diagrams like
the rest of Figure~\ref{fig:morph}.
When a complex consists of a
single operator $E$ ($E'=E''=\cdots = 0$, and we drop the now
unnecessary superscript $^{(0)}$), the cohomology has the
obvious identification $H(E) = \ker E / \im 0 = \ker E$ and $H'(E) =
\ker 0 / \im E = \coker E$. Once again, this highlights the
interpretation of a complex as a generalization of a linear differential
equation and of its cohomology as generalizing the space of solutions.
It is not hard to check that the defining properties of a morphism
(cochain map) $k, k', \ldots$ imply that it descends to well-defined
maps in cohomology $k^{(n)} \colon H^{(n)}(E) \to H^{(n)}(F)$. In the
simplest example of a~single operator, this means that $k\colon \ker E
\to \ker F$ maps solutions of $E[u] = 0$ to solutions of $F[v] = 0$
(also $k'\colon \coker E \to \coker F$ in this case). An
\emph{automorphism} is a morphism from a given equation to itself, in
which case the operator $k$ mapping solutions of the equation to
solutions is also known as a \emph{symmetry operator}.

A morphism is \emph{induced by a homotopy} if there exists a
\emph{homotopy}, that is, a sequence of operators $h^{(0)},h',\ldots$ (again,
extended by zero operators in either direction, as necessary) fitting
into the morphism diagram in Figure~\ref{fig:morph}$(b)$,
such that each horizontal arrow satisfies the formula
\begin{equation}
 k^{(n)} = h^{(n)} \circ E^{(n)} + F^{(n-1)}\circ h^{(n-1)},
\end{equation}
which is indeed a morphism. It is a straightforward exercise to check
that a morphism induced by a homotopy always descends to the zero map in
cohomology. Two morphisms $k_1^{(0)}, k'_1, \ldots$ and $k_2^{(0)}, k'_2, \ldots$
are \emph{equivalent up to homotopy} when the difference $\big(k_2^{(0)}-k_1^{(0)}\big),
(k'_2-k'_1), \ldots$ is induced by a homotopy. Two equations are
\emph{equivalent up to homotopy} if there exists morphisms between them
that are \emph{mutually inverse up to homotopy} (their compositions are
equivalent to the identity morphisms up to homotopy). All the relevant
operators can be exhibited in an \emph{equivalence diagram}, like
in Figure~\ref{fig:morph}$(c)$,
where the solid arrows commute and the homotopy corrections enter the
identities
\begin{gather}
\bar{k}^{(n)} \circ k^{(n)} = \id
- h^{(n)}\circ E^{(n)} - E^{(n-1)} \circ h^{(n-1)}, \nonumber
\\
k^{(n)} \circ \bar{k}^{(n)} = \id - \tilde{h}^{(n)}\circ \tilde{E}^{(n)} - \tilde{E}^{(n-1)} \circ \tilde{h}^{(n-1)}.\label{eq:homotopy-iso}
\end{gather}
Equivalence up to homotopy, both for operators and equations, of course
defines an equivalence relation.

\begin{rmk}
Unwinding the definition of a morphism $k$ between two single operator
complexes corresponding to the equations $E[v] = 0$ and $F[w] = 0$, we
find that when $E[v] = 0$ and $w=k^{(0)}[v]$, then $F[w] = F\big[k^{(0)}[v]\big]
= k^{(1)}[E[v]] = 0$. Hence, as promised, in this case, a~morphism
induces a map from solutions to solutions, which is a familiar notion in
PDE theory. Similarly, a symmetry (or automorphism) in our sense induces
a map from solutions to solutions of a given equation. Extending our
earlier example by an inverse morphism $\bar{k}$ up to homotopy implies
a bijection between the solutions of $E[v] = 0$ and $F[w] = 0$, since
then $\bar{k}^{(0)}\big[k^{(0)}[v]\big] = v - h^{(0)}[E[v]] = v$ and
$k^{(0)}\big[\bar{k}^{(0)}[w]\big] = w - \tilde{h}^{(0)}[F[w]] = w$. But even
more than that, an equivalence allows us to transfer the solutions of
inhomogeneous problems between equations, namely $E[v] = u$ is solved by
$v = \bar{k}^{(0)}[w] + h^{(0)}[u]$ if $w$ solves $F[w] = k^{(1)}[u]$,
which can be checked by direct calculation.

Essentially, our notion of morphism is the same as the usual PDE notion
of a map between solution spaces induced by a differential operator, but
equipped with extra structure. For some equations that are not in some
sense pathological, the extra operators $k^{(n)}$ may be reconstructed
(perhaps up to homotopy equivalence) from the knowledge of a single
operator $k^{(0)}$. Then, the above two notions actually coincide.
See~\cite{kh-vwtriang} for a bit more discussion on this point.
Considering subtle cases when differences between the two notions
actually arise is beyond the scope of this work. We will restrict
ourselves to using the more structured notion of morphism, which is
happens to be most convenient for our work.
\end{rmk}

\subsection{Triangular decoupling strategy} \label{sec:triang-strategy}

In this section, we outline our general decoupling strategy. We
basically summarize the more pedagogical discussion
from~\cite{kh-vwtriang}. There are two small differences. We use
slightly different nota\-tion, more adapted to the later needs of
Section~\ref{sec:schw}. Also, we take advantage of having defined, in
the preceding section, morphisms for complexes of possibly more than one
differential operators. Part of the discussion in~\cite{kh-vwtriang} was
unnecessarily awkward because we insisted on using morphisms only with
single differential operators (namely, the homotopy corrections on the
right-hand side of~\eqref{eq:homotopy-iso} did not have the same uniform
structure because of the truncation).

The input is a differential equation $E[u] = 0$, together with the
following morphisms
\begin{equation} \label{eq:input-ops}
 \begin{tikzcd}[column sep=2cm,row sep=2cm]
 \bullet
 \ar{d}[description]{\D_D}
 \ar{r}{D}
 \&
 \bullet
 \ar{d}[description]{E}
 \\
 \bullet
 \ar{r}[swap]{D'}
 \&
 \bullet
 \end{tikzcd},
 \qquad
 \begin{tikzcd}[column sep=2cm,row sep=2cm]
 \bullet
 \ar{d}[description]{E}
 \ar{r}{T}
 \&
 \bullet
 \ar{d}[description]{\D_T}
 \\
 \bullet
 \ar{r}[swap]{T'}
 \&
 \bullet
 \end{tikzcd},
 \qquad
 \begin{tikzcd}[column sep=2cm,row sep=2cm]
 \bullet
 \ar{d}[description]{\begin{bmatrix} E \\ T \end{bmatrix}}
 \ar{r}{\Phi}
 \&
 \bullet
 \ar{d}[description]{\D_\Phi}
 \\
 \bullet
 \ar{r}[swap]{\begin{bmatrix} \Phi' & -\Delta_{\Phi T} \end{bmatrix}}
 \&
 \bullet
 \end{tikzcd}
\end{equation}
from/to ``simpler'' equations $\D_D$, $\D_T$, $\D_\Phi$.
The end result is an equation in
block upper triangular form with $\D_D$, $\D_\Phi$ and $\D_T$ on the
diagonal, where both the diagonal blocks and the off-diagonal ones have
also been simplified as much as possible, keeping upper triangular
structure, as shown in~\eqref{eq:E-decoupled}.
For our purposes, we may refer to the degrees of freedom
captured by the $D$ operator as \emph{pure gauge modes}, those by the
$\Phi$ operator as \emph{gauge invariant modes}, and those by the $T$
operator as \emph{constraint violating modes}. This terminology will be
justified in how the strategy will be applied in Section~\ref{sec:schw}
to Maxwell and linearized Einstein equations in harmonic gauge (see also
the detailed discussion in~\cite{kh-vwtriang}). At the abstract level,
we have simply introduced convenient names to the degrees of freedom
corresponding to the 3 sectors of the eventual block upper triangular
decoupling $\bar{E}$ in~\eqref{eq:E-decoupled}.

The first non-trivial step is to find an equivalence diagram,
illustrated in~\eqref{eq:res-gauge-equiv}, between~$\D_D$ (the \emph{residual gauge equation}) and the joint system $E[u] =
0$, $\Phi[u] = 0$ (\emph{vanishing of gauge invariant modes}), $T[u] =
0$ (\emph{gauge fixing condition}). That is, we must complement the
input differential operators already introduced in~\eqref{eq:input-ops}
by a bunch of new operators that are defined by their relation to the
following diagram being an equivalence up to homotopy:
\begin{subequations}
\begin{equation} \label{eq:res-gauge-equiv}
 \scalebox{0.9}{\begin{tikzcd}[column sep=4cm,row sep=3cm]
 \bullet
 \ar{d}[description]{\begin{bmatrix} E \\ \Phi \\ T \end{bmatrix}}
 \ar[shift left]{r}{\bar{D}}
 \&
 \bullet
 \ar{d}[description]{\D_D}
 \ar[shift left]{l}{D}
 \\
 \bullet
 \ar{d}[description]{\begin{bmatrix}
 -\Phi' & \D_\Phi & \Delta_{\Phi T} \\
 -T' & 0 & \D_T \end{bmatrix}\hspace{-3em}}
 \ar[shift right]{r}[swap]{\begin{bmatrix}
 \bar{D}' & -\Delta_{D \Phi} & -\Delta_{D T} \end{bmatrix}}
 \ar[dashed,bend left]{u}{\begin{bmatrix}
 h_E & \bar{\Phi} & \bar{T} \end{bmatrix}}
 \&
 \bullet
 \ar{d}[description]{0}
 \ar[shift right]{l}[swap]{\begin{bmatrix} D' \\ H_\Phi \\ H_T \end{bmatrix}}
 \ar[dashed,bend right]{u}[swap]{h_D}
 \\
 \bullet
 \ar[shift right]{r}[swap]{0}
 \ar[dashed,bend left=60]{u}{\begin{bmatrix}
 -\bar{\Phi}' & -\bar{T}' \\
 h_{\Phi \Phi} & h_{\Phi T} \\
 h_{T \Phi} & h_{T T} \end{bmatrix}}
 \&
 \bullet
 \ar[shift right]{l}[swap]{0}
 \end{tikzcd}}
\end{equation}
supplemented by two more operators $\bar{H}_\Phi$, $\bar{H}_T$ satisfying
the relation
\begin{equation} \label{eq:homotopy-compat}
\begin{aligned}
&\bar{D} \begin{bmatrix} h_E & \bar{\Phi} & \bar{T} \end{bmatrix}
 - h_D \begin{bmatrix}
 \bar{D}' & -\Delta_{D \Phi} & -\Delta_{D T} \end{bmatrix}
= \begin{bmatrix} \bar{H}_\Phi & \bar{H}_T \end{bmatrix}
 \begin{bmatrix}
 -\Phi' & \D_\Phi & \Delta_{\Phi T} \\
 -T' & 0 & \D_T \end{bmatrix}\!,
\\
&\begin{bmatrix} \bar{D}' & -\Delta_{D \Phi} & -\Delta_{D T} \end{bmatrix}
 \begin{bmatrix}
 -\bar{\Phi}' & -\bar{T}' \\
 h_{\Phi \Phi} & h_{\Phi T} \\
 h_{T \Phi} & h_{T T}
 \end{bmatrix}
= \D_D \begin{bmatrix} \bar{H}_\Phi & \bar{H}_T \end{bmatrix} \!.
\end{aligned}
\end{equation}
\end{subequations}
If such a diagram does not exist, than our strategy cannot proceed
directly. In this work, we will only consider cases where this step of
the strategy succeeds.
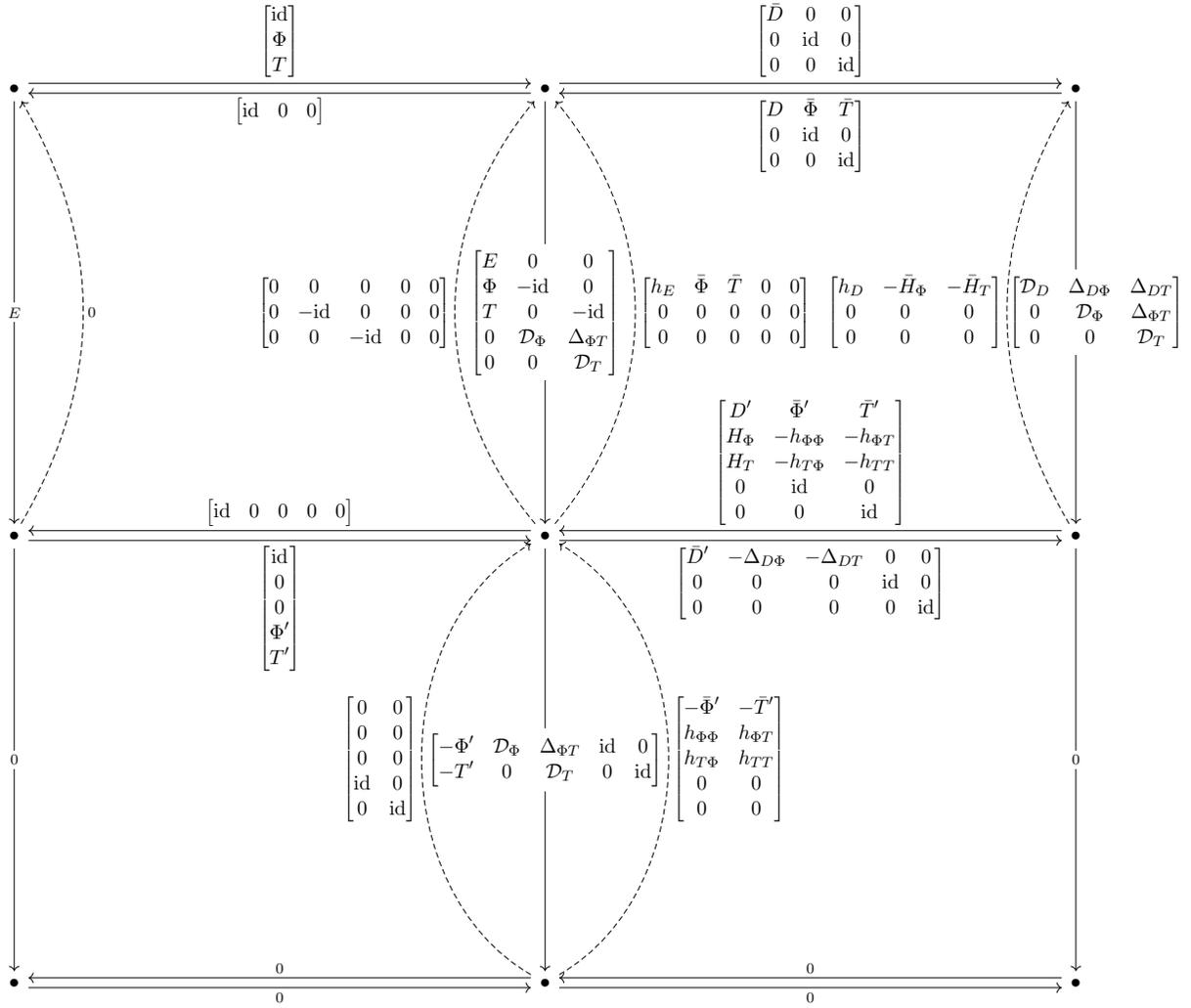
\begin{figure}[t!]
\centering
$\hspace*{-5mm}\scalebox{0.72}{%
\begin{tikzcd}[column sep=9.5cm,row sep=8cm]
 \bullet
 \ar{d}[description]{E}
 \ar[shift left]{r}{\begin{bmatrix} \id \\ \Phi \\ T \end{bmatrix}}
 \&
 \bullet
 \ar{d}[description]{\begin{bmatrix}
 E & 0 & 0 \\
 \Phi & -\id & 0 \\
 T & 0 & -\id \\
 0 & \D_\Phi & \Delta_{\Phi T} \\
 0 & 0 & \D_T \end{bmatrix}}
 \ar[shift left]{l}{\begin{bmatrix} \id & 0 & 0 \end{bmatrix}}
 \ar[shift left]{r}{\begin{bmatrix}
 \bar{D} & 0 & 0 \\
 0 & \id & 0 \\
 0 & 0 & \id \end{bmatrix}}
 \&
 \bullet
 \ar{d}[description]{\begin{bmatrix}
 \D_D & \Delta_{D \Phi} & \Delta_{D T} \\
 0 & \D_\Phi & \Delta_{\Phi T} \\
 0 & 0 & \D_T \end{bmatrix}\hspace{-2em}}
 \ar[shift left]{l}{\begin{bmatrix}
 D & \bar{\Phi} & \bar{T} \\
 0 & \id & 0 \\
 0 & 0 & \id \end{bmatrix}}
 \\
 \bullet
 \ar{d}[description]{0}
 \ar[shift right]{r}[swap]{\begin{bmatrix} \id \\ 0 \\ 0 \\ \Phi' \\ T' \end{bmatrix}}
 \ar[dashed,bend right]{u}[swap]{0}
 \&
 \bullet
 \ar[shift right]{l}[swap]{\begin{bmatrix} \id & 0 & 0 & 0 & 0 \end{bmatrix}}
 \ar{d}[description]{\begin{bmatrix}
 -\Phi' & \D_\Phi & \Delta_{\Phi T} & \id & 0 \\
 -T' & 0 & \D_T & 0 & \id \end{bmatrix}}
 \ar[shift right]{r}[swap]{\begin{bmatrix}
 \bar{D}' & -\Delta_{D \Phi} & -\Delta_{D T} & 0 & 0 \\
 0 & 0 & 0 & \id & 0 \\
 0 & 0 & 0 & 0 & \id \end{bmatrix}}
 \ar[dashed,bend left=40]{u}{\begin{bmatrix}
 0 & 0 & 0 & 0 & 0 \\
 0 & -\id & 0 & 0 & 0 \\
 0 & 0 & -\id & 0 & 0 \end{bmatrix}}
 \ar[dashed,bend right=40]{u}[swap]{\begin{bmatrix}
 h_E & \bar{\Phi} & \bar{T} & 0 & 0 \\
 0 & 0 & 0 & 0 & 0 \\
 0 & 0 & 0 & 0 & 0
 \end{bmatrix}}
 \&
 \bullet
 \ar{d}[description]{0}
 \ar[shift right]{l}[swap]{\begin{bmatrix}
 D' & \bar{\Phi}' & \bar{T}' \\
 H_\Phi & -h_{\Phi \Phi} & -h_{\Phi T} \\
 H_T & -h_{T \Phi} & -h_{T T} \\
 0 & \id & 0 \\
 0 & 0 & \id \end{bmatrix}}
 \ar[dashed,bend left]{u}{\begin{bmatrix}
 h_D & -\bar{H}_\Phi & -\bar{H}_T \\
 0 & 0 & 0 \\
 0 & 0 & 0 \end{bmatrix}}
 \\
 \bullet
 \ar[shift right]{r}[swap]{0}
 \&
 \bullet
 \ar[shift right]{r}[swap]{0}
 \ar[shift right]{l}[swap]{0}
 \ar[dashed,bend left=60]{u}{\begin{bmatrix}
 0 & 0 \\
 0 & 0 \\
 0 & 0 \\
 \id & 0 \\
 0 & \id \end{bmatrix}}
 \ar[dashed,bend right=60]{u}[swap]{\begin{bmatrix}
 -\bar{\Phi}' & -\bar{T}' \\
 h_{\Phi \Phi} & h_{\Phi T} \\
 h_{T \Phi} & h_{T T} \\
 0 & 0 \\
 0 & 0 \end{bmatrix}}
 \&
 \bullet
 \ar[shift right]{l}[swap]{0}
 \end{tikzcd}}
$
\caption{Concatenation of diagrams of morphisms showing the equivalence
 between $E[u] = 0$ and a~system of equations in block upper triangular
 form.\label{fig:big-equiv}}
\end{figure}

\begin{rmk}
There are a few potential points of failure, which should all be
checked.

First, the gauge-invariant and constraint violating modes should be
compatible and independent from the pure gauge modes, namely
$\left[\begin{smallmatrix} \Phi \\ T \end{smallmatrix}\right] \circ D
\sim 0$ up to homotopy. This condition is captured by the following
subset of identities from~\eqref{eq:res-gauge-equiv}:
\begin{equation}
 \begin{bmatrix} \Phi \\ T \end{bmatrix} D
 = \begin{bmatrix} H_\Phi \\ H_T \end{bmatrix} \D_D,
 \qquad
 \begin{bmatrix} \Phi' \\ T' \end{bmatrix} D'
 = \begin{bmatrix} \D_\Phi & \Delta_{\Phi T} \\ 0 & \D_T \end{bmatrix}
 \begin{bmatrix} H_\Phi \\ H_T \end{bmatrix} \D_D.
\end{equation}
Strictly speaking, to justify the adjective \emph{gauge-invariant}, we
should have $\Phi \circ D = 0$ ($H_T=0$), but here we are allowing
$\Phi$ to mix with $E$ and $T$, that is, terms that vanish on-shell or
by the gauge-fixing condition. In practive, within this work we will
always be able to make an initial choice of $\Phi$ such that $H_T=0$,
and then update it to fit our needs while maintaining the above
relations.

Note that we have chosen to represent the $E$-$\Phi$-$T$ equation in
diagram~\eqref{eq:res-gauge-equiv} by a complex of length 2, with the
second operator representing Noether identities for the first one. If we
do not choose this complex correctly, then the construction of the
equivalence diagram may be obstructed, for instance when some Noether
identity for the first operator is not a consequence of the rows of the
second operator. When such obstructions are absent, the existence of
operators $\bar{H}_\Phi$, $\bar{H}_T$ in~\eqref{eq:homotopy-compat}
actually follows from the relations already exhibited
by~\eqref{eq:res-gauge-equiv}. In this work, we will only encounter
cases where both sides of the equivalence
diagram~\eqref{eq:res-gauge-equiv} are correctly chosen, but to save
space we simply postulate the factorizations~\eqref{eq:homotopy-compat}
separately.

Finally, another point of failure could be that the $T[u]=0$ and
$\Phi[u]=0$ conditions are simply not strong enough to reduce the number
of remaining degrees of freedom to equal that of the $\D_D$ equation.
Again, in this work we will only encounter cases where this step
succeeds.
\end{rmk}

\looseness=-1
Next, using only operators that have already been introduced above,
together with their defining properties, we can construct two more
equivalence diagrams that are concatenated and explicitly shown in
Figure~\ref{fig:big-equiv}.
Composing the morphisms in Figure~\ref{fig:big-equiv}, in both
directions, we obtain the much more compact equivalence with an upper
triangular form, which we denote by~$\tilde{E}$,
\begin{equation} \label{eq:E-decoupled}
\scalebox{0.87}{\begin{tikzcd}[column sep=4cm,row sep=3.5cm]
 \bullet
 \ar{d}[description]{E}
 \ar[shift left]{r}{\begin{bmatrix} \bar{D} \\ \Phi \\ T \end{bmatrix}}
 \&
 \bullet
 \ar{d}[description]{\hspace{-5em}\tilde{E} = \begin{bmatrix}
 \D_D & \Delta_{D \Phi} & \Delta_{D T} \\
 0 & \D_\Phi & \Delta_{\Phi T} \\
 0 & 0 & \D_T \end{bmatrix}}
 \ar[shift left]{l}{\begin{bmatrix} D & \bar{\Phi} & \bar{T} \end{bmatrix}}
 \\
 \bullet
 \ar[shift right]{r}[swap]{\begin{bmatrix} \bar{D}' \\ \Phi' \\ T' \end{bmatrix}}
 \ar[dashed,bend left]{u}{h_E}
 \&
 \bullet
 \ar[shift right]{l}[swap]{\begin{bmatrix} D' & \bar{\Phi}' & \bar{T}' \end{bmatrix}}
 \ar[dashed,bend right=50]{u}[swap]{\begin{bmatrix}
 h_D & -\bar{H}_\Phi & -\bar{H}_T \\
 -H_\Phi & h_{\Phi \Phi} & h_{\Phi T} \\
 -H_T & h_{T \Phi} & h_{T T} \end{bmatrix}.}
 \end{tikzcd}}
\end{equation}
Conversely, knowing the final equivalence
diagram~\eqref{eq:E-decoupled} of $E$ with the decoupled form
$\tilde{E}$, simply by placing individual operators in the appropriate
places, we can reproduce the equivalence~\eqref{eq:res-gauge-equiv} of
the $E$-$\Phi$-$T$ system with $\D_D$ and the
relations~\eqref{eq:homotopy-compat}

The triangular system $\tilde{E}$ may still be rather complicated. So
the strategy proceeds with two kinds of simplifications: simplifying the
diagonal blocks, and simplifying the off-diagonal blocks (cf.~\cite[Sections
2.2--3]{kh-vwtriang}). The first kind is summarized in the following
\begin{lem} \label{lem:diag-simp}
If there exists an equivalence diagram between equations $E_2$ and
$\bar{E}_2$, then it implies the following equivalence diagram for an upper
triangular system, where $E_2$ is the middle block:
\begin{gather}
\scalebox{1.2}{\begin{tikzcd}[column sep=huge,row sep=huge]
 \bullet
 \ar[swap]{d}[description]{E_2}
 \ar[shift left]{r}{k}
 \&
 \bullet
 \ar[shift left]{l}{\bar{k}}
 \ar{d}[description]{\bar{E}_2}
 \\
 \bullet
 \ar[shift left]{r}{k'}
 \ar[dashed,bend left=40]{u}{h}
 \&
 \bullet
 \ar[shift left]{l}{\bar{k}'}
 \ar[dashed,bend right=40]{u}[swap]{\bar{h}}
\end{tikzcd}}
 \\ \qquad
{}\implies\scalebox{.95}{\begin{tikzcd}[column sep=5cm,row sep=5cm]
 \bullet
 \ar[swap]{d}[description]{\begin{bmatrix}
 E_1 & \Delta_{12} & \Delta_{13} \\
 0 & E_2 & \Delta_{23} \\
 0 & 0 & E_3
 \end{bmatrix}}
 \ar[shift left]{r}{\begin{bmatrix}
 \id & 0 & 0 \\ 0 & k & 0 \\ 0 & 0 & \id
 \end{bmatrix}}
 \&
 \bullet
 \ar[shift left]{l}{\begin{bmatrix}
 \id & 0 & 0 \\ 0 & \bar{k} & -h\Delta_{23} \\ 0 & 0 & \id
 \end{bmatrix}}
 \ar{d}[description]{\begin{bmatrix}
 E_1 & \bar{\Delta}_{12} & \bar{\Delta}_{13} \\
 0 & \bar{E}_2 & \bar{\Delta}_{23} \\
 0 & 0 & E_3
 \end{bmatrix}}
 \\
 \bullet
 \ar[shift left]{r}{\begin{bmatrix}
 \id & -\Delta_{12} h & 0 \\ 0 & k' & 0 \\ 0 & 0 & \id
 \end{bmatrix}}
 \ar[dashed,bend left=60]{u}{\begin{bmatrix}
 0 & 0 & 0 \\ 0 & h & 0 \\ 0 & 0 & 0
 \end{bmatrix}}
 \&
 \bullet
 \ar[shift left]{l}{\begin{bmatrix}
 \id & 0 & 0 \\ 0 & \bar{k}' & 0 \\ 0 & 0 & \id
 \end{bmatrix}}
 \ar[dashed,bend right=60]{u}[swap]{\begin{bmatrix}
 0 & 0 & 0 \\ 0 & \bar{h} & 0 \\ 0 & 0 & 0
 \end{bmatrix},}
\end{tikzcd}}
\end{gather}
where $\bar{\Delta}_{12} = \Delta_{12} \bar{k}$, $\bar{\Delta}_{23} = k'
\Delta_{23}$ and $\bar{\Delta}_{13} = \Delta_{13} - \Delta_{12} h
\Delta_{23}$, provided the extra identity $k h = \bar{h} k'$ also holds
$($though it is only needed to verify the form of the barred homotopy
operator on the right$)$.
\end{lem}

The proof is by direct calculation. Though, we should note that when $E_2$
has no non-trivial Noether identities (for example, in the ODE case,
when it can be solved for all highest derivatives), then the extra
identity on the homotopy operators is automatic. Indeed, then
\begin{equation}
 \big(k h - \bar{h} k'\big) E_2
 = \big(k - h\bar{E}_2 k\big) - (k - k h E_2)
 = \big(k \bar{k}\big) k - k \big(\bar{k} k\big) = 0,
\end{equation}
and since $E_2$ has no non-trivial Noether identities, we must have $kh -
\bar{h} k' = 0$.

Note that Lemma~\ref{lem:diag-simp} remains valid when either of the
untransformed diagonal blocks $E_1$ or~$E_3$ is zero dimensional (a
$0\times 0$ matrix, or just absent, in other words).

The second kind of simplification can be achieved by applying the next
\begin{lem} \label{lem:offdiag-simp}
For any operators $\delta$ and $\delta'$ the following equivalence
diagram holds:
\begin{equation}
\scalebox{0.9}{%
\begin{tikzcd}[column sep=3.5cm,row sep=3.5cm]
 \bullet
 \ar[swap]{d}[description]{\begin{bmatrix}
 E_1 & \Delta \\
 0 & E_2
 \end{bmatrix}}
 \ar[shift left]{r}{\begin{bmatrix}
 \id & \delta \\ 0 & \id
 \end{bmatrix}}
 \&
 \bullet
 \ar[shift left]{l}{\begin{bmatrix}
 \id & -\delta \\ 0 & \id
 \end{bmatrix}}
 \ar{d}[description]{\begin{bmatrix}
 E_1 & \bar{\Delta} \\
 0 & E_2
 \end{bmatrix}}
 \\
 \bullet
 \ar[shift left]{r}{\begin{bmatrix}
 \id & \delta' \\ 0 & \id
 \end{bmatrix}}
 \ar[dashed,bend left=50]{u}{\begin{bmatrix}
 0 & 0 \\ 0 & 0
 \end{bmatrix} }
 \&
 \bullet
 \ar[shift left]{l}{\begin{bmatrix}
 \id & -\delta' \\ 0 & \id
 \end{bmatrix}}
 \ar[dashed,bend right=50]{u}[swap]{\begin{bmatrix}
 0 & 0 \\ 0 & 0
 \end{bmatrix},}
\end{tikzcd}
}
\end{equation}
where $\bar{\Delta} = \Delta - E_1\delta + \delta' E_2$.
\end{lem}
The proof is again by direct calculation. This Lemma implies that by a
clever choice of $\delta$ and $\delta'$ it might be possible to set
$\bar{\Delta} = 0$, thus decoupling the two diagonal blocks, or at least
set it some preferred canonical form. For instance, in the case when
$E_1$ and $E_2$ are ODEs that can be solved for their higher
derivatives, say respectively of orders $n_1$ and $n_2$, then by an
obvious choice of $\delta$ and $\delta'$ the off-diagonal term
$\bar{\Delta}$ can be reduced to order $\min(n_1,n_2)$.

Our decoupling strategy proceeds by applying the last two lemmas
recursively, repartitioning our triangular equation into blocks as
necessary, until no further simplification is possible or practical. In
the successful applications of this strategy in Section~\ref{sec:schw}
we stop when the system~$\tilde{E}$ has been reduced to fully upper
triangular form, where each operator on the diagonal is a \emph{spin-$s$
Regge--Wheeler} operator (Section~\ref{sec:rw}), while off-diagonal terms
are mostly zero, with those that are non-vanishing expressed in terms of
a small number of canonical representatives. We dedicate the following
section (Section~\ref{sec:rw}) to giving an effective method of deciding
when Lemma~\ref{lem:offdiag-simp} can be used to cancel (or at least
``maximally simplify'') specific off-diagonal terms in a triangular
system with Regge--Wheeler operators on the diagonal.

\subsection{Formal adjoints} \label{sec:adjoint}

Given a local sesquilinear scalar paring $\langle -, -\rangle =
\int\d{x}\, \langle -, -\rangle_x$, with $\langle -, -\rangle_x$ the
corresponding pointwise non-degenerate pairing, the \emph{formal
adjoint} $E^*$ of a differential operator $E$ is defined by the usual
relation $\langle E^*[v], u \rangle = \langle v, E[u] \rangle$ for
smooth compactly supported test functions~$u$ and~$v$. If an operator
satisfies $E^* = E$, then it is \emph{(formally) self-adjoint}. As
defined, formal adjoints always exist and are unique. In local
coordinates they may be computed by the standard rules $(E F)^* = F^*
E^*$, $\del_{x^i}^* = -\del_{x^i}$ and $f(x)^*$ for a multiplicative
matrix operator coinciding with its usual pointwise adjoint with respect
to $\langle -, -\rangle_x$.

On a manifold with a metric, for tensor fields $v$ and $u$, we
define the paring $\langle v, u\rangle$ to be the integral with respect
to the metric volume form of the metric contraction of $v^*$ (the
complex conjugate of $v$) and $u$. We use this convention in
Section~\ref{sec:schw}, both for the spacetime metric $\gf_{\mu\nu}$ on
$\M = \mathbb{R}^2 \times S^2$ and for the metric $g_{ab}$ on the
radio-temporal $\mathbb{R}^2$-factor.

In Sections~\ref{sec:rw} and~\ref{sec:schw}, we will also be dealing
with ordinary differential operators with respect to the Schwarzschild
radial variable $r$. In that context, we take the pairing $\langle v, u
\rangle = \int_{2M}^\oo v^* u \, \d{r}$, where $u$ is a column vector
valued function and $v^*$ is the conjugate transpose of a column vector
valued function.

When the operators depend on spectral parameters, like $\omega$ and $l$,
they are considered to be real, so that $\omega^* = \omega$ and $l^* =
l$. However, that does not prevent us from considering complex values of
these parameters once all the adjoints have been computed.

\section{Morphisms and extensions of Regge--Wheeler equations}\label{sec:rw}

In this section, we present the theory needed for simplifying the
off-diagonal elements of rational ODE systems in upper triangular form. We first
concisely summarize some necessary notions and results
from~\cite{kh-rwtriang} (where a more detailed and pedagogical treatment
can be found) and then extend them. In this section, we will work
exclusively with rational differential operators.

The following functions will appear in the Schwarzschild metric and its
derivatives:
\begin{equation}
 f(r) = 1 - \frac{2M}{r}, \qquad
 f_1(r) = r\del_r \frac{2M}{r}.
\end{equation}
Let also $l$, $s$ and $\omega$ be real (or more generally complex)
constants, with $\B_l = l(l+1)$ a convenient notation. Throughout this
work, we use only $\del_r$ to denote differentiation, while primes are
used only as decorations (to distinguish two objects from each
other, like to operators $k$ and $k'$).

Define the \emph{$($generalized$)$ spin-$s$ Regge--Wheeler operator} with
\emph{mass parameter} $M$, \emph{angular momentum quantum number} $l$
and \emph{frequency} $\omega$ by (following the notation conventions of
Section~\ref{sec:eq-mor})
\begin{equation} \label{eq:rw}
 \D_s \phi = \del_r f \del_r \phi
 - \frac{1}{r^2} \big[\B_l + \big(1-s^2\big)f_1\big] \phi + \frac{\omega^2}{f} \phi.
\end{equation}
We will usually assume that $l=0,1,2,\ldots$, that $s$ is a non-negative
integer (in fact, only $s=0,1,2$ will concern us), and that $\omega \ne
0$. Of course, $\D_s$ has rational coefficients in $r$ and all of these
parameters.

Consider the following equivalence of upper triangular Regge--Wheeler
systems:
\begin{equation} \label{eq:rw-sys}
\begin{tikzcd}[column sep=huge,row sep=huge]
 \bullet
 \ar[swap]{d}{\begin{bmatrix} \D_{s_0} & \Delta \\ 0 & \D_{s_1} \end{bmatrix}}
 \ar{r}{\begin{bmatrix} \id & \delta \\ 0 & \id \end{bmatrix}}
 \&
 \bullet
 \ar{d}{\begin{bmatrix} \D_{s_0} & \bar{\Delta} \\ 0 & \D_{s_1} \end{bmatrix}}
 \\
 \bullet
 \ar[swap]{r}{\begin{bmatrix} \id & \eps \\ 0 & \id \end{bmatrix}}
 \&
 \bullet
\end{tikzcd}
 \!, \qquad
\begin{tikzcd}[column sep=huge,row sep=huge]
 \bullet
 \ar[swap]{d}{\begin{bmatrix} \D_{s_0} & \bar{\Delta} \\ 0 & \D_{s_1} \end{bmatrix}}
 \ar{r}{\begin{bmatrix} \id & -\delta \\ 0 & \id \end{bmatrix}}
 \&
 \bullet
 \ar{d}{\begin{bmatrix} \D_{s_0} & \Delta \\ 0 & \D_{s_1} \end{bmatrix}}
 \\
 \bullet
 \ar[swap]{r}{\begin{bmatrix} \id & -\eps \\ 0 & \id \end{bmatrix}}
 \&
 \bullet
\end{tikzcd} \!.
\end{equation}
As discussed in Section~\ref{sec:formal}, this system is reducible to
diagonal form by the above equivalence (Lemma~\ref{lem:offdiag-simp}) if
and only if the following equation is satisfied by some $\delta$ and
$\eps$ with rational coefficients:
\begin{equation} \label{eq:rw-offdiag}
 \D_{s_0} \circ \delta = (\Delta-\bar{\Delta}) + \eps \circ \D_{s_1}.
\end{equation}
For simplicity of notation, where it does not introduce any ambiguities,
we set $\bar{\Delta} = 0$ below. By~\cite[Theorem~3.1]{kh-rwtriang},
without loss of generality, we can consider this problem restricted to
first order $\delta$, $\epsilon$ and $\Delta$, with $\delta$ uniquely
determining $\eps$ and vice versa. For instance, if $\Delta$ is of
order higher than one, we can always reduce the order by writing $\Delta
= \Delta_1 + \Delta' \circ \D_{s_1}$ with $\Delta_1$ of order at most
one, and absorb $\Delta'$ into $\eps$ (or analogously with $\delta$).

Let us parametrize these operators as follows:
\begin{subequations}
\begin{gather}
\label{eq:Delta-param}
 \Delta = \frac{{\rm i}\omega r}{r^2}
 (-\Delta_- + \{r f \Delta_+, \del_r\}), \\
\label{eq:delta-param}
 \delta = \delta_+
 - 2\del_r(r f \delta_-) + f_1 \delta_-
 + \{r f \delta_-, \del_r\}, \\
\label{eq:eps-param}
 \eps = \delta_+
 + 2\del_r(r f \delta_-) - f_1 \delta_-
 + \{r f \delta_-, \del_r\},
\end{gather}
\end{subequations}
\noeqref{eq:Delta-param}%
\noeqref{eq:delta-param}%
\noeqref{eq:eps-param}%
where $\delta_\pm$, $\Delta_\pm$ are arbitrary rational functions and we
have used the anti-commutator notation $\{g,h\} = g\circ h + h\circ g$
for differential operators, namely $\{g_0,\del_r\} = 2g_0\del_r +
(\del_r g_0)$ for any function~$g_0$. Plugging this parametrization
into~\eqref{eq:rw-offdiag} and comparing coefficients of powers of
$\del_r$, we find the equivalent ODE system
\begin{gather} \label{eq:rw-decoupling}
\underbrace{\begin{bmatrix}
\frac{1}{{\rm i}\omega} \del_r & \frac{(s_0^2-s_1^2) f_1}{{\rm i}\omega r}
\\[2ex]
-\frac{(s_0^2\!-\!s_1^2) f_1}{{\rm i}\omega r}
& \!\!\!\frac{\del_r r f \del_r r f \del_r
\!-\! \big\{f (2\B_l\! -\! (s_0^2\!+\!s_1^2\!+\!1)f_1 \!+\! \tfrac{1}{2} f
\!-\! \tfrac{2}{f} {\omega^2 r^2}), \del_r\big\}}{{\rm i}\omega}
\end{bmatrix}}_{\DD(s_0,s_1)}\!\!
\begin{bmatrix} \delta_+ \\ \delta_- \end{bmatrix}
\!=\!\begin{bmatrix} \Delta_+ \\ \Delta_- \end{bmatrix} \!.
\end{gather}

The parametrization we have used for $\Delta$ and $\delta$ is different
from the one we used previously in~\cite{kh-rwtriang}, but it is more
convenient for some simple technical reasons. If $\Delta$ is
parametrized by $(\Delta_+, \Delta_-)$, then $\Delta^*$ is parametrized
by $(\Delta_+, -\Delta_-)$. Moreover, given a solution $\delta$
of~\eqref{eq:rw-decoupling} parametrized by $(\delta_+, \delta_-)$, if
we simultaneously replace $\Delta$ by $\Delta^*$ and exchange $s_0$ with
$s_1$, then the solution to the new problem is parametrized in the same
way but by $(\delta_+^*, -\delta_-^*)$. The last property is due to the
formal self-adjointness of $\D_s^* = \D_s$. With our new
parametrization, it is also easy to check that the operator
$\DD(s_0,s_1)$ in~\eqref{eq:rw-decoupling} is itself formally
self-adjoint, and setting $s_0=s_1$ puts it into diagonal form.
The version of the decoupling equation used in~\cite{kh-rwtriang} did
not manifestly exhibit these properties.

Furthermore, the ODE coefficients in~\eqref{eq:rw-decoupling} are not
just rational functions, but even more conveniently they are Laurent
polynomials in $r$. The rationality of the coefficients means that we
can use the results from~\cite{kh-rwtriang} to systematically
solve~\eqref{eq:rw-decoupling}, or check that it has no solution.
Actually, we will need to extend the results of~\cite{kh-rwtriang} in an
important way. When $\Delta$ is such that no solution $\delta$ exists,
we would like to find some canonical choice of $\bar{\Delta}$ such that
replacing the right-hand side by $\big(\Delta - \bar{\Delta}\big)$ does make the
equation solvable. That is, we would like to be able to pick some
convenient representatives of the equivalence classes in the cokernel of
$\DD(s_0,s_1)$ in~\eqref{eq:rw-decoupling}. In fact, exploiting the fact
that its coefficients are actually Laurent polynomials, we will show in
Theorem~\ref{thm:coker} that fixing the locations of the poles of
$\Delta$, the cokernel of the operator in~\eqref{eq:rw-decoupling} is
finite dimensional and a complete set of representatives can be
explicitly selected.

First, let us recall some useful terminology and notions
from~\cite{kh-rwtriang}, where further details and more references to
the literature on this topic can be found.
We will be working with%
 \footnote{What follows is standard notation in commutative algebra,
 where in general $\CC$ can be replaced by any commutative ring $R$ and
 the variable $r$ could be replaced by any set of independent symbolic
 variables.} %
complex valued \emph{rational functions} in $r$,
$\CC(r)$, \emph{polynomials} in $r$ and $r^{-1}$, $\CC[r]$ and
$\CC\big[r^{-1}\big]$, \emph{formal power series} in $r$ and $r^{-1}$,
$\CC[[r]]$ and $\CC\big[\big[r^{-1}\big]\big]$, as well as \emph{formal Laurent series}
in $r$. Let us distinguish between
\emph{unbounded} Laurent series $\CC\big[\big[r,r^{-1}\big]\big]$, \emph{bounded $($from
below$)$} Laurent series $\CC\big[r^{-1}\big][[r]]$, \emph{bounded from above}
Laurent series $\CC[r]\big[\big[r^{-1}\big]\big]$ and \emph{Laurent polynomials}
$\CC[r,r^{-1}]$. Of course, we could also consider Laurent series
centered at some other $r = \rho \ne 0$, but for convenience of notation
whenever possible we will stick with $\rho=0$.

By convention, matrix ordinary differential operators $E[\phi]$ act on
column vectors $\phi$ with components either in $\CC(r)$, or
belonging to one of these classes of Laurent series. When $E[\phi]$ has
bounded Laurent series coefficients, its action on bounded Laurent
series produces bounded Laurent series (from above or below,
respectively). Rational coefficients become bounded Laurent series upon
appropriate expansion. For the action of $E[\phi]$ to be well-defined on
unbounded Laurent series, its coefficients must be Laurent polynomials.
In each case, let us call such an $E$ a \emph{compatible differential
operator}.

We call a Laurent polynomial matrix $S = S(r)$ \emph{Laurent unimodular}
if its inverse $S(r)^{-1}$ is also Laurent polynomial (equivalently, the
determinant of $S$ is proportional to a single power of~$r$ and
non-vanishing). We say that Laurent unimodular matrices $S$ and $T$ are
respectively \emph{source} and \emph{target leading multipliers} of
$E$ when, after expanding all rational coefficients as bounded Laurent
series, we have
\begin{equation} \label{eq:lead-mult}
 E\big[S(r) \phi_n r^n\big] = T(r) \big(E_n \phi_n r^n + r^n O(r)\big),
\end{equation}
with the components of $O(r)$ all in $r\CC[[r]]$ and $E_n$ an
$r$-independent matrix that is invertible for almost all $n$ (i.e.,\ all
but finitely many). We call $E_n$ the \emph{leading characteristic
matrix} of $E$ with respect to the given multipliers. Similarly, we say
that $S$ and $T$ are respectively \emph{source} and \emph{target
trailing multipliers} of $E$ when, after expanding all rational
coefficients as bounded from above Laurent series, we have
\begin{equation} \label{eq:trail-mult}
 E\big[S(r) \phi_n r^n\big] = T(r) \big(E_n \phi_n r^n + r^n O\big(r^{-1}\big)\big),
\end{equation}
with the components of $O\big(r^{-1}\big)$ all in $r^{-1}\CC\big[\big[r^{-1}\big]\big]$ and $E_n$
an $r$-independent matrix that is invertible for almost all $n$. We call
$E_n$ the \emph{trailing characteristic matrix} of $E$ with respect to
the given multipliers.

Those integers $n\in \mathbb{Z}$ such that $\det E_n = 0$, which is a
polynomial equation in $n$, are called (respectively \emph{leading} or
\emph{trailing}) \emph{$($integer$)$ characteristic roots} or
\emph{exponents} of $E$ with respect to given multipliers $S$, $T$. We
denote the set of such leading characteristic exponents by
$\check{\sigma}(E)$ and the set of such trailing characteristic
exponents by $\hat{\sigma}(E)$, with implicit dependence on the $S$, $T$
multipliers, of course.

If the leading (trailing) multipliers $S$, $T$ of $E$ are known, the
formal adjoint operator $E^*$ obviously has $S_{E^*} = T^{-*}$ and
$T_{E^*} = S^{-*}$ as corresponding multipliers,
where $A^{-*} = (A^*)^{-1}$ is the inverse of the
conjugate transpose matrix. Dependence on $n$ in $E_n$ is due to the
action of differential operators $r\del_r r^n = n r^n$. Due to the
identity $(r \del_r)^* = -(r\del_r + 1)$, using the above multipliers
for $E^*$ allows us identify the characteristic exponents of $E^*$ as
$\hat{\sigma}(E^*) = -\hat{\sigma}(E)-1$ and $\check{\sigma}(E^*) =
-\check{\sigma}(E)-1$.

We will not dwell on when leading or trailing multipliers exist or on their uniqueness, but
will just assume that they are given for any particular problem. Our
focus, instead, is to use these data as a certificate that a concrete,
finite dimensional linear algebra computation can \emph{prove} some
statement about a rational ODE. In practice, often $S$ and $T$ may be
taken to be diagonal, with appropriately chosen powers of $r$ on the
diagonal. Otherwise, they could be determined by a recursive procedure
similar to that used in the analysis of regular and irregular
singularities for ODEs with meromorphic coefficients~\cite{wasow}.
Related algorithms can be found in the recent
papers~\cite{abramov16,ark17}.

For the inhomogeneous problem $E[\phi] = \beta$, the solution $\phi$ may
only have poles at a subset of the union of the poles of $\beta$ and of
the (regular or irregular) singular points of $E[\phi]$. When the number
of such singular points is finite, the knowledge of the leading and
trailing multipliers at the singular points (including $r=\oo$) and at
the poles of $\beta$ is sufficient to reduce the problem of finding all
rational solutions $\phi$ for rational $\beta$ to a finite dimensional
linear algebra problem~\cite[Theorem~2.4]{kh-rwtriang}. For our purposes, we
now need to extend this result to characterize both the kernel and
cokernel of $E[\phi]$ when acting on the different classes of Laurent
series. There exist results in the literature~\cite{cluzeau-quadrat, oaku-takayama-tsai} that are, on an abstract level, equivalent or even more
general than what we present below. But, being written in the abstract
language of $D$-modules, they would require substantial interpretation
to be applicable to our very concrete situation.

Let us fix the following notation. We denote the kernel of $E$ by $\ker
E$ in the $\CC\big[r,r^{-1}\big]$ class, by $\ker_+ E$ in the $\CC\big[r^{-1}\big][[r]]$
class, by $\ker_- E$ in the $\CC[r]\big[\big[r^{-1}\big]\big]$ class, and by $\ker_{-+}
E$ in the $\CC\big[\big[r,r^{-1}\big]\big]$ class. For the cokernel of $E$, we write
$\coker E$, $\coker_+ E$, $\coker_- E$ and $\coker_{-+} E$, using the
same conventions. In each case, the operator $E[\phi]$ must have a
well-defined action on the Laurent series of the appropriate class. If
we have in mind Laurent series centered at $r=\rho$ instead of $r=0$,
then we write $\ker^\rho$ and $\coker^\rho$ with appropriate subscripts.
Finally, we write $\ker_{\mathrm{rat}}^{\rho_1,\rho_2,\ldots}$ and
$\coker_{\mathrm{rat}}^{\rho_1,\rho_2,\ldots}$ when we restrict our
attention to rational functions with possible poles only at $r=\rho_1,
\rho_2, \ldots$.

\begin{prp} \label{prp:findim}
If a compatible equation $E[\phi]=0$ $($with given multiplier matrices$)$
has finitely many singular points,
then $\ker E$, $\ker_+ E$, $\ker_- E$, $\ker_{-+} E$ or
$\ker_{\mathrm{rat}}^{\rho_1,\rho_2,\ldots,\rho_N} E$, as appropriate,
is finite dimensional.
\end{prp}
\begin{proof}
For rational functions with prescribed poles, this is the statement
of~\cite[Corollary~2.5]{kh-rwtriang}. For the other cases, the proof can be
attacked with similar methods, so we merely sketch it. For $\ker_+ E$,
we can split $\phi = \phi' + S \phi_+$, using the leading multiplier
$S$, where the powers of $r$ in~$S^{-1}\phi'$ are bounded from below by
$\check{n}$ and from above by $\hat{n}$ satisfying $\check{n} \ll
\check{\sigma}(E) \ll \hat{n}$, while the powers of $r$ in $\phi_+$ are only
bounded from below by $\hat{n}$. Then there are no obstructions to
solving $E[S \phi_+] = -E[\phi']$ for $\phi_+$ order by order, since no
characteristic exponents will appear by our conditions on $\hat{n}$.
Then finding all $\phi'$ such that $E[\phi' + S\phi_+] = 0$ reduces to a
finite dimensional linear algebra problem. The $\ker_- E$ case is
completely analogous. For $\ker_{-+} E$, one needs to start with the
splitting $\phi = \check{S} \phi_- + \phi' + S \phi_+$, where now
$\check{S}$ is the trailing multiplier and again proceed analogously.
\end{proof}

To deal with cokernels, it is helpful to introduce the sesquilinear
\emph{residue pairing}
\begin{equation} \label{eq:res-def}
 \langle \alpha, \phi \rangle_\rho = \Res_\rho (\alpha^* \phi),
\end{equation}
where we take the \emph{residue at $r=\rho$} (defined as the coefficient of
$(r-\rho)^{-1}$) of the product of a column vector $\phi$ and the
conjugate transposed column vector $\alpha$ ($r^*=r$ and its powers are
treated as real under conjugation). We may drop the subscript when $\rho
= 0$, $\langle -, - \rangle = \langle -, - \rangle_0$. The residue
pairing is of course well-defined for rational functions, but it is also
well-defined for the following pairs of Laurent series classes: (bounded
from above and below, unbounded), (bounded from below, bounded from
below), and (bounded from above, bounded from above). This pairing has
several nice properties. It is easy to see that it is
\emph{non-degenerate} in either argument, for each possible combination
of pairs Laurent series classes. If $S$ is a Laurent unimodular matrix,
then the pairing is obviously \emph{invariant} under $\langle S^{-*}
\alpha, S \phi \rangle = \langle \alpha, \phi \rangle$. The following
lemma is slightly less obvious and will be particularly helpful.

\begin{lem} \label{lem:res-adjoint}
For a compatible differential operator $E[\phi]$, the formal adjoint
$E^*[\alpha]$ is also the adjoint with respect to the residue pairing,
$\langle \alpha, E[\phi] \rangle = \langle E^*[\alpha], \phi \rangle$.
The same is true with the arguments of the pairing reversed.
\end{lem}

\begin{proof}
The last statement holds simply because $\langle \alpha, \phi \rangle^*
= \langle \phi, \alpha \rangle$. To prove the adjoint relation, we need
to make use of the defining formula~\eqref{eq:res-def}. First, note that
it is obviously true that $\langle \alpha, C \phi \rangle =
\langle C^* \alpha, \phi \rangle$, when $C$ is a constant matrix, and
that $\langle \alpha, r^n \phi \rangle = \langle r^n \alpha, \phi
\rangle$, for any power of $r$. Then, due to the Leibniz identity
$\alpha^* \del_r \phi - (-\del_r \alpha)^* \phi = \del_r (\alpha^*
\phi)$ and the observation that $\Res_0 (\del_r w) = 0$ for any Laurent
series $w$, it is also true that $\langle \alpha, \del_r \phi \rangle =
\langle -\del_r \alpha, \phi \rangle$. Hence, the adjoitness property
follows for all differential operators, since any operator can be
written as a~sum of compositions of these three kinds of operators.
\end{proof}

We can now state and prove the following lemmas characterizing cokernels
in the bounded Laurent series and Laurent polynomial cases. In each
case, since $E^{**} = E$, we are free to exchange the roles of $E$ and
$E^*$.

\begin{lem} \label{lem:coker-semibounded}
Let $E[\phi] = 0$ be an ODE with coefficients that are Laurent series
bounded from below $($above$)$, with given multiplier matrices.
Then the induced sesquilinear pairing $\langle [\alpha], \phi \rangle =
\Res_0(\alpha^* \phi)$ is non-degenerate between $\coker_+ E^*$ and
$\ker_+ E$ $($between $\coker_- E^*$ and $\ker_- E)$.
\end{lem}

\begin{proof}
Obviously, $E[\phi]$ is a compatible differential operator. We will
consider Laurent series bounded from below (the bounded from above case
is completely analogous).

Let $\mathcal{L}$ and $\mathcal{R}$ denote the spaces of left and right
arguments of the residue pairing. Denote by $\mathcal{L}_n \subset
\mathcal{L}$ and $\mathcal{R}_n \subset \mathcal{R}$ the subspaces
consisting of Laurent series containing powers of $r$ only equal to $n$
or greater. Note that $\mathcal{L}_n^\perp = \mathcal{R}_{-n}$ and vice
versa. By the invariance of the residue pairing, without loss of
generality, we can assume that $E[\phi]$ has trivial source and target
multipliers (they are just identity matrices). Then $E[\mathcal{R}_n]
\subset \mathcal{R}_n$ for any $n$. Fix some exponents $\check{n} \ll
\check{\sigma}(E) \ll \hat{n}$.

By invoking Lemma~\ref{lem:res-adjoint} and using standard arguments, we
can conclude that the residue pairing descends to the non-degenerate
pairing $\langle [\alpha], [\phi] \rangle = \langle \alpha, \phi
\rangle$ on the finite-dimensional sub-quotients $\mathcal{L}_{-\hat{n}}
/ \mathcal{L}_{-\check{n}} \times \mathcal{R}_{\check{n}} /
\mathcal{R}_{\hat{n}}$, with the induced operators $[E][\phi]=[E[\phi]]$
and $[E^*][\alpha] = [E^*[\alpha]]$ still satisfying the adjoint
relation $\langle [\alpha], [E][\phi]] \rangle = \langle [E^*][\alpha]
, [\phi] \rangle$. Hence, by standard finite dimensional linear algebra,
this induced pairing descends also to a non-degenerate pairing between
$\ker [E]$ and $\coker [E^*]$, for the operators induced on these
sub-quotients. But by the same arguments as in the proof of
Proposition~\ref{prp:findim}, the sub-quotients provide isomorphisms
$\ker_+ E = \ker [E]$ and $\coker_+ E^* = \ker [E^*]$. Hence, since these
isomorphisms respect the hierarchy of induced pairings, the residue
pairing induces a non-degenerate pairing between $\coker_+ E^*$ and
$\ker_+ E$.\vspace{-1ex}
\end{proof}

\begin{lem} \label{lem:coker-bounded}
Let $E[\phi] = 0$ be an ODE with Laurent polynomial coefficients $($with
given multiplier matrices$)$. Then the induced sesquilinear pairing $\langle
\alpha, [\phi] \rangle = \Res_0(\alpha^* \phi)$ is non-degenerate
between $\ker_{-+} E^*$ and $\coker E$.
\end{lem}

\begin{proof}
Clearly, $E$ is a compatible differential operator, the induced pairing
is well-defined and the non-degeneracy in the first argument follows
directly from the non-degeneracy of the residue pairing. It remains to
show non-degeneracy in the second argument.

Suppose $\psi$ is a Laurent polynomial and $\langle \beta, [\psi]
\rangle = 0$, for any $\beta \in \ker_{-+} E^* \supset \ker_\pm E^*$.
Then, by Lemma~\ref{lem:coker-semibounded}, there exist bounded
(respectively from above or below) Laurent series $\phi_\pm$ such that
$\psi = E[\phi_\pm]$. Pick any Laurent polynomial $\alpha \in (\ker_+
E)^\perp \cap (\ker_- E)^\perp$. Then, again by
Lemma~\ref{lem:coker-semibounded}, there exist bounded (respectively
from above or below) Laurent series $\beta_\pm$ such that $\alpha =
E^*[\beta_\pm]$, and hence $\beta_+ - \beta_- \in \ker_{-+} E^*$. By the
hypotheses on $\psi$,
\begin{align*}
 0 &= \langle \beta_+ - \beta_-, [\psi] \rangle
 = \langle \beta_+, E[\phi_+] \rangle - \langle \beta_-, E[\phi_-] \rangle
= \langle E^*[\beta_+], \phi_+ \rangle
 - \langle E^*[\beta_-], \phi_- \rangle
\\
 &= \langle \alpha, \phi_+ - \phi_- \rangle.
\end{align*}
Since the residue paring is non-degenerate and $\alpha$ was arbitrary up
to being annihilated by $\ker_\pm E$ (finite dimensional subspaces, by
Proposition~\ref{prp:findim}), there must be $\phi'_\pm \in \ker_\pm E$
such that $\phi_+ - \phi_- = \phi'_+ - \phi'_-$. But then $\phi_+ -
\phi'_+ = \phi_- - \phi'_-$ and hence, by their boundedness properties,
both sides are equal to a Laurent polynomial $\phi$ such that $E[\phi] =
\psi$, implying as desired that $[\psi] = 0 \in \coker E$.\vspace{-1ex}
\end{proof}

Finally, we can package the above results in the following, for our
purposes, conveniently phrased
\begin{thm} \label{thm:coker}
Consider an ordinary differential operator $E[\phi]$ with Laurent
polynomial coefficients with a finite number of singular points $($with
given multiplier matrices$)$. Suppose that the leading multipliers and
integer characteristic exponents are known at each point $\rho \in
\mathbb{C}$. The formal adjoint $E^*[\alpha]$ will share the same
properties.
\begin{enumerate}\itemsep=0pt
{\samepage\item[$(a)$] Given $\rho \in \mathbb{C}$, there is a finite basis
$\big(\alpha^\rho_j\big)$ spanning a subspace of $\ker_+^\rho E^*$ and as many
polynomials $\big(\psi^\rho_k\big)$ in $(r-\rho)^{-1}$ such that the residue
pairing $R^\rho_{jk} = \Res_\rho \big(\big(\alpha^\rho_j\big)^* \psi^\rho_k\big)$ is
full rank and such that, for any rational $\psi$, there exist constants
$c_\rho^k$ and a Laurent polynomial~$\phi_\rho$ such that
\begin{gather}
 \bar{\psi}^\rho = \psi - \sum_k c_\rho^k \psi^\rho_k - E[\phi_\rho],\vspace{-1ex}
\end{gather}
has no poles at $r=\rho$.

}

\item[$(b)$] There is a finite basis $(\alpha_j)$ for $\ker_{-+} E^*$ and as many
Laurent polynomials $(\psi_k)$ such that $R_{jk} = \Res_\rho (\alpha_j^*
\psi_k)$ is full rank, for any Laurent polynomial $\psi$, there
exist constants $c^k$ and a Laurent polynomial $\phi$ such that
\begin{equation}
 E[\phi] = \psi - \sum_k c^k \psi_k.
\end{equation}
\end{enumerate}
\end{thm}

\begin{proof}
Part $(a)$ follows from Proposition~\ref{prp:findim} and
Lemma~\ref{lem:coker-semibounded}. The implication is that both
$\ker_+^\rho E^*$ and $\coker_+^\rho E$ are finite dimensional. Hence,
the dimension of the subspace of $\coker_+^\rho E$ represented by
polynomials in $(r-\rho)^{-1}$ must also be finite. Let $\big(\psi^\rho_k\big)$
be a set of polynomials in $(r-\rho)^{-1}$ that represent a basis of this
subspace. Then, by non-degeneracy of the residue pairing, we can pick a
finite subset $\big(\alpha_\rho^k\big)$ satisfying the desired non-degeneracy
condition. So by subtracting the appropriate linear combination $\sum_k
c_\rho^k \psi^\rho_k$ from $\psi$ we can find a representative
$\bar{\psi}^\rho$ from the same equivalence class in $\coker_+^\rho E$
that has no poles at $r = \rho$. This argument gives us~$\phi_\rho$ that
is merely bounded from below as a Laurent series, but once we have it,
we can simply truncate it after some sufficiently high finite power or
$r$.

Part $(b)$ follows from Proposition~\ref{prp:findim} and
Lemma~\ref{lem:coker-bounded} similarly. Once we know that $\ker_{-+}
E^*$ and $\coker E$ are finite dimensional, it remains only to pick a
basis $(\alpha_j)$ for $\ker_{-+} E^*$ and a finite set $(\psi_k)$ whose
representatives form a basis of $\coker E$.
\end{proof}

\begin{rmk}
Theorem~\ref{thm:coker} also implies that
$\coker_{\mathrm{rat}}^{\rho_1,\rho_2,\ldots} E$, the cokernel of $\phi
\mapsto \psi = E[\phi]$ for~$\psi$ and $\phi$ rational with potential
pole locations fixed at $\rho_1, \rho_2, \ldots$, is finite dimensional.
The functions $\big(\psi_k, \psi^{\rho_1}_k, \psi^{\rho_2}_k, \ldots\big)$ then
constitute a complete set of independent representatives of the cokernel
equivalence classes exactly when they satisfy the conditions in $(a)$ and
$(b)$ of the theorem.
\end{rmk}

We conclude this section by applying Theorem~\ref{thm:coker} to the
operator $\DD(s_0,s_1)$ from the Regge--Wheeler decoupling
equation~\eqref{eq:rw-decoupling}. Let us denote the multiplier and
characteristic matrices as in the following identity:
\begin{equation}
 \DD(s_0,s_1) S_\rho \delta_n (r-\rho)^n
 = T_\rho \big(E_{\rho,n} \delta_n + O(r-\rho)\big) (r-\rho)^n.
\end{equation}

The only singular points of $\DD(s_0,s_1)$ are $r=0$, $2M$, $\oo$. They are
accompanied by the following data:
\begin{itemize}\itemsep=0pt
\item $\rho=0$:
 $\check{\sigma}_0(\DD(s_0,s_1)) = \{\pm s_0 \pm s_1\}$, with
\begin{gather}
S_0 = \begin{bmatrix}1 & 0 \\0 & \tfrac{r}{2M}\end{bmatrix} \!, \quad\
T_0 = \begin{bmatrix}\tfrac{2M}{r} & 0 \\0 & \tfrac{4M^2}{r^2}\end{bmatrix}\!,\quad\
E_{0,n} = \frac{1}{2{\rm i}M\omega}
\begin{bmatrix} n & s_0^2-s_1^2 \\ s_1^2-s_0^2 & n(n^2 - 2s_0^2 - 2s_1^2) \end{bmatrix}\!,
\\
\det E_{0,n}= -\frac{(n+s_0+s_1) (n+s_0-s_1) (n-s_0+s_1) (n-s_0-s_1)}{4M^2\omega^2}.
\end{gather}

\item $\rho=2M$:
$\check{\sigma}_{2M}(\DD(s_0,s_1)) = \{ 0 \}$, with
\begin{gather}
S_{2M} = \begin{bmatrix}1 & 0 \\0 & 1\end{bmatrix} \!, \qquad
T_{2M} = \begin{bmatrix}\tfrac{1}{f} & 0 \\0 & \tfrac{1}{f}\end{bmatrix}\!, \qquad 
E_{2M,n} = \frac{1}{2{\rm i}M\omega}
\begin{bmatrix}
n & 0 \\0 & n (n^2 + 16 M^2\omega^2)\end{bmatrix}\!,
\\
\det E_{2M,n} = -\frac{n^2(n^2 + 16 M^2\omega^2)}{4M^2\omega^2}.
\end{gather}

\item $\rho=\oo$:
$\hat{\sigma}_\oo(\DD(s_0,s_1)) = \{ 0 \}$, with
\begin{gather}
S_\oo = \begin{bmatrix} 1 & 0 \\ 0 & \tfrac{1}{\omega r} \end{bmatrix} \!, \qquad
T_\oo = \begin{bmatrix} \tfrac{1}{\omega r} & 0 \\ 0 & 1 \end{bmatrix}\!, \qquad
E_{\oo,n} = \frac{1}{\rm i} \begin{bmatrix} n & 0 \\ 0 & 4n \end{bmatrix}\!,\qquad
\det E_{\oo,n} = -4n^2.
\end{gather}

\item $\rho \not\in \{ 0, 2M, \oo \}$:
$\hat{\sigma}_\rho(\DD(s_0,s_1)) = \{ 0, 1, 2 \}$, with
\begin{gather}
S_\rho = \begin{bmatrix} 1 & 0 \\ 0 & 1 \end{bmatrix}\!, \quad\
T_\rho = \begin{bmatrix}\tfrac{1}{(r-\rho)} & 0 \\0 & \tfrac{1}{(r-\rho)^3}\end{bmatrix}, \quad\
E_{\rho,n} = \frac{1}{{\rm i}\omega}
\begin{bmatrix} n & 0 \\ 0 & \rho^2 f(\rho)^2 n (n-1) (n-2) \end{bmatrix}\!,
\\
\det E_{\rho,n} = -\frac{\rho^2 f(\rho)^2 n^2 (n-1) (n-2)}{\omega^2}.
\end{gather}
\end{itemize}

Notice that, besides $r=0$ and $r=\oo$, the only singular point of
$\DD(s_0,s_1)$ is $r=2M$. In practice (in Section~\ref{sec:schw}), we
will only deal with rational functions that have poles at these three
points, so we will restrict our attention only to them. Given the above
data, it is possible to explicitly calculate (using computer algebra)
the dimensions of the solution spaces $\ker_+^{2M} \DD(s_0,s_1)^*$ and
$\ker_{-+} \DD(s_0,s_1)^*$, as well as their bases. Since these solutions
are presented by infinite series, it is impractical to write them here
explicitly. Instead, we will simply report their dimensions, while
explicitly writing a choice corresponding cokernel representatives
$\big(\Delta_k, \Delta^{2M}_k\big)$. In practice, the self-adjointness
$\DD(s_0,s_1)^* = \DD(s_0,s_1)$ turns out to be extremely useful and
saves a significant amount of effort since the same computer algebra
code can be used for both operators.

For $\Delta$ involving poles at points other than $\{0,\oo\}$, we can
report the following.
\begin{itemize}\itemsep=0pt
\item $\rho \not\in \{0,2M,\oo\}$:
 $\dim \ker_+^\rho \DD(s_0,s_1)^* = 2$, while
 \begin{equation}
 (\Delta^\rho_k) = \left(
 \begin{bmatrix} \tfrac{1}{r-\rho} \\ 0 \end{bmatrix}\!,
 \begin{bmatrix} 0 \\ \tfrac{1}{r-\rho} \end{bmatrix}
 \right) \!.
 \end{equation}
 Although we will not need this case later (in Section~\ref{sec:schw}),
 we report it here for completeness.
\item $\rho=2M$:
 $\dim \ker_+^{2M} \DD(s_0,s_1)^* = 2$, while
 \begin{equation}
 (\Delta^\rho_k) = \left(
 \begin{bmatrix} \tfrac{1}{r-2M} \\ 0 \end{bmatrix}\!,
 \begin{bmatrix} 0 \\ \tfrac{1}{r-2M} \end{bmatrix}
 \right) = \left(
 \begin{bmatrix} \tfrac{1}{rf} \\ 0 \end{bmatrix}\!,
 \begin{bmatrix} 0 \\ \tfrac{1}{rf} \end{bmatrix}
 \right)\!.
 \end{equation}
\end{itemize}

For $\Delta$ a Laurent polynomial, we can report the following. This
data depends on the values of $s_0$ and $s_1$, so we concentrate only on
the cases that will be needed later (in Section~\ref{sec:schw}).
\begin{itemize}\itemsep=0pt
\item $(s_0,s_1) = (0,0)$:
 $\dim \ker_{-+} \DD(s_0,s_1)^* = 5$, while
 \begin{equation}
 (\Delta_k) = \left(
 \begin{bmatrix} \tfrac{1}{r} \\ 0 \end{bmatrix}\!,
 \begin{bmatrix} 0 \\ 1 \end{bmatrix}\!,
 \begin{bmatrix} 0 \\ \tfrac{1}{r} \end{bmatrix}\!,
 \begin{bmatrix} 0 \\ \tfrac{1}{r^2} \end{bmatrix}\!,
 \begin{bmatrix} 0 \\ \tfrac{1}{r^3} \end{bmatrix}
 \right)\!.
 \end{equation}
\item $(s_0,s_1) = (1,1)$:
 $\dim \ker_{-+} \DD(s_0,s_1)^* = 5$, while
 \begin{equation}
 (\Delta_k) = \left(
 \begin{bmatrix} \tfrac{1}{r} \\ 0 \end{bmatrix}\!,
 \begin{bmatrix} 0 \\ 1 \end{bmatrix}\!,
 \begin{bmatrix} 0 \\ \tfrac{1}{r} \end{bmatrix}\!,
 \begin{bmatrix} 0 \\ \tfrac{1}{r^2} \end{bmatrix}\!,
 \begin{bmatrix} 0 \\ \tfrac{1}{r^4} \end{bmatrix}
 \right) \!.
 \end{equation}
\item $(s_0,s_1) = (2,2)$:
 $\dim \ker_{-+} \DD(s_0,s_1)^* = 5$, while
 \begin{equation}
 (\Delta_k) = \left(
 \begin{bmatrix} \tfrac{1}{r} \\ 0 \end{bmatrix}\!,
 \begin{bmatrix} 0 \\ 1 \end{bmatrix}\!,
 \begin{bmatrix} 0 \\ \tfrac{1}{r} \end{bmatrix}\!,
 \begin{bmatrix} 0 \\ \tfrac{1}{r^2} \end{bmatrix}\!,
 \begin{bmatrix} 0 \\ \tfrac{1}{r^6} \end{bmatrix}
 \right)\!.
 \end{equation}
\item $(s_0,s_1) = (0,1)$ or $(1,0)$:
 $\dim \ker_{-+} \DD(s_0,s_1)^* = 4$, while
 \begin{equation}
 (\Delta_k) = \left(
 \begin{bmatrix} \tfrac{1}{r} \\ 0 \end{bmatrix}\!,
 \begin{bmatrix} 0 \\ 1 \end{bmatrix}\!,
 \begin{bmatrix} 0 \\ \tfrac{1}{r} \end{bmatrix}\!,
 \begin{bmatrix} 0 \\ \tfrac{1}{r^3} \end{bmatrix}
 \right)\!.
 \end{equation}
\item $(s_0,s_1) = (0,2)$ or $(2,0)$:
 $\dim \ker_{-+} \DD(s_0,s_1)^* = 4$, while
 \begin{equation}
 (\Delta_k) = \left(
 \begin{bmatrix} \tfrac{1}{r} \\ 0 \end{bmatrix}\!,
 \begin{bmatrix} 0 \\ 1 \end{bmatrix}\!,
 \begin{bmatrix} 0 \\ \tfrac{1}{r} \end{bmatrix}\!,
 \begin{bmatrix} 0 \\ \tfrac{1}{r^4} \end{bmatrix}
 \right)\!.
 \end{equation}
\item $(s_0,s_1) = (1,2)$ or $(2,1)$:
 $\dim \ker_{-+} \DD(s_0,s_1)^* = 4$, while
 \begin{equation}
 (\Delta_k) = \left(
 \begin{bmatrix} \tfrac{1}{r} \\ 0 \end{bmatrix}\!,
 \begin{bmatrix} 0 \\ 1 \end{bmatrix}\!,
 \begin{bmatrix} 0 \\ \tfrac{1}{r} \end{bmatrix}\!,
 \begin{bmatrix} 0 \\ \tfrac{1}{r^5} \end{bmatrix}
 \right)\!.
 \end{equation}
\end{itemize}

\begin{rmk} \label{rmk:relbound}
Obviously, the choice of cokernel representatives given above is not
unique. Our choice was mainly dictated by simplicity and practicality.
However, it satisfies another important requirement. When each cokernel
representative is rewritten as a differential operator $\Delta =
\frac{\Delta_0 + \Delta_1 f\del_r}{f}$ via the
parametrization~\eqref{eq:Delta-param}, the resulting coefficients
$\Delta_0$ and $\Delta_1$ are bounded functions of $r$ on the interval
$(2M,\oo)$. This requirement is sufficient to ensure that the $\Delta$
is \emph{relatively bounded} by $\D_s$ as a (possibly unbounded)
operator $L^2\big(2M,\oo; f^{-1}\,\d{r}\big) \to L^2(2M,\oo; f\,\d{r})$, which
means that $\|\Delta\phi\| \le C_2 \| f\D_s \phi \| + C_0 \|\phi\|$ for each
$\phi \in D(\D_s)$, with some constants $C_0, C_2 > 0$. In turn,
relative boundedness ensures that the operator matrix inverse
\begin{equation}
 \begin{bmatrix}
 \D_{s_0} & \Delta \\
 0 & \D_{s_1}
 \end{bmatrix}^{-1}
 =
 \begin{bmatrix}
 \D_{s_0}^{-1} & -\D_{s_0}^{-1} \Delta \D_{s_1}^{-1} \\[.5ex]
 0 & \D_{s_1}^{-1}
 \end{bmatrix}
\end{equation}
is itself bounded whenever the inverses $\D_{s_0}^{-1}$ and
$\D_{s_1}^{-1}$ are bounded.

To see the desired relative boundedness, it is obviously equivalent to
establish the relative boundedness of $f\Delta$ with respect to $f\D_s$.
If we introduce the Regge--Wheeler \emph{tortoise coordinate} $r_* =
r_*(r)$, defined by $\d{r_*} = \frac{\d r}{f}$, then $L^2\big(2M,\oo;
f^{-1}\d{r}\big) \cong L^2(-\oo,\oo; \d{r_*}) =: \mathcal{H}$ and $f\D_s$ is
an unbounded operator on $\mathcal{H}$ (in fact~\cite{dimock-kay-schw1},
it is also essentially self-adjoint starting from the core of compactly
supported test functions, with its $\omega$-resolvent set equal to
$\mathbb{C}\setminus \mathbb{R}$). By the boundedness conditions on
$\Delta_0$ and $\Delta_1$, $f\Delta$ can be relatively bounded on
$\mathcal{H}$ by a constant coefficient differential operator in
$\del_{r_*} = f\del_r$ of order one or higher. But since the
$\del_{r_*}$-coefficients of $f\D_s$ are bounded and the second order
coefficients even tend to positive constants as $|r_*|\to \oo$, $f\D_s$
can relatively bound any constant coefficient operator of order two or
lower, which hence includes $f\Delta$.
\end{rmk}

Finally, we would like to record here another result that was
actually already implicit in the examples presented
in~\cite[Section~4.1]{kh-rwtriang}, but not made explicit there. If we
consider the Regge--Wheeler decoupling equation~\eqref{eq:rw-decoupling}
with zero right-hand side, then solutions constitute morphisms from
$\D_{s_1} \phi_1 = 0$ to $\D_{s_0} \phi_0 = 0$. Such a morphism, if it
existed between different $s_0$ and $s_1$ it could be called a
\emph{spin raising/lowering} operator for radial modes. Such spin
raising and lowering operators are known to exist for angular
modes~\cite{whiting-shah-spheroidal} and their existence for radial
modes would have important applications. Unfortunately, explicit
computation shows that the result is negative, at least when we restrict
ourselves to operators with rational coefficients.

\begin{thm}\label{thm:rw-maps}
For $\omega \ne 0$ and for $s_0,s_1 \in \{0,1,2\}$ there exist no
non-vanishing morphisms between the equations $\D_{s_0}\phi_0 = 0$ and
$\D_{s_1} \phi_1 = 0$, with the exception of the identity morphism
$\phi_0 \mapsto \phi_1 = \phi_0$ when $s_0 = s_1$.
\end{thm}

The proof follows from direct computation, by using the methods
of~\cite[Theorem~2.4]{kh-rwtriang} to find all rational solutions of the
Regge--Wheeler decoupling equation~\eqref{eq:rw-decoupling}. We have only
done the computation for the above values of $s_0$ and $s_1$ because the
location of the characteristic exponents of the equation depends on
these spins in such a way that the dimension of the corresponding linear
algebra problem grows as $O(|s_0|+|s_1|)$. But it is likely that the
above result actually holds for all $s_0$ and $s_1$.

\subsection{Equivalence of Regge--Wheeler and Zerilli equations}\label{sec:rwz}

We have defined the spin-$s$ Regge--Wheeler operator $\D_s$
in~\eqref{eq:rw}. When $s=2$, $\D_2$ is closely related to the so-called
\emph{Zerilli} operator
\begin{equation}
 \D_2^+ \phi = \del_r f \del_r \phi
 - \frac{\A_l + 3(\B_l-2)f_1 (1+3f_1) + 9 f_1^3}
 {r^2(\B_l-2+3f_1)^2}
 + \frac{\omega^2}{f} \phi.
\end{equation}
In fact, the equations $\D_2\phi = 0$ and $\D_2^+\phi = 0$ are
equivalent (in the sense of Section~\ref{sec:formal}) for all values of
$\omega$ except the \emph{algebraically special}
frequencies~\cite{couch-newman, maassen-brink} satisfying
\begin{equation} \label{eq:alpha-def}
 \alpha := (6\omega r f_1)^2 + \A_l^2 = (12 M \omega)^2 + \A_l^2 = 0.
\end{equation}
This can be seen from the following equivalence diagram:
\begin{equation} \label{eq:rwz-equiv}
\scalebox{1.2}{\begin{tikzcd}[column sep=3cm,row sep=2cm]
 \bullet
 \ar{d}[description]{\D_2^+}
 \ar[shift left]{r}{\frac{(12M)^2}{\alpha} f\C_-}
 \&
 \bullet
 \ar[shift left]{l}{f \C_+}
 \ar{d}[description]{\frac{\alpha}{(12M)^2} \D_2}
 \\
 \bullet
 \ar[shift right,swap]{r}{\C_- f}
 \ar[dashed, bend left=50]{u}{\frac{(12M)^2}{\alpha} f}
 \&
 \bullet
 \ar[shift right,swap]{l}{\frac{(12M)^2}{\alpha} \C_+ f}
 \ar[dashed, bend right=50,swap]{u}{\frac{(12M)^4}{\alpha^2} f,}
\end{tikzcd}}
\end{equation}
where
\begin{equation}
 \C_\pm = \pm \del_r + \frac{\A_l}{12M f} + \frac{3f_1}{r(\B_l-2+3f_1)}
\end{equation}
are the \emph{Chandrasekhar transformation
operators}~\cite[Section~4.26]{chandrasekhar}, \cite{dotti-prl,
gjk-darboux}. The validity of the above equivalence diagrams may be
checked by direct calculation, or can be seen more abstractly from the
identities
\begin{gather}
 \D_2 = -\C_- f \C_+ + \frac{\alpha}{(12M)^2 f}, \qquad
 \D_2^+ = -\C_+ f \C_- + \frac{\alpha}{(12M)^2 f}.
\end{gather}

The necessity of the failure of the equivalence identities
in~\eqref{eq:rwz-equiv} when $\alpha = 0$ (due to division by $\alpha$)
can be explained as follows. When $\alpha=0$, the two operators simplify
to $\D_2 = -\C_- f \C_+$ and $\D_2^+ = -\C_+ f \C_-$. Reducing each of
them to a first order system gives the equivalences
\begin{gather}
\D_2\sim \begin{bmatrix} \C_+ & \tfrac{1}{f} \\ 0 & \C_- \end{bmatrix}\!, \qquad
\begin{bmatrix} \C_- & \tfrac{1}{f} \\ 0 & \C_+ \end{bmatrix} \sim \D_2^+\!.
\end{gather}
Then, using the same machinery that we applied earlier to upper
triangular Regge--Wheeler systems (Theorem~\ref{thm:rw-maps} and other
results), one can conclude that $(a)$ there are no non-vanishing morphisms
between the equations $\C_+\phi_+ = 0$ and $\C_-\phi_- = 0$ and $(b)$
neither of the upper triangular systems above can be decoupled to
diagonal form. On the other hand, given $(a)$, it is not hard to see that
the existence of an equivalence $\D_2 \sim \D_2^+$ would imply that both of
the above upper triangular systems decouple to the diagonal form
$\left[\begin{smallmatrix} \C_+ & 0 \\ 0 & \C_-
\end{smallmatrix}\right]$, which contradicts $(b)$.

Thus, no manipulation will be able to get rid of the poles at $\alpha=0$
from the equivalence in~\eqref{eq:rwz-equiv}.

\begin{rmk} \label{rmk:rwz-selfadj}
Since both $\D_2$ and $\D_2^+$ are formally self-adjoint, while $\alpha$
and $M$ are real, taking formal adjoints of all the operators in the
equivalence diagram~\eqref{eq:rwz-equiv} gives another equivalence
diagram. The way that we have distributed various factors of $\alpha$
and the adjoint identity $\C_+ = \C_-^*$ makes the new equivalence
diagram identical to the original one, and hence the whole diagram
formally self-adjoint. While we could have chosen to distribute various
constant factors among the operators in~\eqref{eq:rwz-equiv} in
different ways, the choice of no prefactor for $\D_2^+$ and the
self-adjointness of the whole diagram dictates the coefficient
$\frac{\alpha}{(12M)^2}$ in front of $\D_2$.
\end{rmk}

\section{Tensor wave equations on Schwarzschild}\label{sec:schw}

Consider the exterior Schwarzschild spacetime $(\M,\gf)$ of
mass $M>0$, where $\M \cong \mathbb{R}^2\times S^2$. When
$\big(S^2,\Omega\big)$ is the unit round sphere and $(t,r)$, with $-\oo<t<\oo$
and $2M<r<\oo$, are the Schwarzschild coordinates on the $(\mathbb{R}^2,
g)$ factor, the exterior Schwarzschild metric has the following
$(2+2)$-warped product form:
\begin{equation}
 \gf := g + r^2 \Omega, \qquad
 g := -f(r) \, {\d t}^2 + \frac{{\d r}^2}{f(r)}, \qquad
 f := 1 - \frac{2M}{r}.
\end{equation}
For convenience, we also define (consistently with Section~\ref{sec:rw})
\begin{equation}
 f_1 := r \del_r f = \frac{2M}{r}.
\end{equation}
This is a vacuum spacetime \big(the Einstein tensor ${}^4G_{\mu\nu} =
{}^4R_{\mu\nu} - \frac{1}{2}\, {}^rR \, \gf_{\mu\nu} = 0$ vanishes\big)
descri\-bing the asymptotically flat region outside the horizon of a
static spherically symmetric black hole of mass $M$.

The goal of this section is to write down the vector and Lichnerowicz
wave equations
\begin{equation} \label{eq:vw-lich-4def}
 \sqone v_\mu := \sqf v_\mu = 0
 \qquad \text{and} \qquad
 \sqtwo p_{\mu\nu} := \dalf p_{\mu\nu}
 - 2 \, {}^4 R_{(\mu}{}^{\la\ka}{}_{\nu)} p_{\la\ka} = 0,
\end{equation}
on this spacetime and to reduced them to their radial mode equations,
with respect to spherical and time harmonic modes. Eventually we will
also show their equivalence to a simplified triangular system of
Regge--Wheeler equations~\eqref{eq:rw}.

For background, let us recall that the vector wave operator $\sqone$ can
be obtained from the Maxwell equation by adding a term factoring through
the harmonic (Lorenz) gauge fixing condition $\grf^\nu v_\nu = 0$,
namely
\begin{equation}
 \sqone v_\mu = \grf^\nu \left(\grf_\nu v_\mu - \grf_\mu v_\nu\right)
 + \grf_\mu \grf^\nu v_\nu \,.
\end{equation}
Similarly, the Lichnerowicz wave operator $\sqtwo$ can be obtained from
the linearized Einstein equation by adding a term factoring through the
harmonic (de~Donder) gauge fixing condition $\grf^\la
\overline{p}_{\nu\la} = 0$, namely
\begin{equation}
 \sqtwo p_{\mu\nu}
 = -2\, \overline{{}^4\dot{G}_{\mu\nu}[p]}
 + 2\, \grf_{(\mu} \grf^\la \overline{p}_{\nu)\la},
\end{equation}
where $\overline{p}_{\mu\nu} = p_{\mu\nu} - \frac{1}{2} p^\la_\la \,
\gf_{\mu\nu}$ is the \emph{trace reversal} operation and
${}^4\dot{G}_{\mu\nu}[p]$ is the linearized Einstein tensor, defined by
${}^4G_{\mu\nu}\big[\gf + p\big] = {}^4\dot{G}_{\mu\nu}[p] + O\big(p^2\big)$, keeping in
mind that ${}^4G_{\mu\nu}\big[\gf\big] = 0$ for the background Schwarzschild
metric.

We will proceed by first using the covariant formalism
of~\cite{martel-poisson} for $(2+2)$-warped products and spherical
harmonics, then by using Schwarzschild coordinates, then by presenting
the ingredients needed to implement the general decoupling strategy from
Section~\ref{sec:triang-strategy}, and finally by giving the explicit
equivalence. The intermediate calculations in the decoupling strategy
are rather long, unenlightening and best performed by computer algebra.
Thus, we present only the final result, but along the way giving some
guiding comments on how it could be reproduced. The reader is referred
to~\cite{kh-vwtriang}, where the intermediate steps were worked out in
more detail for the vector wave operator $\sqone$.

To start, let us briefly introduce the covariant $2+2$ formalism
from~\cite{martel-poisson}. We use Greek indices, $(\mu,\nu,\ldots)$,
for \emph{spacetime} or $\M$-tensors, lower or upper case Latin indices,
$(a,b,\ldots)$ or $(A,B,\ldots)$, respectively, for tensors on the
$\mathbb{R}^2$ or $S^2$ factors. The tangent and cotangent subspaces of
$\mathbb{R}^2$ and $S^2$ are naturally identified with the appropriate
subspaces of $T\M$ and $T^*\M$, allowing us to map tensors from its
factors to $\M$. Thus, any spacetime tensor can be decomposed into
sectors with possibly mixed lower-upper case Latin indices, e.g.,
\begin{gather}
v_\mu \to \begin{bmatrix} v_a \\ r v_A \end{bmatrix}\!,\qquad
w^\nu \to \begin{bmatrix} w^b & r^{-1} v^B \end{bmatrix} \!.
\end{gather}
Note the extra factors of $r$, which make it convenient to raise and
lower Greek indices with $\gf_{\mu\nu}$, while their sector components
have their indices raised and lowered respectively by $g_{ab}$ and
$\Omega_{AB}$. This convention extends to higher rank tensors by
respecting tensor products. We define the following spacetime tensors by
their decomposition into sectors:
\begin{gather}
g_{\mu\nu}\rightarrow
\begin{bmatrix} g_{ab} & 0 \\ 0 & 0 \end{bmatrix}\!, \qquad
\Omega_{\mu\nu}\rightarrow
\begin{bmatrix} 0 & 0 \\ 0 & \Omega_{AB} \end{bmatrix}\!.
\end{gather}
We denote respectively by $\nabla_a$ and $D_A$ the Levi-Civita
connections on $\big(\mathbb{R}^2, g\big)$ and $\big(S^2, \Omega\big)$. In~terms of them
and the decomposition into sectors, the spacetime Levi-Civita connection
is determined by
\begin{equation}
\grf_\mu v_\nu \to
\begin{bmatrix}
\nabla_a v_b & r \nabla_a v_B \\ r D_A v_b & r^2 D_A \tfrac{v_B}{r}
\end{bmatrix}
- \begin{bmatrix}
0 & 0 \\ r \frac{r_b}{r} v_A & -r^2 \Omega_{AB} \tfrac{r^c v_c}{r}
\end{bmatrix}\!,\qquad \text{where} \quad
 r_a = \nabla_a r = (\d r)_a.
\end{equation}
Now, working in Schwarzschild $(t,r)$ coordinates, we define $t^a =
(\del_t)^a$, which is a timelike Killing vector. We also introduce the
notation $t_a = g_{ab} t^b = -f (\d t)_a$ and $\eps_{ab} = 2 (\d t)_{[a}
(\d r)_{b]} = -2 f^{-1} t_{[a} r_{b]}$. Then $r_a r^a = f$, $t_a t^a =
-f$, and $t^a = -\eps^{ab} r_b$. The Levi-Civita connection $\nabla_a$
is determined by
\begin{equation} \label{eq:schw-cov-frame}
\nabla_a t_b = \frac{f_1}{2r} \eps_{ab} \qquad \text{and} \qquad
\nabla_a r_b = \frac{f_1}{2r} g_{ab}.
\end{equation}
The respective Riemann tensors of the $\big(\mathbb{R}^2, g\big)$ and $\big(S^2,
\Omega\big)$ factors are
\begin{equation}
R_{abcd} = \frac{f_1}{r^2} (g_{ac}g_{bd} - g_{ad}g_{bc})\qquad \text{and} \qquad
R_{ABCD} = \Omega_{AC}\Omega_{BD} - \Omega_{AD}\Omega_{BC}.
\end{equation}
Finally, the spacetime Riemann tensor (computed
in~\cite[equation~(23)]{kh-tangherlini}, for instance) is conveniently
expressed as
\begin{align*}
{}^4 R_{\mu\nu\la\ka}
 = {}&\frac{f_1}{r^2} (g_{\mu\la}g_{\nu\ka} - g_{\mu\ka}g_{\nu\la}) 
 + \frac{f_1}{r^2} (\Omega_{\mu\la}\Omega_{\nu\ka} - \Omega_{\mu\ka}\Omega_{\nu\la})
\\
 &- \frac{f_1}{2 r^2} (g_{\mu\la} \Omega_{\nu\ka}
 - g_{\nu\la} \Omega_{\mu\ka}
 - g_{\mu\ka} \Omega_{\nu\la}
 + g_{\nu\ka} \Omega_{\mu\la}).
\end{align*}

Scalar spherical harmonics $Y^{lm}$, where $l=0,1,2,\ldots$ and $-l \le
m \le l$ are respectively the \emph{angular momentum} and
\emph{magnetic} quantum numbers, are normalized eigenfunctions of the
Laplacian on the unit sphere $\big(S^2, \Omega\big)$, defined by the conditions
\begin{equation}
 D_A D^A Y^{lm} = - l(l+1) Y^{lm}, \qquad 
 (Y^{lm})^* = Y^{l(-m)}, \qquad
 \langle (Y^{l'm'})^* \; Y^{lm} \rangle = \delta_{l'l} \delta_{m'm},
\end{equation}
where we have used the short-hand $\langle - \rangle =
\int_{\mathbb{S}^2}\epsilon \, (-)$, with $\epsilon_{AB}$ the volume
form on $\big(S^2, \Omega\big)$. Where no confusion would arise, we shall
regularly omit the $l$ and $m$ indices from $Y^{lm}$ and their
coefficients.
We use the convenient notations (already mentioned at the top of
Section~\ref{sec:appl})
\begin{equation}
 \B_l := l(l+1)
 \qquad \text{and} \qquad
 \A_l := (l-1)l(l+1)(l+2) = \B_l (\B_l-2).
\end{equation}
By taking covariant derivatives of the $Y$ scalar harmonics and
combining them with~$\Omega_{AB}$ and~$\epsilon_{AB}$, we can define
also vector and tensor spherical harmonics:
\begin{alignat*}{3}
 &Y_A:= D_A Y, \qquad&&
 Y_{AB}:= D_A Y_B + \frac{\B_l}{2} \Omega_{AB} Y,&
\\
 & X_A := \epsilon_{BA} D^B Y, \qquad&&
 X_{AB}:= D_A X_B + \frac{\B_l}{2} \epsilon_{AB} Y.&
\end{alignat*}
The above definitions are chosen so that the following convenient
orthogonality and normalization properties hold:
\begin{alignat*}{3}
 &\big\langle \big(Y^A_{l'm'}\big)^* \; Y_A^{lm} \big\rangle = \B_l \delta_{l'l} \delta_{m'm},&&&
 \\
 &\big\langle \big(X^A_{l'm'}\big)^* \; X_A^{lm} \big\rangle = \B_l \delta_{l'l} \delta_{m'm},\qquad&&
 \big\langle \big(X^A_{l'm'}\big)^* \; Y_A^{lm} \big\rangle = 0,&
 \\
& \big\langle \big(Y^{AB}_{l'm'}\big)^* \; Y_{AB}^{lm} \big\rangle = \frac{\A_l}{2} \delta_{l'l} \delta_{m'm},\qquad&&
 \big\langle \big(Y^{l'm'} \Omega^{AB}\big)^* \; Y_{AB}^{lm} \big\rangle = 0,&
 \\
 &\big\langle \big(X^{AB}_{l'm'}\big)^* \; X_{AB}^{lm} \big\rangle = \frac{\A_l}{2} \delta_{l'l} \delta_{m'm},\qquad&&
 \big\langle \big(X^{AB}_{l'm'}\big)^* \; Y_{AB}^{lm} \big\rangle = 0.&
\end{alignat*}
Note the consequent vanishing of $Y_A^{0m} = X_A^{0m} = Y_{AB}^{0m} =
X_{AB}^{0m} = 0$ and $Y_{AB}^{1m} = X_{AB}^{1m} = 0$. The tensorial
spherical harmonics also satisfy the following eigenvalue and divergence
identities:
\begin{alignat*}{3}
 &D_C D^C Y_{A} = -(\B_l-1) Y_{A},\qquad&&
 D^A Y_{A} = - \B_l Y,&
 \\
& D_C D^C X_{A} = -(\B_l-1) X_{A},\qquad&&
 D^A X_{A} = 0,&
 \\
 &D_C D^C Y_{AB} = -(\B_l-4) Y_{AB},\qquad&&
 D^A Y_{AB} = -\frac{\A_l}{2\B_l} Y_{B},&
 \\
& D_C D^C X_{AB} = -(\B_l-4) X_{AB},\qquad&&
 D^A X_{AB} = -\frac{\A_l}{2\B_l} X_B.&
\end{alignat*}
Under the parity inversion of the sphere, all the $Y$-harmonics
transform as $Y \mapsto (-)^l Y$ and so are called \emph{even}, while
all the $X$-harmonics transform as $X \mapsto (-)^{l+1} X$ and so are
called \emph{odd}~\cite{regge-wheeler}. Our conventions for spherical
harmonics follow~\cite{martel-poisson}, whose Appendix~A may be
consulted for comparison with other conventions in the literature.

When needed, we will also decompose sectorial tensors on the
$\big(\mathbb{R}^2, g\big)$ factor into coordinate components with respect to
the Schwarzschild coordinates $(t,r)$ with the conventions
\begin{alignat*}{3}
&v_b = v_t (\d t)_b + v_r (\d r)_b,\qquad&&
h_{cd} = h_{tt} (\d t)_c (\d t)_d + 2 h_{tr} (\d t)_{(c} (\d r)_{d)} + h_{rr} (\d r)_c (\d r)_d,&
\\
&v^a = v^t (\del_t)^a + v^r (\del_r)^a,\qquad&&
{h}^{ab} = {h}^{tt} (\del_t)^a (\del_t)^b + 2 {h}^{tr} (\del_t)^{(a} (\del_r)^{b)} + {h}^{rr} (\del_r)^a (\del_r)^b.&
\end{alignat*}
If we raise and lower indices with $g_{ab}$, this convention implies the
relations $v_t = -f v^t$, $v_r = f^{-1} v^r$, $h_{tt} = (-f)^2 h^{tt}$, etc.

\begin{rmk} \label{rmk:laurent-poles}
In the sequel, while working at the mode level, we allow ourselves the
freedom of dividing by the expressions $\omega$, $\B_l$ or $\A_l$ in the
mode parameters. The expressions for various radial mode operators will
be valid for any value of the mode parameters where they are
well-defined, which might exclude $\omega=0$, $l=0$ or $l=1$, when
dividing by $\omega$ or $\B_l$. These singular cases would need to be
investigated separately. However, since all expressions will be
rational, it would be sufficient to consider Laurent expansions in the
mode parameters and just compare corresponding coefficients.
\end{rmk}

\emph{Harmonic dependence} on time with frequency $\omega$ will refer to
being proportional to ${\rm e}^{-{\rm i}\omega t}$ with time independent
coefficients.

\begin{rmk} \label{rmk:dimensions}
We will use the convention that $r$ and $\omega^{-1}$ will carry the
dimensions of \emph{length}, $r^{-1}$, $\omega$ and $\del_r$ will carry
the dimensions of \emph{inverse length}, while $f$, $f_1$, $\B_l$ and
$\A_l$ will be dimensionless. Clearly, by this convention, length
dimension is additive under multiplication and composition of
differential operators. For practical convenience, upon radial mode
decomposition, we will try to make use of only dimensionless
differential operators, by for instance inserting extra factors of
$\omega$, where needed.
\end{rmk}

\begin{rmk}
Whenever we need to take a formal adjoint of a differential operator, we
will follow the conventions from Section~\ref{sec:adjoint}.
\end{rmk}

Before proceeding to mode decomposing concrete equations, let us make
some general remarks about the notation. If $s_{\mu\nu\ldots}$ is a
tensor field, then schematically we will denote its mode coefficients
(omitting indices) as
\begin{equation}
 s = \bm{s}^\text{even}\cdot Y + \bm{s}^\text{odd} \cdot X
 = \big(\mathfrak{s}^\text{even}\cdot
 Y + \mathfrak{s}^\text{odd}\cdot X\big) {\rm e}^{-{\rm i}\omega t}.
\end{equation}
The $\bm{s}$ coefficients represent the even and odd multiplets obtained
after the decomposition of a~spacetime tensor into sectors, absorbing
all angular dependence and spherical tensor indices into the spherical
harmonics $Y$ and $X$. The $\mathfrak{s}$ coefficients represent the
even and odd multiplets after decomposing the remaining sectorial
tensors in $\bm{v}$ into coordinate components with respect to the
Schwarzschild coordinates $(t,r)$ on $(\mathbb{R}^2,g)$, absorbing all
time dependence into the harmonic factor ${\rm e}^{-{\rm i}\omega t}$. When
convenient, the $\bm{s}$ and $\mathfrak{s}$ coefficients may be defined
with extra $r$-, $l$- and $\omega$-dependent normalization factors.

Given a spacetime differential operator $s' = O[s]$ (in general $s$ and
$s'$ need not be of the same tensor rank), which is even under inversion
of the $\big(S^2,\Omega\big)$ factor, we will schematically denote its action on
the mode coefficients as
\begin{alignat*}{3}
&\bm{s}'^{\text{even}} = O_e[\bm{s}^{\text{even}}],\qquad&&
 \mathfrak{s}'^{\text{even}} = O_e[\mathfrak{s}^{\text{even}}],&
\\
&\bm{s}'^{\text{odd}} = O_o[\bm{s}^{\text{odd}}],\qquad&&
 \mathfrak{s}'^{\text{odd}} = O_o[\mathfrak{s}^{\text{odd}}].&
\end{alignat*}
The $O_e$ and $O_o$ operators acting on the $\mathfrak{s}$ coefficients
are precisely the \emph{radial mode operator} versions of the original
spacetime operator $O$ that we will eventually be working with.

In principle, the intermediate stage of the decomposition involving
$\bm{s}$ and $\bm{s}'$ can be omitted from the presentation. But we
believe that it is worth keeping them, because these expressions would
be useful for obtaining the explicit radial mode equations for other
coordinate systems on the $\big(\mathbb{R}^2, g\big)$ factor, for instance the
Eddington--Finkelstein coordinates instead of the Schwarzschild one. In
addition, these intermediate expressions would be useful checkpoints for
anyone trying to reproduce our calculations.

\subsection{Scalar wave equation} \label{sec:sw}

Before proceeding to tensor wave equations, for illustration purposes,
let us first handle the simplest case of the scalar wave equation
\begin{subequations}
\begin{equation} \label{eq:sw-def}
 z' = \sqzer z := \dalf z = 0.
\end{equation}
Following earlier conventions, we use the notation $z = \bm{z} Y =
\mathfrak{z} Y {\rm e}^{-{\rm i}\omega t}$ and $z' = \frac{\bm{z}'}{r^2} Y =
\frac{\mathfrak{z}'}{r^2} Y {\rm e}^{-{\rm i}\omega t}$. First, we separate out $S^2$
dependence via spherical harmonics:
\begin{equation} \label{eq:sw-covar} \noeqref{eq:sw-covar}
 \frac{1}{r^2} \, \bm{z}' \, Y
 = \sqzer \big(\bm{z}(t,r) Y\big)
 = \big(\sqzer \bm{z}(t,r)\big) Y,
 \qquad \text{with} \quad
 \bm{z}' = \sqzer \bm{z}(t,r) := \big(\nabla^a r^2\nabla_a - \B_l\big) \bm{z}.
\end{equation}
Then, we separate out the $t$ coordinate via harmonic time dependence:
\begin{gather} \label{eq:sw-coord}
 \mathfrak{z}' {\rm e}^{-{\rm i}\omega t}
 = \sqzer \big(\mathfrak{z}(r) {\rm e}^{-{\rm i}\omega t}\big)
 = \big(\sqzer \mathfrak{z}(r)\big) {\rm e}^{-{\rm i}\omega t},
 \qquad 
 \mathfrak{z}' = \sqzer \mathfrak{z}(r)
 := \bigg(\del_r f r^2 \del_r + \frac{\omega^2 r^2}{f} - \B_l\bigg) \mathfrak{z}.
\end{gather}
\end{subequations}
The last formula defines the \emph{radial mode ODE} for the scalar wave
equation~\eqref{eq:sw-def}. Note that, to economize notation, we have
reused the symbol $\sqzer$, though it changes meaning when acting on
different function spaces, which should be clear from context.

Note that the final expression for $\sqzer \mathfrak{z}(r)$ is formally
self-adjoint, $\sqzer^* = \sqzer$. That is no accident, since we have
chosen our definitions such that the sesquilinear forms (with compactly
supported $\mathfrak{y}(r)$ and $\mathfrak{z}(r)$)
\begin{align}
 \int_{\mathbb{R}^2} \eps \, r^2 \int_{S^2} \epsilon \,
 \mathfrak{y}(r)^* (Y^{l'm'})^* {\rm e}^{{\rm i}\omega' t} \sqzer \mathfrak{z}(r) Y^{lm} {\rm e}^{-{\rm i}\omega t}
 &= \delta_{l'l} \delta_{m'm}
 \int_{\mathbb{R}^2} \eps \, \mathfrak{y}(r)^* {\rm e}^{{\rm i}\omega' t} \sqzer \mathfrak{z}(r) {\rm e}^{-{\rm i}\omega t}
\\
 &= 2\pi \delta(\omega'-\omega) \delta_{l'l} \delta_{m'm}
 \int_{2M}^\oo \d{r} \, \mathfrak{y}(r)^* \sqzer \mathfrak{z}(r)
\end{align}
all agree and are symmetric up to complex conjugation, because the
original $\sqzer$ operator is formally self-adjoint on spacetime. We
will follow the same convention in the rest of the section. This means
that some of the components of the radial mode equations in the tensor
case might have spurious looking factors of $\B_l$ or $\A_l$. Such
factors arise from absorbing the norms of the corresponding tensor
harmonics and are needed to maintain manifest formal self-adjointness of
the radial mode ODEs.

The transformation
\begin{gather}
 \phi_0 = -{\rm i}\omega r \, \mathfrak{z},\qquad
 \mathfrak{z} = -\frac{1}{{\rm i}\omega r} \phi_0,
\end{gather}
puts $\sqzer \mathfrak{z} = 0$ in direct equivalence with the $s=0$
Regge--Wheeler equation $\D_0 \phi_0 = 0$:
\begin{gather} \label{eq:sw-triang}
\sqzer \bigg({-}\frac{1}{{\rm i}\omega r}\bigg) \phi_0
= ({\rm i}\omega r) \frac{1}{\omega^2} \D_0 \phi_0,\qquad
\frac{1}{\omega^2} \D_0 (-{\rm i}\omega r) \mathfrak{z}
= \bigg(\frac{1}{{\rm i}\omega r}\bigg) \sqzer \mathfrak{z}.
\end{gather}
In the sense of Section~\ref{sec:formal} we have an equivalence of
$\sqzer$ with $\D_0$, exhibited by the diagram
\begin{equation} \label{eq:sw-equiv}
\scalebox{1.2}{\begin{tikzcd}[column sep=huge,row sep=huge]
 \bullet
 \ar[swap]{d}[description]{\sqzer}
 \ar[shift left]{r}{k_{0}=-{\rm i}\omega r}
 \&
 \bullet
 \ar[shift left]{l}{\bar{k}_{0}=-\frac{1}{{\rm i}\omega r}}
 \ar{d}[description]{\tilde{\sqzer}}
 \\
 \bullet
 \ar[shift left]{r}{k'_{0}=\frac{1}{{\rm i}\omega r}}
 \ar[dashed,bend left=45]{u}{h_0 = 0}
 \&
 \bullet
 \ar[shift left]{l}{\bar{k}'_{0}={\rm i}\omega r}
 \ar[dashed,bend right=45]{u}[swap]{\tilde{h}_0 = 0,}
\end{tikzcd}}
 \qquad \text{where} \quad
 \tilde{\sqzer} := \frac{1}{\omega^2} \D_0 \,.
\end{equation}

Notice that the original operator $\sqzer$ is dimensionless. We have
chosen to include various factors of $\omega$ to keep all the other
operators in the diagram dimensionless as well. We will follow the same
convention throughout all further calculations.

\begin{rmk}
Note also that the above equivalence diagram is \emph{self-adjoint} in
the following sense: $\sqzer^* = \sqzer$, $\tilde{\sqzer}^* =
\tilde{\sqzer}$, $h_0^* = h_0$ and $\tilde{h}_0^* = \tilde{h}_0$, while
$\bar{k}_0 = (k'_0)^*$ and $\bar{k}'_0 = (k_0)^*$. That is, taking
formal adjoints of all operators in the diagram, we obtain the same
diagram. The various factors of ${\rm i}$ and $\omega$ were chosen also to
exhibit this property. In particular, the powers of $\omega$ in $k_0$,
$k'_0$ and $\tilde{\sqzer}$ are not all independent. There is an obvious
economy to self-adjoint equivalence diagrams: the morphism need only be
specified in one direction, with the inverse morphism recovered by
taking adjoints. In subsequent cases, we have succeeded in identifying a
self-adjoint equivalences, but with slightly modified notion of
self-adjointness for the decoupled triangular form. It is not obvious to
us at the moment why we were successful in that sense and whether that
is due to a~fortuitous coincidence or a general pattern.
\end{rmk}

\subsection{Vector wave equation} \label{sec:vw}

To present the radial mode equation for $\sqone v_\mu = 0$, we follow
the same pattern as in the case of the scalar wave equation
(Section~\ref{sec:sw}).

The tensor spherical harmonic decomposition of a covariant vector field
is
\begin{gather}
v_\mu = v_\mu^{\text{even}} + v_\mu^{\text{odd}}, \qquad
v_\mu^{\text{even}}\to \sum_{lm}
\begin{bmatrix}\bm{v}_a^{lm}\, Y^{lm} \\\bm{u}^{lm}\, r Y_A^{lm}\end{bmatrix}\!,\qquad
v_\mu^{\text{odd}}\to \sum_{lm}
\begin{bmatrix} 0 \\ \bm{w}^{lm}\, r X_A^{lm} \end{bmatrix} \!.
\end{gather}
For conciseness, as before, we will omit the $lm$ indices. Following the
notation conventions given at the top of Section~\ref{sec:schw}, we
first separate out the spherical harmonics (bold coefficients) and then
harmonic time dependence (fraktur coefficients). We will need to use two
conventions to normalize the coefficients, which we distinguish by
primes (recall that primes are decorations and not derivatives):
\begin{equation}\label{eq:v-modsep}
\begin{aligned}
 v_\nu
 &\to \begin{bmatrix} \bm{v}_b Y \\ \bm{u}\, r Y_B \end{bmatrix}
 + \begin{bmatrix} 0 \\ \bm{w}\, r X_B \end{bmatrix}
 = \left(\begin{bmatrix} \mathfrak{v}_b Y \\ \mathfrak{u}\, r Y_B \end{bmatrix}
 + \begin{bmatrix} 0 \\ \mathfrak{w}\, r X_B \end{bmatrix} \right) {\rm e}^{-{\rm i}\omega t},
 \\
 v^{\prime\mu}
 &\to \frac{1}{r^2} \left(\begin{bmatrix} \bm{v}^{\prime a} Y \\ \tfrac{\bm{u}'}{\B_l}\, r^{-1} Y^A \end{bmatrix}
 + \begin{bmatrix} 0 \\ \tfrac{\bm{w}'}{\B_l}\, r^{-1} X^A \end{bmatrix} \right)
 = \frac{1}{r^2} \left(\begin{bmatrix} \mathfrak{v}^{\prime a} Y \\ \tfrac{\mathfrak{u}'}{\B_l}\, r^{-1} Y^A \end{bmatrix}
 + \begin{bmatrix} 0 \\ \tfrac{\mathfrak{w}'}{\B_l}\, r^{-1} X^A \end{bmatrix} \right) {\rm e}^{-{\rm i}\omega t}.
\end{aligned}
\end{equation}
\begin{subequations}
With these conventions, again reusing the symbol $\sqone$ at each stage
of mode separation of $v'_\mu = \sqone v_\mu$, we have
\begin{gather}
\begin{bmatrix} \bm{v}^{\prime a} \\ \bm{u}' \end{bmatrix}
= \sqone_e \begin{bmatrix} \bm{v}_b \\ \bm{u} \end{bmatrix}\notag
:= \begin{bmatrix} \nabla_m r^2 \nabla^m \bm{v}^a \\ \nabla_m \B_l r^2 \nabla^m \bm{u} \end{bmatrix}
- \B_l \begin{bmatrix} \bm{v}^a \\ \B_l \bm{u}\end{bmatrix}
\\ \hphantom{\begin{bmatrix} \bm{v}^{\prime a} \\ \bm{u}' \end{bmatrix}=}
{}+ \begin{bmatrix} -2 r^a r^b & 2\B_l r^a \\ 2\B_l r^b & 0 \end{bmatrix}
\begin{bmatrix} \bm{v}_b \\ \bm{u} \end{bmatrix}
+ (1-f)\begin{bmatrix} 0 \\ \B_l \bm{u} \end{bmatrix}= 0,
\label{eq:vwe-covar}\noeqref{eq:vwe-covar}
\\[1ex]
\label{eq:vwo-covar} \noeqref{eq:vwo-covar}
\bm{w}' =\sqone_o w :=
\nabla_m \B_l r^2 \nabla^m \bm{w} - \B_l (\B_l \bm{w}) + \B_l (1-f) \bm{w} = 0,
\end{gather}
after separating out the $S^2$ dependence, and
\begin{gather}
\begin{bmatrix} \mathfrak{v}^{\prime t} \\ \mathfrak{v}^{\prime r} \\ \mathfrak{u}'\end{bmatrix}
= \sqone_e
\begin{bmatrix} \mathfrak{v}_t \\ \mathfrak{v}_r \\ \mathfrak{u} \end{bmatrix}\notag
:= \begin{bmatrix}[r@{\,}l] -\del_r f^{-1} r^2 f\del_r &\mathfrak{v}_t \\
\del_r f r^2 f \del_r &\mathfrak{v}_r \\ \del_r \B_l r^2 f\del_r &\mathfrak{u} \end{bmatrix}
+ \bigg( \frac{\omega^2}{f} - \frac{\B_l}{r^2} \bigg)
\begin{bmatrix}
[r@{\,}l] -f^{-1} r^2 &\mathfrak{v}_t \\ f r^2 &\mathfrak{v}_r \\ \B_l r^2 &\mathfrak{u} \end{bmatrix}
\\ \hphantom{\begin{bmatrix} \mathfrak{v}^{\prime t} \\ \mathfrak{v}^{\prime r} \\ \mathfrak{u}'\end{bmatrix}=}
{}+ {\rm i}\omega r \frac{f_1}{f}
\begin{bmatrix} 0 & 1 & 0 \\ -1 & 0 & 0 \\ 0 & 0 & 0\end{bmatrix}
\begin{bmatrix} \mathfrak{v}_t \\ \mathfrak{v}_r \\ \mathfrak{u} \end{bmatrix}
+\begin{bmatrix} 0 & 0 & 0 \\ 0 & -2 f^2 & 2\B_l f \\ 0 & 2\B_l f & \B_l f_1 \end{bmatrix}
\begin{bmatrix} \mathfrak{v}_t \\ \mathfrak{v}_r \\ \mathfrak{u} \end{bmatrix}= 0, \label{eq:vwe-coord}
\\[1ex]
\label{eq:vwo-coord}
\mathfrak{w}' = \sqone_o \mathfrak{w}:= \del_r \B_l r^2 f \del_r \mathfrak{w}
+ \bigg(\frac{\omega^2}{f} - \frac{\B_l}{r^2}\bigg) \B_l r^2 \mathfrak{w}
 + \B_l (1-f) \mathfrak{w}= 0,
\end{gather}
after separating out the time dependence.
Our normalization conventions for the mode coefficients were chosen
precisely to maintain manifest formal self-adjointness, $\sqone_{\rm e}^* =
\sqone_e$ and $\sqone_o^* = \sqone_o$, at each stage of the mode
separation.
\end{subequations}

\begin{rmk}
Note also that these expressions are valid for all values of $\omega$
and $l$. When $l=0$, $Y_A^{l=0} = X_A^{l=0} = 0$, meaning that their
coefficients are not well-defined. By consistency, the corresponding
components of $\sqone_e$ and $\sqone_o$ simply vanish.
\end{rmk}

Next, we will introduce the operators needed to decouple radial modes
into \emph{pure gauge} and \emph{constraint violating} modes, to set up
the decoupling strategy (cf.~Section~\ref{sec:triang-strategy}). The key
identities, analogous to $E_0 T = T' E_1$ and $E_1 D = D' E_0$, hold
already at the spacetime level:
\begin{equation}
 \sqzer \, \grf^\mu v_\mu = \grf^\mu \, \sqone v_\mu
 \qquad \text{and} \qquad
 \sqone \, \grf_\mu z = \grf_\mu \, \sqzer z \,.
\end{equation}
In addition, analogous to $T D = H_T E_0$, we have the composition identity
\begin{equation}
\grf^\mu \grf_\mu z = \sqzer z.
\end{equation}

Now we define the operators $T_1$, $T'_1$, $D_1$, $D'_1$ and their
successive mode decompositions:

\noindent
for $z = \grf^\mu v_\mu =: T_1[v_\nu ]$,
\begin{gather*}
 \bm{z} = T_{1e}\begin{bmatrix}[c@{~}c] \bm{v}_b \\ \bm{u} \end{bmatrix}\!
 := \begin{bmatrix}
 \frac{1}{r^2} \nabla^b r^2 &
 -\frac{\B_l}{r}
 \end{bmatrix} \ \begin{bmatrix} \bm{v}_b \\ \bm{u} \end{bmatrix}\!,
 \quad\,
 \frac{\mathfrak{z}}{{\rm i}\omega} = T_{1e}\begin{bmatrix} \mathfrak{v}_t \\ \mathfrak{v}_r \\ \mathfrak{u} \end{bmatrix}\!
 := \begin{bmatrix} \frac{1}{f} &
 \frac{1}{{\rm i}\omega r^2} \del_r f r^2 &
 -\frac{\B_l}{{\rm i}\omega r} \end{bmatrix} \!
\begin{bmatrix} \mathfrak{v}_t \\ \mathfrak{v}_r \\ \mathfrak{u} \end{bmatrix}\!,\\
 T_{1o} := 0;
\end{gather*}
for $z' = \grf_\mu v^{\prime\mu} =: T'_1[v']$,
\begin{gather*}
\bm{z}' = T'_{1e}\begin{bmatrix} \bm{v}^{\prime b} \\ \bm{u}' \end{bmatrix}
:= \frac{1}{r}\begin{bmatrix}r\nabla_b &-1\end{bmatrix}\!
\begin{bmatrix} \bm{v}^{\prime b} \\ \bm{u}' \end{bmatrix}\!,
\qquad
\frac{\mathfrak{z}'}{{\rm i}\omega} = T'_{1e}\begin{bmatrix} \mathfrak{v}^{\prime t} \\ \mathfrak{v}^{\prime r} \\ \mathfrak{u}' \end{bmatrix}
:= \begin{bmatrix} -1 & \frac{1}{{\rm i}\omega}\del_r & -\frac{1}{{\rm i}\omega r} \end{bmatrix}\!
\begin{bmatrix} \mathfrak{v}^{\prime t} \\ \mathfrak{v}^{\prime r} \\ \mathfrak{u}' \end{bmatrix}\!,
\\
T'_{1o} := 0 ;
\end{gather*}
for $v_\mu = \grf_\mu z =: D_1[z]$,
\begin{gather*}
\begin{bmatrix} \bm{v}_a \\ \bm{u} \end{bmatrix}
= D_{1e}[\bm{z}] := \frac{1}{r} \begin{bmatrix} r\, \nabla_a \\ 1 \end{bmatrix} \bm{z},
\qquad
\frac{1}{{\rm i}\omega} \begin{bmatrix} \mathfrak{v}_t \\ \mathfrak{v}_r \\ \mathfrak{u} \end{bmatrix}
= D_{1e}[\mathfrak{z}] := \begin{bmatrix} -1 \\ \tfrac{1}{{\rm i}\omega} \del_r \\[.5ex]
\tfrac{1}{{\rm i}\omega r} \end{bmatrix} \mathfrak{z},\qquad
D_{1o}:= 0 ;
\end{gather*}
for $v'_\mu = \grf_\mu z' =: D'_1[z']$,
\begin{gather*}
\begin{bmatrix} \bm{v}^{\prime a} \\ \bm{u}' \end{bmatrix}
= D'_{1e}[\bm{z}'] :
= \begin{bmatrix} r^2 \nabla^a \tfrac{1}{r^2} \\[.5ex] \tfrac{\B_l}{r} \end{bmatrix} \bm{z}',\qquad
\frac{1}{{\rm i}\omega} \begin{bmatrix} \mathfrak{v}^{\prime t} \\ \mathfrak{v}^{\prime r} \\ \mathfrak{u}'\end{bmatrix}
= D'_{1e}[\mathfrak{z}']
:= \begin{bmatrix}\tfrac{1}{f} \\[1ex] f r^2 \del_r \tfrac{1}{{\rm i}\omega r^2} \\[1ex]
\tfrac{\B_l}{{\rm i}\omega r} \end{bmatrix} \mathfrak{z}',\qquad
D'_{1o} := 0.
\end{gather*}

The extra ${\rm i}\omega$ factors in the normalizations have been
chosen so that all operators acting on radial modes are dimensionless
and the first two of the following commutative diagrams remain valid at
each stage of mode separation:
\begin{equation} \label{eq:vw-input}
\begin{tikzcd}[column sep=huge,row sep=huge]
 \bullet \ar[swap]{d}{\D_{D_1} = \sqzer} \ar{r}{D_1} \&
 \bullet \ar{d}{\sqone}
 \\
 \bullet \ar[swap]{r}{D'_1} \&
 \bullet
\end{tikzcd} \!,
 \qquad
\begin{tikzcd}[column sep=huge,row sep=huge]
 \bullet \ar[swap]{d}{\sqone} \ar{r}{T_1} \&
 \bullet \ar{d}{\D_{T_1} = \sqzer}
 \\
 \bullet \ar[swap]{r}{T'_1} \&
 \bullet
\end{tikzcd} \!,
 \qquad
\begin{tikzcd}[column sep=huge,row sep=huge]
 \bullet \ar[swap]{d}{\begin{bmatrix}\sqone \\ T_1\end{bmatrix}} \ar{r}{\Phi_1} \&
 \bullet \ar{d}{\D_{\Phi_1}}
 \\
 \bullet \ar[swap]{r}{\begin{bmatrix}\Phi'_1 & -\Delta_{\Phi_1 T_1}\end{bmatrix}} \&
 \bullet
\end{tikzcd}\! \!.
\end{equation}
The last diagram holds at the radial mode level and shows the decoupling
of \emph{gauge invariant} modes (cf.~Section~\ref{sec:triang-strategy}).
It decomposes into independent even and odd sectors as follows:
\begin{gather}
\Phi_{1e} := \begin{bmatrix} 0 -f f\del_r r \end{bmatrix} \!,\qquad
\Phi_{1o} := -{\rm i}\omega r,
\\
\Phi'_{1e} := \frac{1}{r^2} \begin{bmatrix} 0 -\B_l r^2\del_r \frac{f}{r} \end{bmatrix} \!,\qquad
\Phi'_{1o} := \frac{1}{{\rm i}\omega r},
\\
\D_{\Phi_{1e}} := \frac{\B_l}{\omega^2} \D_1,\qquad
\D_{\Phi_{1o}} := \frac{\B_l}{\omega^2} \D_1,
\\
\Delta_{\Phi_{1e} T_{1e}} := -\frac{\B_l f_1}{{\rm i}\omega r},\qquad
\Delta_{\Phi_{1o} T_{1o}} := 0.
\end{gather}

We are now ready to give the full triangular decoupling of the even and
odd sectors of the vector wave equation $\sqone v_\mu = 0$.

\subsubsection{Odd sector} \label{sec:vwod}
This sector is structurally similar to the scalar wave case
(Section~\ref{sec:sw}). It is particularly simple; since it does not
contain any gauge modes, it is not constrained by the harmonic gauge
condition and the operators $\Phi_{1o}$, $\Phi'_{1o}$ are directly
invertible. Thus, the transformation
\begin{gather}
\phi_1 = -{\rm i}\omega r w,\qquad
w = -\frac{1}{{\rm i}\omega r} \phi_1,
\end{gather}
puts $\sqone_o w = 0$ in direct equivalence with the $s=1$ Regge--Wheeler
equation $\D_1 \phi_1 = 0$:
\begin{gather} \label{eq:vwo-final-triang}
\sqone_o \bigg({-}\frac{1}{{\rm i}\omega r}\bigg) \phi_1
= ({\rm i}\omega r) \frac{\B_l}{\omega^2} \D_1 \phi_1,\qquad
\frac{\B_l}{\omega^2} \D_1 (-{\rm i}\omega r) w
= \bigg(\frac{1}{{\rm i}\omega r}\bigg) \sqone_o w.
\end{gather}
In diagrammatic form, we have
\begin{equation} \label{eq:vw1o-equiv}
\scalebox{1.2}{
\begin{tikzcd}[column sep=huge,row sep=huge]
 \bullet
 \ar[swap]{d}[description]{\sqone_o}
 \ar[shift left]{r}{k_{1o}=-{\rm i}\omega r}
 \&
 \bullet
 \ar[shift left]{l}{\bar{k}_{1o}=-\frac{1}{{\rm i}\omega r}}
 \ar{d}[description]{\tilde{\sqone}_o}
 \\
 \bullet
 \ar[shift left]{r}{k'_{1o}=\frac{1}{{\rm i}\omega r}}
 \ar[dashed,bend left=45]{u}{0}
 \&
 \bullet
 \ar[shift left]{l}{\bar{k}'_{1o}={\rm i}\omega r}
 \ar[dashed,bend right=45]{u}[swap]{0,}
\end{tikzcd}}
 \qquad \text{where} \quad
 \tilde{\sqone}_o := \frac{\B_l}{\omega^2} \D_1 .
\end{equation}
This equivalence diagram is \emph{self-adjoint} in the sense explained
in Section~\ref{sec:sw} for the scalar wave equation.

\subsubsection{Even sector} \label{sec:vwev}
This sector is more complicated and is the first prototype for the
abstract approach to triangular decoupling that we outlined earlier in
Section~\ref{sec:formal}. All that we need to feed into it are the even
sectors of the diagrams~\eqref{eq:vw-input} and the additional
composition identities
\begin{equation}
 T_{1e} D_{1e} = \underbrace{-\frac{1}{\omega^2 r^2}}_{H_{T_{1e}}} \sqzer
 \qquad \text{and} \qquad
 \Phi_{1e} D_{1e} = \underbrace{~ 0 ~}_{H_{\Phi_{1e}}} \sqzer.
\end{equation}
Since the implementation of our approach in this particular case was
explained in detail in~\cite{kh-vwtriang} (in particular, see the
step-by-step discussion there in Section~3.2), we only state the final
result, which has only been slightly adjusted for the purposes of this
work (the end of the section shows how). The final decoupled triangular
form is
\begin{equation}
\sqone_e \begin{bmatrix}\mathfrak{v}_t \\ \mathfrak{v}_r \\ \mathfrak{u} \end{bmatrix} = 0
 \iff
\underbrace{{\frac{1}{\omega^2}} \begin{bmatrix}
\D_0 & 0 & -\tfrac{f_1}{r^2} \big(\B_l+\tfrac{1}{2}f_1\big) \\
0 & \B_l \D_1 & 0 \\ 0 & 0 & \D_0 \end{bmatrix}}_{\tilde{\sqone}_e}
\begin{bmatrix} \phi_0 \\ \phi_1 \\ \phi'_0 \end{bmatrix} = 0.
\end{equation}

While $\tilde{\sqone}_e$ is obviously not formally self-adjoint, its
equivalence with the formally self-adjoint $\sqone_e$ survives in the
existence of an operator $\Sigma_{1e}$ effecting the equivalence between
$\tilde{\sqone}_e$ and $\tilde{\sqone}_{\rm e}^*$, namely $\tilde{\sqone}_e
\Sigma_{1e} = \Sigma_{1e} \tilde{\sqone}_{\rm e}^*$, where
\begin{gather}
 \Sigma_{1e} = \begin{bmatrix}
 0 & 0 & 1 \\
 0 & 1 & 0 \\
 1 & 0 & 0
 \end{bmatrix}\!.
\end{gather}

The precise equivalence identities take the form
\begin{gather*} \label{eq:vw1e-equiv}
\scalebox{1.2}{\begin{tikzcd}[column sep=huge,row sep=huge]
 \bullet
 \ar[swap]{d}[description]{\sqone_e}
 \ar[shift left]{r}{k_{1e}}
 \&
 \bullet
 \ar[shift left]{l}{\bar{k}_{1e}}
 \ar{d}[description]{\tilde{\sqone}_e}
 \\
 \bullet
 \ar[shift left]{r}{k'_{1e}}
 \ar[dashed,bend left=40]{u}{h_{1e}}
 \&
 \bullet
 \ar[shift left]{l}{\bar{k}'_{1e}}
 \ar[dashed,bend right=40]{u}[swap]{\tilde{h}_{1e}}
\end{tikzcd}}\!,
\end{gather*}
where
\begin{gather*}
 \bar{k}_{1e} =
 \begin{bmatrix}
-{\rm i}\omega r & {\rm i}\omega r \B_l & \tfrac{1}{2} {\rm i}\omega r (\B_l+f)
\\[.5ex]
\del_r r & -\tfrac{\B_l}{r} \del_r r^2 & -\tfrac{1}{2}\big((\B_l-f) \tfrac{1}{r^2} \del_r r^3 + 2f+f_1\big)
\\[1ex]
1 & -\tfrac{f}{r}\del_r r^2 - \B_l & -\tfrac{1}{2} (2f\del_r r + \B_l + f)
\end{bmatrix} \frac{1}{\omega^2 r^2},
\\[1ex]
\bar{k}'_{1e} =
\begin{bmatrix}
\tfrac{{\rm i}\omega r}{f} & -\tfrac{{\rm i}\omega r}{f} & -\tfrac{{\rm i}\omega r (\B_l+f)}{2f}
\\[1ex]
f r^2 \del_r \tfrac{1}{r} & -f r \del_r & -\tfrac{1}{2} f ((\B_l-f) \del_r r + 2f - f_1)
\\[.5ex]
\B_l & -r\del_r f - \B_l & -\tfrac{1}{2} \B_l (2 r \del_r f + \B_l - f)
 \end{bmatrix}\!,
 \\[1ex]
 h_{1e} = \frac{1}{\omega^2 r^2}
 \begin{bmatrix}
0 & 0 & 0 \\ 0 & 1 & 0 \\ 0 & 0 & -\tfrac{f}{\B_l}
 \end{bmatrix}\!,
 \qquad
 \tilde{h}_{1e} =
 \begin{bmatrix}
-\tfrac{1}{2}(\B_l-f) & \tfrac{1}{2}(\B_l+f) & \tfrac{1}{4}(\B_l+f)^2
\\[1ex]
-1 & \tfrac{\B_l+f}{\B_l} & \tfrac{1}{2}(\B_l+f)
\\ [1ex]
1 & -1 & -\tfrac{1}{2}(\B_l-f)
 \end{bmatrix}\!,
\end{gather*}
while the remaining operators can be recovered from the identities
\begin{equation}
k_{1e} = \Sigma_{1e} (\bar{k}'_{1e})^*, \qquad
k'_{1e} = \Sigma_{1e} (\bar{k}_{1e})^*,
\end{equation}
where also
\begin{equation}
h_{1e}^* = h_{1e}, \qquad
\tilde{h}_{1e}^* = \Sigma_{1e} \tilde{h}_{1e} \Sigma_{1e}.
\end{equation}

\begin{rmk} \label{rmk:twist-selfadj}
These last identities embody the modified sense in which the equivalence
diagram~\eqref{eq:vw1e-equiv} is self-adjoint, in the sense discussed
earlier in Section~\ref{sec:sw}, but with a twist provided by the matrix
$\Sigma_{1e}$. In fact, if we replace $\tilde{\sqone}_e$ by $\Sigma_{1e}
\tilde{\sqone}_e$, it is no longer upper triangular, but it is formally
self-adjoint
\begin{equation}
 \big(\Sigma_{1e} \tilde{\sqone}_e\big)^* = \Sigma_{1e} \tilde{\sqone}_e,
\end{equation}
and the twist by $\Sigma_{1e}$ is no longer necessary. Since we place
importance on the upper triangular form, we will not take advantage of
this replacement.
\end{rmk}

The above upper triangular form of $\tilde{\sqone}_e$ allows us to
classify all of its symmetries (or automorphisms in the sense of
Section~\ref{sec:formal}). The key result is the absence of
non-vanishing morphisms between the Regge--Wheeler equations $\D_{s_0}$
and $\D_{s_1}$, except the identity morphism when $s_0 = s_1$
(Theorem~\ref{thm:rw-maps}). This prevents the coupling of the $\D_0$
and $\D_1$ blocks by an automorphism. Also, the single non-vanishing
(and non-removable) off-diagonal element in $\tilde{\sqone}_e$ prevents
the exchange of the order of the $\D_0$ blocks. Hence, any automorphism
of $\tilde{\sqone}_e$ must also be upper triangular, with every
non-vanishing matrix element proportional to the identity.

{\samepage
With the above logic in mind, the most general automorphism takes the
form $\tilde{\sqone}_e A = A \tilde{\sqone}_e$, with
\begin{equation}
 A = \begin{bmatrix}
 a_1 & 0 & b_1 \\
 0 & a_2 & 0 \\
 0 & 0 & a_1
 \end{bmatrix}\!,
\end{equation}
parametrized by the $3$ constants $a_1$, $b_1$, $a_2$. It is invertible when
$a_1,a_2\ne 0$.

}

With the above choice of $k_{1e}$ and $k'_{1e}$, up to homotopy, we have
the following equivalences of operators:
\begin{gather}
 k_{1e} \circ D_{1e} \circ \bar{k}_{0}
 \sim \begin{bmatrix} 1 \\ 0 \\ 0 \end{bmatrix}\!,
\qquad
 k_{0} \circ T_{1e} \circ \bar{k}_{1e}
 \sim \begin{bmatrix} 0 & 0 & 1 \end{bmatrix}\!.
\end{gather}

Finally, for the record, denoting by $\big(\tilde{\sqone}_e\big)^\text{old}$,
$\big(\bar{k}_{1e}\big)^\text{old}$ and $\big(\bar{k}'_{1e}\big)^\text{old}$ by
corresponding operators from~\cite{kh-vwtriang}, let us note the
explicit relations
\begin{gather}
 \tilde{\sqone}_e = \begin{bmatrix}
 1 & 0 & 0 \\
 0 & \B_l & 0 \\
 0 & 0 & 1
 \end{bmatrix} (\tilde{\sqone}_e)^\text{old},
 \quad\
 \bar{k}_{1e} = (\bar{k}_{1e})^\text{old} \begin{bmatrix}
 1 & 0 & \tfrac{1}{2}\B_l \\
 0 & 1 & 0 \\
 0 & 0 & 1
 \end{bmatrix}\!,
 \quad\
 \bar{k}'_{1e} = (\bar{k}'_{1e})^\text{old} \begin{bmatrix}
 1 & 0 & \tfrac{1}{2}\B_l \\
 0 & \tfrac{1}{\B_l} & 0 \\
 0 & 0 & 1
 \end{bmatrix}\!.
\end{gather}

\subsection{Lichnerowicz wave equation} \label{sec:lich}
To present the radial mode equation for $\sqtwo p_{\mu\nu} = 0$, we
follow the same pattern as in the case of the scalar and vector wave
equations (Sections~\ref{sec:sw} and~\ref{sec:vw}).

The tensor spherical harmonic decomposition of a symmetric covariant
$2$-tensor is
\begin{gather}\label{eq:pmodes}
p_{\mu\nu}
= p_{\mu\nu}^{\text{even}} + p_{\mu\nu}^{\text{odd}},
\\
p_{\mu\nu}^{\text{even}}\to \sum_{lm}
\begin{bmatrix}
 \bm{h}_{ab}^{lm} Y^{lm} & r \, \bm{j}_a^{lm} Y_B^{lm} \\
 r \, \bm{j}_b^{lm} Y_A^{lm} & r^2 \,
 (\bm{K}^{lm} \Omega_{AB} Y + \bm{G}^{lm} Y_{AB}^{lm})
\end{bmatrix}\!,
\\
p_{\mu\nu}^{\text{odd}}\to \sum_{lm}
\begin{bmatrix}
 0 & r \, \bm{h}_a^{lm} X_B^{lm} \\
 r \, \bm{h}_b^{lm} X_A^{lm} & r^2 \, \bm{h}_2^{lm} X_{AB}^{lm}
\end{bmatrix}\!.
\end{gather}
For conciseness, as before, we will omit the $lm$ indices. Following the
notation conventions given at the top of Section~\ref{sec:schw}, we
first separate out the spherical harmonics (bold coefficients) and then
harmonic time dependence (fraktur coefficients). We will need to use two
conventions to normalize the coefficients, which we distinguish by
primes (recall that primes are decorations and not derivatives):
\begin{gather}\label{eq:p-modsep}
p_{\la\ka}\to
\begin{bmatrix}
\bm{h}_{cd}\, Y & \bm{j}_c\, r Y_D \\
\bm{j}_d\, r Y_C & r^2 (\bm{K}\, \Omega_{CD} Y + \bm{G}\, Y_{CD})
\end{bmatrix}
+ \begin{bmatrix}
0 & \bm{h}_c\, r X_D \\ \bm{h}_d\, r X_C & \bm{h}_2\, r^2 X_{CD}
\end{bmatrix}
\\ \hphantom{p_{\la\ka}}
{}= \left(\begin{bmatrix}
\mathfrak{h}_{cd}\, Y & \mathfrak{j}_c\, r Y_D \\
\mathfrak{j}_d\, r Y_C & r^2 (\mathfrak{K}\, \Omega_{CD} Y + \mathfrak{G}\, Y_{CD})
\end{bmatrix}
+ \begin{bmatrix}
0 & \mathfrak{h}_c\, r X_D \\ \mathfrak{h}_d\, r X_C & \mathfrak{h}_2\, r^2 X_{CD}
\end{bmatrix} \right) {\rm e}^{-{\rm i}\omega t},
\\[1ex]
p^{\prime\mu\nu}\to \frac{1}{r^2} \left(
\begin{bmatrix}
\bm{h}^{\prime ab}\, Y & \tfrac{\bm{j}^{\prime a}}{\B_l}\, \tfrac{1}{r} Y^B \\[1ex]
\tfrac{\bm{j}^{\prime b}}{\B_l}\, \tfrac{1}{r} Y^A & \tfrac{1}{r^2} (\bm{K}'\, \Omega^{AB} Y + 2 \tfrac{\bm{G}'}{\A_l}\, Y^{AB})
\end{bmatrix}
+
\begin{bmatrix}
0 & \tfrac{\bm{h}^{\prime a}}{\B_l}\, \tfrac{1}{r} X^B \\[1ex]
\tfrac{\bm{h}^{\prime b}}{\B_l}\, \tfrac{1}{r} X^A & 2\tfrac{\bm{h}'_2}{\A_l}\, \tfrac{1}{r^2} X^{AB}
\end{bmatrix} \right)
\\ \hphantom{p^{\prime\mu\nu}}
{}= \frac{1}{r^2} \left(
\begin{bmatrix}
\mathfrak{h}^{\prime ab}\, Y & \tfrac{\mathfrak{j}^{\prime a}}{\B_l}\, \tfrac{1}{r} Y^B \\[1ex]
\tfrac{\mathfrak{j}^{\prime b}}{\B_l}\, \tfrac{1}{r} Y^A & \tfrac{1}{r^2} (\mathfrak{K}'\, \Omega^{AB} Y + 2 \tfrac{\mathfrak{G}'}{\A_l}\, Y^{AB})
\end{bmatrix}
+
\begin{bmatrix}
0 & \tfrac{\mathfrak{h}^{\prime a}}{\B_l}\, \tfrac{1}{r} X^B \\[1ex]
\tfrac{\mathfrak{h}^{\prime b}}{\B_l}\, \tfrac{1}{r} X^A & 2\tfrac{\mathfrak{h}'_2}{\A_l}\, \tfrac{1}{r^2} X^{AB}
\end{bmatrix} \right) {\rm e}^{-{\rm i}\omega t}.
\end{gather}
\begin{subequations}
With these conventions, again reusing the symbol $\sqtwo$ at each stage
of mode separation of \mbox{$p'_{\mu\nu} = \sqtwo p_{\mu\nu}$}, we have
\begin{gather}
\begin{bmatrix} \bm{h}^{\prime ab} \\ \bm{j}^{\prime a} \\ \bm{K}' \\ \bm{G}' \end{bmatrix} =
\sqtwo_e \begin{bmatrix} \bm{h}_{cd} \\ \bm{j}_c \\ \bm{K} \\ \bm{G} \end{bmatrix}
:=
\begin{bmatrix}
\nabla_m r^2 \nabla^m \bm{h}^{ab} \\ \nabla_m 2\B_l r^2 \nabla^m \bm{j}^a \\
\nabla_m 2 r^2 \nabla^m \bm{K} \\ \nabla_m \tfrac{\A_l}{2} r^2 \nabla^m \bm{G}
\end{bmatrix}
- \B_l \begin{bmatrix}
\bm{h}^{ab} \\ 2\B_l \bm{j}^a \\ 2 \bm{K} \\ \tfrac{\A_l}{2} \bm{G}
\end{bmatrix}\nonumber
\\ \hphantom{\begin{bmatrix}\bm{h}^{\prime ab}\\\bm{j}^{\prime a}\\\bm{K}'\\\bm{G}' \end{bmatrix} =}
{}+ \begin{bmatrix}
-4 r^{(a} g^{b)(c} r^{d)} & 4\B_l r^{(a} g^{b)c} & 4 r^a r^b & 0 \\
4\B_l g^{a(c} r^{d)} & -8\B_l r^a r^c & -4\B_l r^a & 2\A_l r^a \\
 4 r^c r^d & -4\B_l r^c & -4 & 0 \\
 0 & 2\A_l r^c & 0 & \A_l
\end{bmatrix}
 \begin{bmatrix}
 \bm{h}_{cd} \\ \bm{j}_c \\ \bm{K} \\ \bm{G} \end{bmatrix}\nonumber
\\ \hphantom{\begin{bmatrix}\bm{h}^{\prime ab}\\\bm{j}^{\prime a}\\\bm{K}'\\\bm{G}' \end{bmatrix} =}
{}+ f_1 \begin{bmatrix}
 0 \\ 2\B_l \bm{j}^a \\ 4 \bm{K} \\ \A_l \bm{G} \end{bmatrix}
- 2f_1 \begin{bmatrix}
 g^{c(a} g^{b)d} - g^{ab}g^{cd} & 0 & g^{ab} & 0 \\ 0 & -\B_l g^{ac} & 0 & 0 \\
 g^{cd} & 0 & -2 & 0 \\ 0 & 0 & 0 & \tfrac{\A_l}{2}
\end{bmatrix}
\begin{bmatrix}
 \bm{h}_{cd} \\ \bm{j}_c \\ \bm{K} \\ \bm{G}\end{bmatrix},
 \label{eq:radial-ev-covar} \noeqref{eq:radial-ev-covar}
\\[1ex]
\begin{bmatrix} \bm{h}^{\prime a} \\ \bm{h}'_2 \end{bmatrix} =
\sqtwo_o \begin{bmatrix} \bm{h}_c \\ \bm{h}_2 \end{bmatrix}
:=\begin{bmatrix}
 \nabla_m 2\B_l r^2 \nabla^m \bm{h}^a \\
 \nabla_m \frac{\A_l}{2} r^2 \nabla^m \bm{h}_2
\end{bmatrix}
- \B_l \begin{bmatrix}
 2\B_l \bm{h}^a \\ \tfrac{\A_l}{2} \bm{h}_2
\end{bmatrix}
+ \begin{bmatrix}
 - 8\B_l r^a r^c & 2\A_l r^a \\ 2\A_l r^c & \A_l
\end{bmatrix}
\begin{bmatrix}
 \bm{h}_c \\ \bm{h}_2
\end{bmatrix} \nonumber
\\ \hphantom{\begin{bmatrix} \bm{h}^{\prime a} \\ \bm{h}'_2 \end{bmatrix} =}
{}+ f_1\begin{bmatrix}
 2\B_l \bm{h}^a \\ \A_l \bm{h}_2
\end{bmatrix}
- 2f_1\begin{bmatrix}
 -\B_l g^{ac} & 0 \\ 0 & \tfrac{\A_l}{2}
\end{bmatrix}
\begin{bmatrix}
 \bm{h}_c \\ \bm{h}_2
\end{bmatrix}\!,
\label{eq:radial-od-covar} \noeqref{eq:radial-od-covar}
\end{gather}
after separating out the $S^2$ dependence, and
after separating out the time dependence the final result is
\begin{gather}
\begin{bmatrix}
\mathfrak{h}^{\prime tr} \\ \mathfrak{j}^{\prime t} \\ \mathfrak{h}^{\prime tt} \\ \mathfrak{h}^{\prime rr} \\ \mathfrak{K}' \\ \mathfrak{j}^{\prime r} \\ \mathfrak{G}'
\end{bmatrix}
=\sqtwo_e
\begin{bmatrix}
 \mathfrak{h}_{tr} \\ \mathfrak{j}_t \\ \mathfrak{h}_{tt} \\ \mathfrak{h}_{rr} \\ \mathfrak{K} \\ \mathfrak{j}_r \\ \mathfrak{G}
\end{bmatrix}
:=
\begin{bmatrix}[r@{\,}l]
 -\partial_r 2r^2 f\partial_r &\mathfrak{h}_{tr}\\[.5ex]
 -\partial_r 2\B_l\tfrac{r^2}{f}f\del_r &\mathfrak{j}_t \\[1ex]
 \partial_r \tfrac{r^2}{f^2} f\partial_r &\mathfrak{h}_{tt} \\[1ex]
 \partial_r f^2r^2 f\partial_r &\mathfrak{h}_{rr}\\[.5ex]
 \partial_r 2r^2 f\partial_r &\mathfrak{K}\\[.5ex]
 \partial_r 2\B_lf r^2 f\partial_r &\mathfrak{j}_r \\[.5ex]
 \partial_r \tfrac{\A_l}{2}r^2 f\partial_r &\mathfrak{G}
\end{bmatrix}
+ \bigg(\frac{\omega^2}{f} - \frac{\B_l}{r^2}\bigg)
\begin{bmatrix}[r@{\,}l]
 -2r^2 &\mathfrak{h}_{tr} \\[.5ex]
 -2\B_l\tfrac{r^2}{f} &\mathfrak{j}_{t} \\[1ex]
 \tfrac{r^2}{f^2} &\mathfrak{h}_{tt} \\[1ex]
 f^2r^2 &\mathfrak{h}_{rr} \\[.5ex]
 2r^2 &\mathfrak{K} \\[.5ex]
 2\B_lfr^2 &\mathfrak{j}_{r} \\[.5ex]
 \tfrac{\A_l}{2}r^2 &\mathfrak{G}
\end{bmatrix}
\\ \quad
{}- 2{\rm i}\omega r \frac{f_1}{f}
\begin{bmatrix}
 0 & 0 &-\tfrac{1}{f} & -f & 0 & 0 & 0 \\
 0 & 0 & 0 & 0 & 0 &-\B_l & 0 \\
 \tfrac{1}{f} & 0 & 0 & 0 & 0 & 0 & 0 \\
 f & 0 & 0 & 0 & 0 & 0 & 0 \\
 0 & 0 & 0 & 0 & 0 & 0 & 0 \\
 0 & \B_l & 0 & 0 & 0 & 0 & 0 \\
 0 & 0 & 0 & 0 & 0 & 0 & 0
\end{bmatrix}
\begin{bmatrix}
 \mathfrak{h}_{tr} \\ \mathfrak{j}_t \\ \mathfrak{h}_{tt} \\ \mathfrak{h}_{rr} \\ \mathfrak{K} \\ \mathfrak{j}_r \\ \mathfrak{G}
\end{bmatrix}\label{eq:radial-ev}
\\[1ex] \quad
{}+
\begin{bmatrix}[c@{\,}c@{\!}c@{~}c@{~}c@{\!\!}c@{\,}c]
 \tfrac{2}{f}(f^2\!+\!1) & -4\B_l & 0 & 0 & 0 & 0 & 0\\[.5ex]
 -4\B_l & -\tfrac{4\B_l f_1}{f} & 0 & 0 & 0 & 0 & 0\\[.5ex]
 0 & 0 & \tfrac{f_1^{2}}{2f^{3}} & - \tfrac{f_1}{2f}(f_1\!+\!4f) & \tfrac{2 f_1}{f} & 0 & 0\\[.5ex]
 0 & 0 & -\tfrac{f_1}{2f}(f_1\!+\!4f) & \tfrac{1}{2} f(f_1^{2}\!-\!8f^{2}) & -2f(3f_1\!-\!2) & 4\B_l f^{2} & 0\\[.5ex]
 0 & 0 & \tfrac{2 f_1}{f} & -2f(3f_1\!-\!2) & 4 (2f_1\!-\!1) & -4\B_l f & 0\\[.5ex]
 0 & 0 & 0 & 4\B_l f^{2} & -4\B_l f & 4\B_l f (3f_1\!-\!\!2) & 2\A_l f \\[.5ex][1ex]
 0 & 0 & 0 & 0 & 0 & 2\A_l f & \A_l
\end{bmatrix}\!
\begin{bmatrix}
 \mathfrak{h}_{tr} \\ \mathfrak{j}_t \\ \mathfrak{h}_{tt} \\ \mathfrak{h}_{rr} \\ \mathfrak{K} \\ \mathfrak{j}_r \\ \mathfrak{G}
\end{bmatrix}\!,
\\
\begin{bmatrix}
 \mathfrak{h}^{\prime t} \\ \mathfrak{h}^{\prime r} \\ \mathfrak{h}'_2
\end{bmatrix}
= \sqtwo_o \begin{bmatrix}
 \mathfrak{h}_t \\ \mathfrak{h}_r \\ \mathfrak{h}_2
\end{bmatrix}
:=
 \begin{bmatrix}[r@{\,}l]
 -\partial_r 2\B_l\tfrac{r^2}{f}f\del_r &\mathfrak{h}_t \\
 \partial_r 2\B_lf r^2 f \partial_r &\mathfrak{h}_r \\
 \partial_r \tfrac{\A_l}{2}r^2 f \partial_r &\mathfrak{h}_2
 \end{bmatrix}
 + \left(\frac{\omega^2}{f} - \frac{\B_l}{r^2}\right)
 \begin{bmatrix}[r@{\,}l]
 -2\B_l\tfrac{r^2}{f} &\mathfrak{h}_t \\
 2\B_l fr^2 &\mathfrak{h}_r \\
 \tfrac{\A_l}{2}r^2 &\mathfrak{h}_2
 \end{bmatrix}
 \\ \hphantom{\begin{bmatrix} \mathfrak{h}^{\prime t} \\ \mathfrak{h}^{\prime r} \\ \mathfrak{h}'_2\end{bmatrix}= }
{} -2{\rm i}\omega r \frac{f_1}{f}
 \begin{bmatrix}[c@{~}c@{~}c]
 0 &-\B_l & 0 \\
 \B_l & 0 & 0 \\
 0 & 0 & 0
\end{bmatrix}
\begin{bmatrix} \mathfrak{h}_t \\ \mathfrak{h}_r \\ \mathfrak{h}_2 \end{bmatrix}
+ \begin{bmatrix}[c@{~}c@{~}c]
 -4\B_l \tfrac{f_1}{f} & 0 & 0 \\
 0 & 4\B_l f (3f_1 - 2) & 2\A_l f \\
 0 & 2\A_l f & \A_l
\end{bmatrix}
\begin{bmatrix}
 \mathfrak{h}_t \\ \mathfrak{h}_r \\ \mathfrak{h}_2
\end{bmatrix}\!.\label{eq:radial-od}
\end{gather}
The operators $\sqtwo_e$ and $\sqtwo_o$
defined in \eqref{eq:radial-ev} and~\eqref{eq:radial-od} are our \emph{radial mode equations}.
\end{subequations}

\begin{rmk}
Note that the final radial mode equations are manifestly self-adjoint,
$\sqtwo_{\rm e}^* = \sqtwo_e$ and $\sqtwo_o^* = \sqtwo_o$. These expressions
are also valid for all values of $\omega$ and $l$. When $l=0$,
$Y_A^{l=0} = X_A^{l=0} = Y_{AB}^{l=0} = X_{AB}^{l=0} = 0$, and when
$l=1$, $Y_{AB}^{l=1} = X_{AB}^{l=1} = 0$, meaning that their
coefficients are not well-defined. By consistency, the corresponding
components of $\sqtwo_e$ and $\sqtwo_o$ simply vanish.
\end{rmk}

Next, we will introduce the operators needed to decouple radial modes
into \emph{pure gauge} and \emph{constraint violating} modes, to set up
the decoupling strategy (cf.~Section~\ref{sec:triang-strategy}). The key
identities, analogous to $E_1 T = T' E_2$, $E_2 D = D' E_1$ and so on,
hold already at the spacetime level:
\begin{equation}
\grf^\nu \overline{\sqtwo p_{\mu\nu}} = \sqone \, \grf^\nu \overline{p}_{\mu\nu},\qquad
 \sqtwo \, \grf_{(\mu} v_{\nu)} = \grf_{(\mu} \sqone v_{\nu)},\qquad
\gf^{\mu\nu} \sqtwo p_{\mu\nu} = \sqzer \, \gf^{\mu\nu} p_{\mu\nu}.
\end{equation}
In addition, analogous to $TD = H_T E_1$ and so on, we have the
composition identities
\begin{equation}
\grf^\nu \overline{\grf_{(\mu} v_{\nu)}} = \frac{1}{2} \sqone v_\mu,\qquad
 \gf^{\mu\nu} \grf_{(\mu} v_{\nu)} = \grf^\mu v_\mu.
\end{equation}

Now we define the operators $T_2$, $T'_2$, $D_2$, $D'_2$, $\tr$, $\tr'$
and their successive mode decompositions (following the same notational
conventions as in Section~\ref{sec:vw} for mode decomposing vectors):

\noindent
for $v_\mu = \grf^\nu \overline{p}_{\mu\nu} =: T_2[p]$,
\begin{gather}
\begin{bmatrix} \bm{v}_a \\ \bm{u} \end{bmatrix}
= T_{2e} \begin{bmatrix} \bm{h}_{cd} \\ \bm{j}_c \\ \bm{K} \\ \bm{G} \end{bmatrix}
:= \frac{1}{r} \begin{bmatrix}
 \frac{1}{r} \nabla^b r^2 \big(\bm{h}_{ab} - \tfrac{1}{2} g_{ab} g^{cd} \bm{h}_{cd} + g_{ab} \bm{K}\big)
 \\
 \tfrac{1}{r^2} \nabla^c r^3 \bm{j}_c
\end{bmatrix}
 \\ \hphantom{\begin{bmatrix} \bm{v}_a \\ \bm{u} \end{bmatrix}=}
{}+ \frac{1}{r} \begin{bmatrix}
 r_a g^{cd} & -\B_l \delta_a^c & 0 & 0 \\
 -\tfrac{1}{2} g^{cd} & 0 & 0 & -\tfrac{\A_l}{2\B_l}
\end{bmatrix}
\begin{bmatrix} \bm{h}_{cd} \\ \bm{j}_c \\ \bm{K} \\ \bm{G} \end{bmatrix}\!,
\\
\bm{w} = T_{2o} \begin{bmatrix} \bm{h}_c \\ \bm{h}_2 \end{bmatrix}
:= \frac{1}{r} \begin{bmatrix}
 \frac{1}{r^2} \nabla^c r^3 & -\frac{\A_l}{2\B_l}
\end{bmatrix}
\begin{bmatrix} \bm{h}_c \\ \bm{h}_2 \end{bmatrix}\!,
\\[2ex]
\frac{1}{{\rm i}\omega} \begin{bmatrix} \mathfrak{v}_t \\ \mathfrak{v}_r \\ \mathfrak{u} \end{bmatrix}
=T_{2e} \begin{bmatrix} \mathfrak{h}_{tr} \\ \mathfrak{j}_t \\ \mathfrak{h}_{tt} \\ \mathfrak{h}_{rr} \\ \mathfrak{K} \\ \mathfrak{j}_r \\ \mathfrak{G} \end{bmatrix}
:= \frac{1}{{\rm i}\omega r}
 \begin{bmatrix}
\tfrac{1}{r} \del_r f r^2 & -\B_l & \tfrac{{\rm i}\omega r}{2} \tfrac{1}{f} & \tfrac{{\rm i}\omega r}{2} f \\[1ex]
\tfrac{{\rm i}\omega r}{f} & 0 & \tfrac{r}{2} \tfrac{1}{f} \del_r & \tfrac{1}{2r^3 f} \del_r r^4 f^2 \\[1ex]
0 & \tfrac{{\rm i}\omega r}{f} & \tfrac{1}{2} \tfrac{1}{f} & -\tfrac{1}{2} f \\[1ex]
 \end{bmatrix}
\begin{bmatrix} \mathfrak{h}_{tr} \\ \mathfrak{j}_t \\ \mathfrak{h}_{tt} \\ \mathfrak{h}_{rr} \end{bmatrix}\!
 \\ \hphantom{\frac{1}{{\rm i}\omega} \begin{bmatrix} \mathfrak{v}_t \\ \mathfrak{v}_r \\ \mathfrak{u} \end{bmatrix}=}
{} + \frac{1}{{\rm i}\omega r}
 \begin{bmatrix} {\rm i}\omega r & 0 & 0 \\ -\tfrac{1}{r} \del_r r^2 & -\B_l & 0 \\
 0 & \tfrac{1}{r^2} \del_r f r^3 & -\tfrac{\A_l}{2\B_l}
 \end{bmatrix}
\begin{bmatrix} \mathfrak{K} \\ \mathfrak{j}_r \\ \mathfrak{G} \end{bmatrix}\!,
\\[2ex]
\frac{1}{{\rm i}\omega} \mathfrak{w}
= T_{2o} \begin{bmatrix} \mathfrak{h}_t \\ \mathfrak{h}_r \\ \mathfrak{h}_2 \end{bmatrix}
:= \frac{1}{{\rm i}\omega r} \begin{bmatrix}
 \frac{{\rm i}\omega r}{f} & \frac{1}{r^2} \del_r f r^3 & - \frac{\A_l}{2\B_l}
\end{bmatrix}
\begin{bmatrix} \mathfrak{h}_t \\ \mathfrak{h}_r \\ \mathfrak{h}_2 \end{bmatrix}\! ;
\end{gather}

\noindent
for $v^{\prime\mu} = \grf_\nu \overline{p}^{\prime\mu\nu} =: T'_2[p']$,
\begin{gather}
\begin{bmatrix} \bm{v}^{\prime a} \\ \bm{u}' \end{bmatrix}
= T'_{2e} \begin{bmatrix} \bm{h}^{\prime cd} \\ \bm{j}^{\prime c} \\ \bm{K}' \\ \bm{G}' \end{bmatrix}
:= \frac{1}{2r} \begin{bmatrix}
2r \delta^a_{(c} \nabla_{d)} & -\delta^a_c & 2r^a & 0 \\
0 & \nabla_c r & \B_l & -2
\end{bmatrix}
\begin{bmatrix} \bm{h}^{\prime cd} \\ \bm{j}^{\prime c} \\ \bm{K}' \\ \bm{G}' \end{bmatrix}\!,
\\
\bm{w}' = T'_{2o} \begin{bmatrix} \bm{h}^{\prime c} \\ \bm{h}'_2 \end{bmatrix}
:= \frac{1}{2r} \begin{bmatrix} \nabla_c r & -2 \end{bmatrix}
\begin{bmatrix} \bm{h}^{\prime c} \\ \bm{h}'_2 \end{bmatrix}\!,
\\
\frac{1}{{\rm i}\omega}
\begin{bmatrix} \mathfrak{v}^{\prime t} \\ \mathfrak{v}^{\prime r} \\ \mathfrak{u}' \end{bmatrix}
= T'_{2e}
\begin{bmatrix}
\mathfrak{h}^{\prime tr} \\ \mathfrak{j}^{\prime t} \\ \mathfrak{h}^{\prime tt} \\ \mathfrak{h}^{\prime rr} \\ \mathfrak{K}' \\ \mathfrak{j}^{\prime r} \\ \mathfrak{G}' \end{bmatrix}
:= \frac{1}{2{\rm i}\omega r}
\begin{bmatrix}[c@{~\,}c@{~\,}c@{~\,}c@{~\,}c@{~\,}c@{~\,}c]
 \tfrac{r}{f} \del_r f & -1 & -2{\rm i}\omega r & 0 & 0 & 0 & 0 \\
 -{\rm i}\omega r & 0 & f_1 f & 2r\del_r - \tfrac{f_1}{f} & -2f & -1 & 0 \\
 0 & -{\rm i}\omega r &
 0 & 0 & \B_l & \del_r r & -2
\end{bmatrix}
\begin{bmatrix}
\mathfrak{h}^{\prime tr} \\ \mathfrak{j}^{\prime t} \\
\mathfrak{h}^{\prime tt} \\ \mathfrak{h}^{\prime rr} \\ \mathfrak{K}' \\ \mathfrak{j}^{\prime r} \\ \mathfrak{G}' \end{bmatrix}\!,
\\
\frac{1}{{\rm i}\omega} \mathfrak{w}=T'_{2o}
\begin{bmatrix} \mathfrak{h}^{\prime t} \\ \mathfrak{h}^{\prime r} \\ \mathfrak{h}'_2 \end{bmatrix}
:= \frac{1}{2{\rm i}\omega r}
\begin{bmatrix} -{\rm i}\omega r & \del_r r & -2 \end{bmatrix}
\begin{bmatrix} \mathfrak{h}^{\prime t} \\ \mathfrak{h}^{\prime r} \\ \mathfrak{h}'_2 \end{bmatrix}\!;
\end{gather}

\noindent
for $p_{\mu\nu} = -(\grf_\mu v_\nu + \grf_\nu v_\mu) =: D_2[v]$,
\begin{gather}
\begin{bmatrix} \bm{h}_{ab} \\ \bm{j}_a \\ \bm{K} \\ \bm{G} \end{bmatrix}
=D_{2e} \begin{bmatrix} \bm{v}_c \\ \bm{u} \end{bmatrix}
:= -\frac{1}{r} \begin{bmatrix}
 r(\nabla_a \bm{v}_b + \nabla_b \bm{v}_a) \\
 r^2 \nabla_a \tfrac{1}{r} \bm{u} \\ 0 \\ 0
\end{bmatrix}
- \frac{1}{r} \begin{bmatrix}
 0 & 0 \\ \delta_a^c & 0 \\ 2 r^c & -\B_l \\ 0 & 2
\end{bmatrix} \begin{bmatrix} \bm{v}_c \\ \bm{u} \end{bmatrix}\!,
\\[1ex]
\begin{bmatrix} \bm{h}_a \\ \bm{h}_2 \end{bmatrix}
= D_{2o}[\bm{w}]
:= -\frac{1}{r} \begin{bmatrix} r^2 \nabla_a \tfrac{1}{r} \\ 2 \end{bmatrix} \bm{w},
\\[1ex]
\frac{1}{{\rm i}\omega}
\begin{bmatrix} \mathfrak{h}_{tr} \\ \mathfrak{j}_t \\ \mathfrak{h}_{tt} \\ \mathfrak{h}_{rr} \\ \mathfrak{K} \\ \mathfrak{j}_r \\ \mathfrak{G} \end{bmatrix}
= D_{2e} \begin{bmatrix} \mathfrak{v}_t \\ \mathfrak{v}_r \\ \mathfrak{u} \end{bmatrix}
:= -\frac{1}{{\rm i}\omega r} \begin{bmatrix}
 f r \del_r \tfrac{1}{f} & -{\rm i}\omega r & 0 \\
 1 & 0 & -{\rm i}\omega r \\
 -2{\rm i}\omega r & - f_1 f & 0 \\[.5ex]
 0 & 2 r \del_r + \tfrac{f_1}{f} & 0 \\[.5ex]
 0 & 2f & -\B_l \\
 0 & 1 & r^2 \del_r \tfrac{1}{r} \\
 0 & 0 & 2
\end{bmatrix}
\begin{bmatrix} \mathfrak{v}_t \\ \mathfrak{v}_r \\ \mathfrak{u} \end{bmatrix}\!,
\\[1ex]
\frac{1}{{\rm i}\omega}
\begin{bmatrix} \mathfrak{h}_t \\ \mathfrak{h}_r \\ \mathfrak{h}_2 \end{bmatrix}
= D_{2o}[\mathfrak{w}]
:= -\frac{1}{{\rm i}\omega r} \begin{bmatrix}
 -{\rm i}\omega r \\[.5ex] r^2 \del_r \tfrac{1}{r} \\[.5ex] 2
 \end{bmatrix} \mathfrak{w};
\end{gather}

\noindent
for $p^{\prime\mu\nu} = -(\grf^\mu v^{\prime\nu} + \grf^\nu v^{\prime\mu}) =: D'_2[v']$,
\begin{gather}
\begin{bmatrix} \bm{h}^{\prime ab} \\ \bm{j}^{\prime a} \\ \bm{K}' \\ \bm{G}' \end{bmatrix}
=D'_{2e}\begin{bmatrix} \bm{v}^{\prime c} \\ \bm{u}' \end{bmatrix}
:= -\frac{2}{r} \begin{bmatrix}
 r^3 \nabla^{(a} \tfrac{1}{r^2} \bm{v}^{\prime b)}
-\tfrac{1}{2} g^{ab} r^3 \nabla_c \tfrac{1}{r^2} \bm{v}^{\prime c} \\[.5ex]
 r^4 \nabla^a \tfrac{1}{r^3} \bm{u}' \\[.5ex]
 r^3 \nabla_c \tfrac{1}{r^2} \bm{v}^{\prime c} \\ 0
\end{bmatrix}
- \frac{2}{r} \begin{bmatrix}
 -r_c g^{ab} & \tfrac{1}{2} g^{ab} \\ \B_l \delta^a_c & 0 \\ 0 & 0 \\ 0 & \tfrac{\A_l}{2\B_l}
\end{bmatrix}
\begin{bmatrix} \bm{v}^{\prime c} \\ \bm{u}' \end{bmatrix}\!,
\\
\begin{bmatrix} \bm{h}^{\prime a} \\ \bm{h}'_2 \end{bmatrix}
= D'_{2o}[\bm{w}']
:= -\frac{2}{r} \begin{bmatrix}
 r^4 \nabla^a \tfrac{1}{r^3} \\ \tfrac{\A_l}{2\B_l}
\end{bmatrix} \bm{w}',
\\
\frac{1}{{\rm i}\omega}
\begin{bmatrix}
\mathfrak{h}^{\prime tr} \\ \mathfrak{j}^{\prime t} \\ \mathfrak{h}^{\prime tt}
\\ \mathfrak{h}^{\prime rr} \\ \mathfrak{K}' \\ \mathfrak{j}^{\prime r} \\ \mathfrak{G}'
\end{bmatrix}
= D_{2e}'
\begin{bmatrix} \mathfrak{v}^{\prime t} \\ \mathfrak{v}^{\prime r} \\ \mathfrak{u}' \end{bmatrix}
:= -\frac{2}{{\rm i}\omega r} \begin{bmatrix}
 f r^3 \del_r \tfrac{1}{r^2} & \tfrac{{\rm i}\omega r}{f} & 0 \\
 \B_l & 0 & \tfrac{{\rm i}\omega r}{f} \\
 \tfrac{{\rm i}\omega r}{2f} & \tfrac{1}{2} r\del_r \tfrac{1}{f} & -\tfrac{1}{2f} \\[.5ex]
 \tfrac{{\rm i}\omega r f}{2} & \tfrac{1}{2} f^2 r^5 \del_r \tfrac{1}{r^4 f} & \tfrac{f}{2} \\[.5ex]
 {\rm i}\omega r & -r^3 \del_r \tfrac{1}{r^2} & 0 \\[.5ex]
 0 & \B_l & f r^4 \del_r \tfrac{1}{r^3} \\
 0 & 0 & \tfrac{A_l}{2\B_l}
\end{bmatrix}
\begin{bmatrix} \mathfrak{v}^{\prime t} \\ \mathfrak{v}^{\prime r} \\ \mathfrak{u}' \end{bmatrix}\!,
\\
\frac{1}{{\rm i}\omega}
\begin{bmatrix} \mathfrak{h}^{\prime t} \\ \mathfrak{h}^{\prime r} \\ \mathfrak{h}'_2 \end{bmatrix}
=D_{2o}'[\mathfrak{w}']
:= -\frac{2}{{\rm i}\omega r} \begin{bmatrix}
 \tfrac{{\rm i}\omega r}{f} \\ f r^4 \del_r \frac{1}{r^3} \\ \tfrac{\A_l}{2\B_l}
\end{bmatrix} \mathfrak{w}'.
\end{gather}
Some of the above radial mode operators are related by taking adjoints,
namely
\begin{gather}
T'_{2e} = -\frac{1}{2} D_{2e}^*, \qquad
D'_{2e} = -2 T_{2e}^*,
 \\
T'_{2o} = -\frac{1}{2} D_{2o}^*,\qquad
D'_{2o} = -2 T_{2o}^*.
\end{gather}

For the trace operators again we return to the convention for
scalars from Section~\ref{sec:sw}:
using $\mathfrak{z}\, Y {\rm e}^{-{\rm i}\omega t} = \bm{z}\, Y= \gf^{\la\ka} p_{\la\ka} =: \tr[p]$ we get
\begin{gather}
\bm{z} = \tr_e \begin{bmatrix}
 \bm{h}_{cd} \\ \bm{j}_c \\ \bm{K} \\ \bm{G}
\end{bmatrix}
:= \begin{bmatrix} g^{ab} & 0 & 2 & 0 \end{bmatrix}
\begin{bmatrix} \bm{h}_{cd} \\ \bm{j}_c \\ \bm{K} \\ \bm{G} \end{bmatrix}\!,
\\
\mathfrak{z}= \tr_e \begin{bmatrix}
 \mathfrak{h}_{tr} \\ \mathfrak{j}_t \\ \mathfrak{h}_{tt} \\
 \mathfrak{h}_{rr} \\ \mathfrak{K} \\ \mathfrak{j}_r \\ \mathfrak{G}
\end{bmatrix}
:= \begin{bmatrix} 0 & 0 & -\frac{1}{f} & f & 2 & 0 & 0 \end{bmatrix}
\begin{bmatrix}
 \mathfrak{h}_{tr} \\ \mathfrak{j}_t \\ \mathfrak{h}_{tt} \\
 \mathfrak{h}_{rr} \\ \mathfrak{K} \\ \mathfrak{j}_r \\ \mathfrak{G}
\end{bmatrix}\!,\qquad
\tr_o := 0 ;
\end{gather}
and using $\frac{1}{r^2} \mathfrak{z}'\, Y {\rm e}^{-{\rm i}\omega t}
= \frac{1}{r^2} \bm{z}'\, Y = \gf_{\la\ka} p^{\prime \la\ka} =: \tr' [p']$, we get
\begin{gather}
\bm{z}' = \tr'_e \begin{bmatrix}
 \bm{h}^{\prime cd} \\ \bm{j}^{\prime c} \\ \bm{K}' \\ \bm{G}'
\end{bmatrix}
:= \begin{bmatrix} g_{ab} & 0 & 1 & 0 \end{bmatrix}
\begin{bmatrix}
 \bm{h}^{\prime cd} \\ \bm{j}^{\prime c} \\
 \bm{K}' \\ \bm{G}' \end{bmatrix}\!,
\\
\mathfrak{z} = \tr'_e \begin{bmatrix}
 \mathfrak{h}^{\prime tr} \\ \mathfrak{j}^{\prime t} \\ \mathfrak{h}^{\prime tt} \\
 \mathfrak{h}^{\prime rr} \\ \mathfrak{K}' \\ \mathfrak{j}^{\prime r} \\ \mathfrak{G}'
\end{bmatrix}
:= \begin{bmatrix} 0 & 0 & -f & \frac{1}{f} & 1 & 0 & 0 \end{bmatrix}
\begin{bmatrix}
 \mathfrak{h}^{\prime tr} \\ \mathfrak{j}^{\prime t} \\ \mathfrak{h}^{\prime tt} \\
 \mathfrak{h}^{\prime rr} \\ \mathfrak{K}' \\ \mathfrak{j}^{\prime r} \\ \mathfrak{G}'
\end{bmatrix}\!,\qquad
\tr'_o := 0.
\end{gather}

The extra ${\rm i}\omega$ factors in the normalizations have been
chosen so that all operators acting on radial modes are dimensionless
and the first three of the following commutative diagrams remain valid
at each stage of mode separation:
\begin{gather}
\begin{tikzcd}[column sep=huge,row sep=huge]
 \bullet \ar[swap]{d}[description]{\D_{D_2} = \sqone} \ar{r}{D_2} \&
 \bullet \ar{d}[description]{\sqtwo}
 \\
 \bullet \ar[swap]{r}{D'_2} \&
 \bullet
\end{tikzcd},
 \qquad
\begin{tikzcd}[column sep=huge,row sep=huge]
 \bullet \ar[swap]{d}[description]{\sqtwo} \ar{r}{T_2} \&
 \bullet \ar{d}[description]{\D_{T_2} = \sqone}
 \\
 \bullet \ar[swap]{r}{T'_2} \&
 \bullet
\end{tikzcd},
\\
\begin{tikzcd}[column sep=huge,row sep=huge]
 \bullet \ar[swap]{d}[description]{\sqtwo} \ar{r}{\tr} \&
 \bullet \ar{d}[description]{\D_{\tr} = \sqzer}
 \\
 \bullet \ar[swap]{r}{\tr'} \&
 \bullet
\end{tikzcd},
 \qquad
\begin{tikzcd}[column sep=huge,row sep=huge]
 \bullet \ar[swap]{d}{\begin{bmatrix}\sqtwo \\ T_2\end{bmatrix}} \ar{r}{\Phi_2} \&
 \bullet \ar{d}{\D_{\Phi_2}}
 \\
 \bullet \ar[swap]{r}{\begin{bmatrix}\Phi'_2 & -\Delta_{\Phi_2 T_2}\end{bmatrix}} \&
 \bullet
\end{tikzcd}.\label{eq:lich-input}
\end{gather}
The last diagram holds at the radial mode level and shows the decoupling
of \emph{gauge invariant} modes (cf.~Section~\ref{sec:triang-strategy}).
It decomposes into independent even and odd sectors as follows. The odd
sector is described by the following radial mode operators:
\begin{gather}
\Phi_{2o} = \begin{bmatrix} 0 & 2f & -fr\del_r \end{bmatrix}\!,\qquad
\Phi'_{2o} = \frac{\A_l}{\omega^2 r^2}
\begin{bmatrix} 0 & \frac{1}{\B_l} & -\frac{2r^3}{\A_l} \del_r \frac{f}{r^2}\end{bmatrix}\!,
\\
\D_{\Phi_{2o}} = \frac{\A_l}{\omega^2} \D_2,\qquad
\Delta_{\Phi_{2o} T_{2o}} = 2f_1 \frac{\A_l}{{\rm i}\omega r}.
\end{gather}

\begin{rmk}
The operator $\Phi_{2o}$ is indeed gauge invariant, $\Phi_{2o} \circ
D_{2o} = 0$. It is essentially the well-known odd sector
\emph{Regge--Wheeler} scalar~\cite[equation~(2.15)]{berndtson}, whose
decoupling was first identified in~\cite{regge-wheeler}.
\end{rmk}

In the even radial mode sector, we could choose $\D_{\Phi_{2e}}$ to be
proportional to $\D_2$ as well, but the corresponding $\Phi_{2e}$ and
$\Phi'_{2e}$ would be rather large and, at the same time, their
complexity would obscure some important properties that we would like to
highlight. Instead, we choose~$\D_{\Phi_{2e}}$ to be a $2\times 2$ first
order system, with corresponding $\Phi_{2e}$ and $\Phi'_{2e}$ that are
simple%
 \footnote{Such a choice is also practically helpful for reducing the
 size and complexity of certain expressions in the computer algebra
 implementation of intermediate steps of our decoupling strategy.} %
and satisfy the gauge invariance condition $\Phi_{2e} \circ D_{2e} = 0$.
The corresponding formulas are 
\begin{gather*}
\Phi_{2e} = \begin{bmatrix}
 1 & -f\del_r\tfrac{r}{f} & 0 & 0 & \tfrac{{\rm i}\omega r}{f} & -{\rm i}\omega r &
\tfrac{{\rm i}\omega r (\B_l-2+3f_1)}{2f} \\[1.5ex]
 0 & 4{\rm i}\omega r & 2 & 0 & f_1 & 0 & -\tfrac{4\omega^2 r^2 -\B_l f_1}{2}
\end{bmatrix}\!,
\\[1ex]
\Phi'_{2e} = \frac{1}{2{\rm i}\omega r}\begin{bmatrix}[c@{~\,}c@{~\,}c@{~\,}c@{~\,}c@{~\,}c@{~\,}c]
 \tfrac{1}{{\rm i}\omega r} & 0 & f & \tfrac{1}{f} & -1 & 0 & 0 \\[.5ex]
 \tfrac{1}{2f} & \tfrac{1}{\B_l} & 0 & 0 & 0 & 0 & 0\end{bmatrix}
\!+\! \frac{\B_l}{4\A_l} \begin{bmatrix}
 \tfrac{4\omega^2 r^2-\B_l f_1}{{\rm i}\omega r f} \\[.5ex]
 -\tfrac{\B_l-2+3f_1}{f^2}
 \end{bmatrix}
 \begin{bmatrix}[c@{~}c@{~}c@{~}c@{~}c@{~}c@{~}c]
 -\frac{f_1}{2{\rm i}\omega r} & 0 & -f & -\frac{1}{f} & 1 & -\frac{2}{\B_l} & 0
 \end{bmatrix}\!,
 \\[1ex]
\D_{\Phi_{2e}} = \begin{bmatrix}
 0 & \tfrac{f}{{\rm i}\omega} \del_r \tfrac{1}{f} \\ \tfrac{1}{{\rm i}\omega f} \del_r f
\end{bmatrix}
- \frac{\B_l}{4\A_l \omega^2 r^2 f}
\\ \hphantom{\D_{\Phi_{2e}}=}
{}\times\!\begin{bmatrix}[c@{~}c]
 (4\omega^2 r^2 \!-\! \B_l (2f\!+\!f_1))^2 \!+\! 8f (4\omega^2 r^2\!-\!\B_l)
 	& \tfrac{{\rm i}\omega r (4\omega^2 r^2 (\B_l-2+3f_1) - 3\B_l f_1^2 - \A_l (2f+f_1))}{f} \\[1.5ex]
 -\tfrac{{\rm i}\omega r (4\omega^2 r^2 (\B_l-2+3f_1) - 3\B_l f_1^2 - \A_l (2f+f_1))}{f}
 	& \tfrac{\omega^2 r^2 (\B_l-2+3f_1)^2}{f^2}
	\end{bmatrix}\!	,
\\[1ex]
\Delta_{\Phi_{2e} T_{2e}} =\frac{\B_l}{\A_l}
\\
{}\times
\left[\begin{smallmatrix}
\frac{(4\omega^2 r^2 - \B_l f_1) f_1 - 4(\B_l-2)f}{4{\rm i}\omega r f^2}\left(f^2 r \del_r\frac{1}{f} - \omega^2 r^2\right)
& \frac{\B_l f_1 - 2(2\omega^2 r^2-(\B_l-2)f)}{2} \left(r\del_r + \frac{f_1}{2f}\right)
& (\B_l f_1 - 4\omega^2 r^2) \left(r^2\del_r \frac{1}{r} + \frac{\B_l}{2f}\right)
\\[1.5ex]
-\frac{(\B_l-2+3f_1)f_1 + 2(\B_l-2)f}{4 f^2} \left(\frac{f}{r} \del_r \frac{r^2}{f}\right)
 & \frac{{\rm i}\omega r (\B_l-2+3f_1)}{2f} \left(r\del_r + \frac{f_1}{2f}\right)
 & \frac{{\rm i}\omega r (\B_l-2+3f_1)}{f}	\left(r^2\del_r \frac{1}{r} + \frac{\B_l}{2f}\right)
 \end{smallmatrix}\right]
\\ \quad
{} + \frac{\B_l}{\A_l} \begin{bmatrix}
 \tfrac{{\rm i}\omega r (4\omega^2 r^2 -\B_l f_1) (2f+f_1)}{4 f^2}& \tfrac{(\B_l-2)(f_1-2f)}{2}& \A_l
\\[.5ex]
\tfrac{(\omega^2 r^2 + f_1 f) (\B_l-2+3f_1)}{2 f^3} 	& \tfrac{{\rm i}\omega r (\B_l-2)}{2f}
& \tfrac{{\rm i}\omega r (\B_l-2)}{f}
 \end{bmatrix}\!.
\end{gather*}

We will now see that this choice of $\D_{\Phi_{2e}}$ is ultimately
equivalent, with minor caveats, to both the Regge--Wheeler operator
$\D_2$ as well as the corresponding \emph{Zerilli} operator
$\D_2^+$~\cite[equations~(3.30) and~(3.31)]{berndtson}, which is commonly used to
decouple gauge invariant perturbations in the even sector. More
precisely, the equivalence is established by concatenating the following
equivalence diagrams, where the middle vertical arrow corresponds to a
self-adjoint first order reduction of $\D_2$:
\begin{equation} \label{eq:DPhi2e-equiv}
\begin{tikzcd}[column sep=5cm,row sep=4cm]
 \bullet
 \ar[swap]{d}{\D_{\Phi_{2e}}}
 \ar[shift left]{r}{q}
 \&
 \bullet
 \ar[shift left]{l}{\bar{q}}
 \ar[swap]{d}[description]{\alpha \begin{bmatrix}
 \frac{\omega^2 r^2}{f} & -\del_r r \\
 r\del_r & \frac{1}{f} - \frac{\B_l - 2}{\omega^2 r^2}
 \end{bmatrix}}
 \ar[shift left]{r}{\begin{bmatrix} 0 & 1 \end{bmatrix}}
 \&
 \bullet
 \ar[shift left]{l}{\begin{bmatrix}
 \frac{f}{\omega^2 r^2} \del_r r \\
 1
 \end{bmatrix}}
 \ar{d}{\alpha \D_2}
 \\
 \bullet
 \ar[shift left]{r}{q'}
 \&
 \bullet
 \ar[shift left]{l}{\bar{q}'}
 \ar[shift left]{r}{\begin{bmatrix} -r \del_r \frac{f}{\omega^2 r^2} & 1\end{bmatrix}}
 \&
 \bullet
 \ar[shift left]{l}{\begin{bmatrix}
 0 \\ 1
 \end{bmatrix}}
\end{tikzcd}
\end{equation}
where the operators in the first square are given by
\begin{gather}
q:= -\frac{\B_l}{\alpha \A_l} \begin{bmatrix}[c@{~}c]
\tfrac{\A_l (4\omega^2 r^2 - 2f(\B_l-2-f_1) - \B_l f_1)}{2\omega^2 r^2}
 & -\tfrac{\A_l (\B_l-2+3f_1) - 6(\B_l-2) f f_1}{2{\rm i}\omega r f}
\\
2\A_l f - 3f_1 (4\omega^2 r^2 - \B_l f_1)
 & -\tfrac{3{\rm i}\omega r f_1 (\B_l-2+3f_1)}{f} \end{bmatrix}\!,
 \\[.5ex]
q':= \frac{\B_l {\rm i}\omega r}{\A_l} \begin{bmatrix}
 -\tfrac{3{\rm i}\omega r f_1 (\B_l-2+3f_1)}{f} & 2\A_l f - 3f_1 (4\omega^2 r^2 - \B_l f_1)
 \\
\tfrac{\A_l (\B_l-2+3f_1) - 6(\B_l-2) f f_1}{2{\rm i}\omega r f}
 & -\tfrac{\A_l (4\omega^2 r^2 - 2f(\B_l-2-f_1) - \B_l f_1)}{2\omega^2 r^2}\end{bmatrix}
\!= \alpha \, {\rm i}\omega r
\begin{bmatrix}[c@{~}c] 0 & -1 \\ 1 & 0 \end{bmatrix}
 q \begin{bmatrix}[c@{~}c] 0 & 1 \\ 1 & 0 \end{bmatrix}\!,
 \\
\bar{q} := -\frac{\A_l}{\B_l} (q')^*,
\qquad
\bar{q}' := -\frac{\A_l}{\B_l}  q^*,
\end{gather}
and $\alpha$ (equation~\eqref{eq:alpha-def}) is the frequency dependent
constant
\begin{equation}
 \alpha := (12 M \omega)^2 + \A_l^2,
\end{equation}
which vanishes at so-called \emph{algebraically special
frequencies}~\cite{couch-newman}. Notice that $\alpha$ appears in the
denominator of $q$ and $\bar{q}'$, thus creating poles at the
algebraically special frequencies, the only poles in frequency other
than $\omega=0$ that enter the equivalence morphisms with eventual
triangular decoupling of $\sqtwo_e$ (Section~\ref{sec:lichev}). The same
constant $\alpha$ and the poles created by it appears in the
\emph{Chandrasekhar transformations} that relate the Regge--Wheeler
$\D_2$ and Zerilli $\D_2^+$ operators. We give a complete discussion of
this relation in Section~\ref{sec:rwz} and the necessity of these poles.

\begin{rmk} \label{rmk:rw-vs-z}
Using the Zerilli operator $\D_2^+$ (and its first order reduction)
instead of the Regge--Wheeler operator $\D_2$ in~\eqref{eq:DPhi2e-equiv}
results in equivalence maps that are free of poles at $\alpha=0$, the
algebraically special frequencies (in fact, free of all poles in
frequency except $\omega=0$). We have checked that this
change also eliminates the $\alpha=0$ from the equivalence morphisms
with the triangular decoupled form in Section~\ref{sec:lichev}. However,
a different undesirable property appears. As is well-known and can be
explicitly seen from our discussion in Section~\ref{sec:rwz}, both the
Zerilli operator $\D_2^+$ and the corresponding decoupling morphisms
from $\sqtwo_e$ contain $r$-dependent poles in the angular momentum
quantum number $l$. This means that the corresponding radial mode
operators cannot be easily interpreted as the mode separated form of
differential operators on spacetime. On the contrary poles at $\alpha=0$
(since it is a constant independent of $r$ and depending polynomially on
$\omega$ and $\B_l$) do not prevent us from finding such an
interpretation in terms of differential operators on spacetime
(Corollary~\ref{cor:rw-equiv-spacetime}). Also, besides being simpler,
we already have a good understanding of $\D_2$ with respect to
Section~\ref{sec:rw}, while the analogous treatment of $\D_2^+$ would
have to be done separately. For these reasons, we refrain from using the
Zerilli operator in this work.
\end{rmk}

We are now ready to give the full triangular decoupling of the even and
odd sectors of the Lichnerowicz wave equation $\sqtwo p_{\mu\nu} = 0$.

\subsubsection{Odd sector} \label{sec:lichod}

This sector is structurally very similar to the even sector of
the vector wave equation (Section~\ref{sec:vwev}). What
we need to feed into the decoupling strategy of Section~\ref{sec:triang-strategy}
are the odd sectors of the diagrams~\eqref{eq:lich-input} and the
additional composition identities
\begin{equation}
 T_{2o} D_{2o} = \underbrace{\frac{1}{\omega^2 r^2} \frac{1}{\B_l}}_{H_{T_{2o}}} \sqone_o
 \qquad \text{and} \qquad
 \Phi_{2o} D_{2o} = \underbrace{~ 0 ~}_{H_{\Phi_{2o}}} \sqone_o \,.
\end{equation}
The key equivalence diagram~\eqref{eq:res-gauge-equiv} for the
$E_{2o}$-$\Phi_{2o}$-$T_{2o}$ system is then obtained by following
exactly the same steps as for the vector even sector (see~\cite[Section~3.2]{kh-vwtriang}). To save space, we do not give these intermediate
details here. As explained in Section~\ref{sec:triang-strategy}, a
version of these intermediate steps can be recovered from the full
triangular simplification diagram~\eqref{eq:lich2o-equiv}. The final
decoupled triangular form is
\begin{equation} \label{eq:lichod-equiv}
\sqtwo_o \begin{bmatrix}
 \mathfrak{h}_t \\ \mathfrak{h}_r \\ \mathfrak{h}_2
 \end{bmatrix} = 0
 \iff
 \underbrace{{-\frac{2}{\omega^2}} \begin{bmatrix}
 \B_l \D_1 & 0 & \tfrac{1}{3} \B_l^2 \tfrac{f_1}{r^2} \\
 0 & \A_l \D_2 & 0 \\
 0 & 0 & \B_l \D_1
 \end{bmatrix}}_{\tilde{\sqtwo}_o}
 \begin{bmatrix}
 \psi_1 \\ \psi_2 \\ \psi'_1
 \end{bmatrix} = 0.
\end{equation}

While $\tilde{\sqtwo}_o$ is obviously not formally self-adjoint, its
equivalence with the formally self-adjoint $\sqtwo_o$ survives in the
existence of an operator $\Sigma_{2o}$ effecting the equivalence between
$\tilde{\sqtwo}_o$ and $\tilde{\sqtwo}_o^*$, namely $\tilde{\sqtwo}_o
\Sigma_{2o} = \Sigma_{2o} \tilde{\sqtwo}_o^*$, where
\begin{equation}
 \Sigma_{2o} = \begin{bmatrix}
 0 & 0 & 1 \\
 0 & -1 & 0 \\
 1 & 0 & 0
 \end{bmatrix}\!.
\end{equation}

The precise equivalence identities take the form
\begin{equation} \label{eq:lich2o-equiv}
\scalebox{1.2}{\begin{tikzcd}[column sep=huge,row sep=huge]
 \bullet
 \ar[swap]{d}[description]{\sqtwo_o}
 \ar[shift left]{r}{k_{2o}}
 \&
 \bullet
 \ar[shift left]{l}{\bar{k}_{2o}}
 \ar{d}[description]{\tilde{\sqtwo}_o}
 \\
 \bullet
 \ar[shift left]{r}{k'_{2o}}
 \ar[dashed,bend left=40]{u}{h_{2o}}
 \&
 \bullet
 \ar[shift left]{l}{\bar{k}'_{2o}}
 \ar[dashed,bend right=40]{u}[swap]{\tilde{h}_{2o}\,,}
\end{tikzcd}}
\end{equation}
where the operators $\bar{k}_{2o}$, $\bar{k}'_{2o}$, $h_{2o}$,
$\tilde{h}_{2o}$ are 
\begin{gather*}
 \bar{k}_{2o} =\!
\begin{bmatrix}[c@{~}c@{\!}c]
 {\rm i}\omega r & {\rm i}\omega r \tfrac{\B_l}{3}
 & \tfrac{{\rm i}\omega r (3\tfrac{1}{r^2}\del_r f_1 f r^3 + \A_l + 3(3\B_l f - f_1^2))}{18\B_l}
 \\[.5ex]
 -r \del_r & -\tfrac{(\B_l-2)}{3} \tfrac{1}{r^2} \del_r r^3
 & \tfrac{3 f_1 \omega^2 r^2 - r \del_r f (\A_l - 3 f (3\B_l+2f_1))
 - 6\A_l f - 3 f f_1 (7\B_l - 2(3f-f_1))}{18\B_l f}
 \\[.5ex]
 -2 & -2 \tfrac{f}{r^2} \del_r r^3 \!-\! \tfrac{2(\B_l-2)}{3}
 & -\tfrac{3 \del_r f (2\B_l+f_1) r + \A_l + 3 (\B_l (f-2f_1) + f_1 (2f-f_1))}{9\B_l}
 \end{bmatrix}\frac{1}{\omega^2 r^2},
 \\[1ex]
 \bar{k}'_{2o} =\!
 \begin{bmatrix}[c@{~}c@{\!\!}c]
 \tfrac{{\rm i}\omega r}{f} & \tfrac{{\rm i}\omega r}{3 f}
 & \tfrac{{\rm i}\omega r (3\del_r f f_1 + 9 \B_l f + \A_l ) }{18 \B_l f}
 \\[.5ex]
 r^3 f \del_r \tfrac{1}{r^2} & \tfrac{f \del_r r}{3}
 & -\tfrac{3 f_1 r^2 w^2 + \tfrac{1}{r} \del_r f ( 3 f (3 \B_l + 2 f_1) - \A_l ) r^2
 - \A_l ( 2 f - f_1 ) - 6 f ( 2 \B_l (3f+f_1) + 3 f f_1 )}{18 \B_l}
 \\[.5ex]
 \tfrac{\B_l - 2}{2} & \tfrac{3\del_r f r + (\B_l - 2)}{6}
 & \tfrac{\A_l ( 3 f + \B_l - 2 ) + 3 r^2 \del_r f (2\A_l + (\B_l - 2) f_1) \tfrac{1}{r}}{36 \B_l}
 \end{bmatrix}\!,
 \\[1ex]
 h_{2o} = \frac{1}{12 \B_l^2 \omega^2 r^2}
 \begin{bmatrix}
 0 & {\rm i}\omega r f_1 & 0 \\[.5ex]
 -{\rm i}\omega r f_1 & 2 (3\B_l+2f_1) & -2 f_1 \\[.5ex]
 0 & -2 f_1 & \tfrac{24 \B_l^2 f}{\A_l}
 \end{bmatrix}\!,
 \\[1ex]
 \tilde{h}_{2o} =
 \begin{bmatrix}
 \tfrac{1}{36} & -\tfrac{1}{108} & -\tfrac{2(2-9f)-\B_l}{648} \\[.5ex]
 0 & -\tfrac{f}{2 \A_l} & \tfrac{1}{108} \\[.5ex]
 0 & 0 & \tfrac{1}{36}
 \end{bmatrix}
 + \frac{1}{\B_l} \begin{bmatrix}
 -\tfrac{2+9f}{36} & \tfrac{2-9f}{108} & \tfrac{49-102f_1-57f_1^2}{648} \\[.5ex]
 \tfrac{1}{23} & -\tfrac{1}{18} & -\tfrac{2-9f}{108} \\[.5ex]
 \tfrac{1}{2} & -\tfrac{1}{23} & -\tfrac{2+9f}{36}
 \end{bmatrix}
 \\ \hphantom{ \tilde{h}_{2o} =}
 {} + \frac{f_1 f}{\B_l^2} \begin{bmatrix}
 -\tfrac{1}{6} & -\tfrac{1}{36} & -\tfrac{33 f_1 - 20}{216} \\[.5ex]
 0 & 0 & \tfrac{1}{36} \\[.5ex]
 0 & 0 & -\tfrac{1}{6}
 \end{bmatrix}
 + \frac{f_1^2}{\B_l^3} \begin{bmatrix}
 0 & 0 & -\tfrac{\omega^2 r^2 - 2 + 2f_1 (3f + f_1)}{72} \\[.5ex]
 0 & 0 & 0 \\[.5ex]
 0 & 0 & 0
 \end{bmatrix}\!,
\end{gather*}
while the remaining operators can be recovered from the identities
\begin{gather*}
 k_{2o} = \Sigma_{2o} (\bar{k}'_{2o})^*,\qquad
 k'_{2o} = \Sigma_{2o} (\bar{k}_{2o})^*,
\end{gather*}
where also
\begin{gather*}
 h_{2o}^* = h_{2o},\qquad
 \tilde{h}_{2o}^*= \Sigma_{2o} \tilde{h}_{2o} \Sigma_{2o}.
\end{gather*}
When it comes to the properties of diagram~\eqref{eq:lich2o-equiv} with
respect to taking formal adjoints, Remark~\ref{rmk:twist-selfadj} applies here
equally well, with appropriate transposition of notation.

The above upper triangular form of $\tilde{\sqtwo}_o$ allows us to
classify all of its symmetries (or automorphisms in the sense of
Section~\ref{sec:formal}). The key result is the absence of
non-vanishing morphisms between the Regge--Wheeler equations $\D_{s_0}$
and $\D_{s_1}$, except the identity morphism when $s_0 = s_1$
(Theorem~\ref{thm:rw-maps}). This prevents the coupling of the $\D_1$
and $\D_2$ blocks by an automorphism. Also, the single non-vanishing
(and non-removable) off-diagonal element in $\tilde{\sqtwo}_o$ prevents
the exchange of the order of the $\D_1$ blocks. Hence, any automorphism
of $\tilde{\sqtwo}_o$ must also be upper triangular, with every
non-vanishing matrix element proportional to the identity.

With the above logic in mind, the most general automorphism takes the
form $\tilde{\sqtwo}_o A = A \tilde{\sqtwo}_o$, with
\begin{equation}
 A = \begin{bmatrix}
 a_1 & 0 & b_1 \\
 0 & a_2 & 0 \\
 0 & 0 & a_1
 \end{bmatrix}\!,
\end{equation}
parametrized by the $3$ constants $a_1$, $b_1$, $a_2$. It is invertible when
$a_1, a_2 \ne 0$.

With the above choice of $k_{2o}$ and $k'_{2o}$, up to homotopy, we have
the following equivalences of operators:
\begin{gather*}
 k_{2o} \circ D_{2o} \circ \bar{k}_{1o}
\sim \begin{bmatrix} 1 \\ 0 \\ 0 \end{bmatrix}\!,\qquad
 k_{1o} \circ T_{2o} \circ \bar{k}_{2o}
\sim \begin{bmatrix} 0 & 0 & 1 \end{bmatrix}\!.
\end{gather*}

\subsubsection{Even sector} \label{sec:lichev}

The even sector of the Lichnerowicz wave equation is by far the most
complicated case of those we presented in this work. It requires the
full arsenal of tools that we have outlined in Sections~\ref{sec:formal}
and~\ref{sec:rw}. If one blindly attacks it, even with the systematic
strategy of Section~\ref{sec:triang-strategy}, the intermediate formulas can
easily become very large, significantly slowing down even computer
algebra calculations. So, before giving the final result below, we first
make some comments about how we implemented the decoupling strategy.
In particular, points 2--6 below describe how to obtain the key
equivalence diagram~\eqref{eq:res-gauge-equiv} for the
$E_{2e}$-$\Phi_{2e}$-$(tr_e,T_{2e})$ system.
\begin{enumerate}\itemsep=0pt
\item
 We start by taking full advantage of the simultaneous decoupling of
 both $T_{2e}$ and $\tr$ from $\sqtwo_e$ and the relation $\tr_e \circ
 D_{2e}[v] \propto T_{1e}[v]$. Namely, we extend the gauge fixing
 condition $T_2[p]=0$ by $\tr[p]=0$. So the choice of \emph{pure
 gauge}, \emph{constraint violating} and \emph{gauge invariant} modes
 that we feed into the strategy of Section~\ref{sec:triang-strategy}
 corresponds to the diagrams
 \begin{gather}
 \begin{tikzcd}[column sep=2cm,row sep=2cm]
 \bullet \ar[swap]{d}{\begin{bmatrix}\sqone_e \\ T_{1e} \end{bmatrix}} \ar{r}{D_{2e}} \&
 \bullet \ar{d}[description]{\sqtwo_e}
 \\
 \bullet \ar[swap]{r}{\begin{bmatrix} D'_{2e} & 0 \end{bmatrix}} \&
 \bullet
 \end{tikzcd}\!,
 \qquad
 \begin{tikzcd}[column sep=2cm,row sep=2cm]
 \bullet \ar[swap]{d}[description]{\sqtwo_e}
 \ar{r}{\begin{bmatrix} \tr_e \\ T_{2e} \end{bmatrix}} \&
 \bullet \ar{d}[description]{\begin{bmatrix}\sqzer & 0 \\ 0 & \sqone_e\end{bmatrix}}
 \\
 \bullet \ar[swap]{r}{\begin{bmatrix} \tr'_e \\ T'_{2e} \end{bmatrix}} \&
 \bullet
 \end{tikzcd}\!,
 \qquad
 \begin{tikzcd}[column sep=2cm,row sep=2cm]
 \bullet \ar[swap]{d}{\begin{bmatrix}\sqtwo_e \\ \tr_e \\ T_{2e}\end{bmatrix}} \ar{r}{\Phi_{2e}} \&
 \bullet \ar{d}{\D_{\Phi_{2e}}}
 \\
 \bullet \ar[swap]{r}{\begin{bmatrix}\Phi'_{2e} & 0 & -\Delta_{\Phi_{2e} T_{2e}}\end{bmatrix}} \&
 \bullet
 \end{tikzcd}\!,
 \end{gather}
 and the composition identities
 \begin{gather}
 T_{2e} D_{2e}= \underbrace{\frac{1}{\omega^2 r^2} \begin{bmatrix}[ccc|c]
 -f & 0 & 0 & 0 \\
 0 & \tfrac{1}{f} & 0 & 0 \\
 0 & 0 & \tfrac{1}{\B_l} & 0
 \end{bmatrix}}_{H_{T_{2e}}}
 \begin{bmatrix} \sqone_e \\ \cmidrule(lr){1-1} T_{1e} \end{bmatrix} \!,
 \\
 \tr_e D_{2e} = \underbrace{\begin{bmatrix} 0 & -2 \end{bmatrix}}_{H_{tr_e}}
 \begin{bmatrix} \sqone_e \\ T_{1e} \end{bmatrix} \!,
 \\
 \Phi_{2e} D_{2e}= \underbrace{\begin{bmatrix} 0 & 0 \end{bmatrix}}_{H_{\Phi_{2e}}}
 \begin{bmatrix} \sqone_e \\ T_{1e} \end{bmatrix} \!.
 \end{gather}
\item
 Note that by the results of Section~\ref{sec:vwev},
 $\left[\begin{smallmatrix} \sqone_e \\ T_{1e}
 \end{smallmatrix}\right]$ is clearly equivalent to the diagonal
 operator $\left[\begin{smallmatrix} \D_0 & 0 \\ 0 & \B_l \D_1
 \end{smallmatrix}\right]$. We use this equivalence to simplify the
 first diagram above before proceeding.
\item
 Next, we need to simplify the joint system $\sqtwo_e p = 0$,
 $\Phi_{2e}[p] = 0$, $\tr_e[p] = 0$ and $T_{2e}[p] = 0$. It is of
 mixed second, first and zeroth orders. We reduce it to mixed first and
 zeroth orders only, as follows. We use $\sqtwo_e$ to eliminate second
 order derivatives from $\del_r T_{2e}$, giving us six independent
 first order equations. Adding to the list the first component of
 $\Phi_{2e}$, we get seven independent first order equations, from
 which we can solve for the first derivatives of the seven metric
 components $\mathfrak{h}_{tr}$, $\mathfrak{j}_t$, $\mathfrak{h}_{tt}$,
 $\mathfrak{h}_{rr}$, $\mathfrak{K}$, $\mathfrak{j}_r$ and
 $\mathfrak{G}$. Using these first order equations to eliminate all
 derivatives from $\del_r \tr$, we get a zeroth order equation.
 Including also $\tr$ itself and the second component of $\Phi_{2e}$,
 we get three independent zeroth order equations.
\item
 Next, we use these three zeroth order equations to eliminate the
 $\mathfrak{h}_{tr}$, $\mathfrak{h}_{tt}$ and $\mathfrak{h}_{rr}$
 components completely. The choice of these three components is
 dictated by the simplicity of the denominators needed for the
 inversion. What remains is the following first order system for the
 remaining four metric components:
 \begin{equation}
 \begin{bmatrix}
 \del_r r
 & \tfrac{\B_l}{{\rm i}\omega r}
 & -\tfrac{\omega^2 r^2 + \B_l f}{{\rm i}\omega r}
 & -\tfrac{2\omega^2 r^2 - \B_l (\B_l-2f)}{2{\rm i}\omega r}
 \\[.5ex]
 \tfrac{2{\rm i}\omega r}{f}
 & \tfrac{1}{r^2}\del_r r^3 - \tfrac{\B_l}{2f}
 & \B_l
 & -\tfrac{4\omega^2 r^2 + \B_l (\B_l-2f)}{4f}
 \\[.5ex]
 \tfrac{{\rm i}\omega r}{f^2}
 & \tfrac{1}{f}
 & \tfrac{1}{f r^2} \del_r r^3 f
 & -\tfrac{\B_l-2}{2f}
 \\[.5ex]
 0
 & \tfrac{1}{f}
 & -2
 & r\del_r + \tfrac{\B_l}{2f}
 \end{bmatrix}
 \begin{bmatrix}
 \mathfrak{j}_t \\ \mathfrak{K} \\ \mathfrak{j}_r \\ \mathfrak{G}
 \end{bmatrix} = 0.
 \end{equation}
\item
 The above $4\times 4$ \emph{first} order system is then equivalent to
 the $2\times 2$ \emph{second} order $\left[\begin{smallmatrix} \D_0 &
 0 \\ 0 & \B_l \D_1 \end{smallmatrix}\right]$. In one direction, the
 equivalence is given by the appropriately modified $D_{2e}$.
 Obviously, the $2\times 2$ second order system can also be reduced to
 a $4\times 4$ first order system, while homotopy equivalence reduces
 the morphism induced by $\D_{2e}$ to a zeroth order operator (a
 $4\times 4$ matrix, in fact) between these two first order systems.
 This $4\times 4$ matrix is invertible and its inverse, after undoing
 the order reductions, gives the equivalence in the other direction.
\item
 In the above steps, the various eliminations of components introduce
 divisions only by $\omega$, $r$ and $f$, except the final $4\times 4$
 matrix inversion, which introduces division by $\B_l$. Whenever
 encountering choices, e.g.,\ which components to eliminate, we are
 guided by the desire not to introduce other factors into the
 denominators, in order to keep the intermediate expressions in the
 calculations more manageable.
\item
 From this point on, it is enough to continue blindly following the
 strategy from Section~\ref{sec:triang-strategy} to end up with a block upper
 triangular form of $\sqtwo_e$. At this stage, the diagonal blocks are
 not all proportional to Regge--Wheeler operators $\D_s$, but include
 also $\sqzer$, $\sqone_e$ and $\D_{\Phi_{2e}}$. Now we can use the
 results of Sections~\ref{sec:sw}, \ref{sec:vwev} and the
 equivalence~\eqref{eq:DPhi2e-equiv} to transform each of these blocks
 themselves to upper triangular form (using the transformation rules
 from Lemma~\ref{lem:diag-simp}), now with only Regge--Wheeler operators
 on the diagonals.
\item
 It remains now only to apply the methods of Section~\ref{sec:rw} to
 systematically simplify the off-diagonal components of the above upper
 triangular system. Practically speaking, it turns out to be convenient
 to first decouple the~$\D_2$ equations from the rest, then the $\D_1$
 equations from the remaining $\D_0$ equations. Within each still
 coupled block, it is convenient to complete the simplification by
 going through successive super-diagonals.
\end{enumerate}

The final decoupled triangular form (recalling the definition of
$\alpha$ in~\eqref{eq:alpha-def}) of $\sqtwo_e
\mathfrak{p}^{\text{even}} = 0$ is
\begin{gather}
 \underbrace{{-\frac{2}{\omega^2}} 
 \begin{bmatrix}[c@{~}c@{~}c@{~}c@{~}c@{~}c@{~}c]
 \D_0 & 0 & -\frac{f_1}{r^2}\big(\B_l + \tfrac{1}{2}f_1\big) &
 0 & \frac{f_1}{r^2}\big(\B_l + \tfrac{1}{2}f_1\big) & 0 &
 \frac{f_1^2}{8r^2} (7\B_l + 2) \\
 0 & \B_l\D_1 & 0 & 0 & 0 & -\tfrac{5}{3} \B_l^2 \frac{f_1}{r^2} & 0 \\
 0 & 0 & \D_0 & 0 & 0 & 0 & \frac{f_1}{r^2}\big(\B_l + \tfrac{1}{2}f_1\big) \\
 0 & 0 & 0 & \alpha\A_l \D_2 & 0 & 0 & 0 \\
 0 & 0 & 0 & 0 & \D_0 & 0 & -\frac{f_1}{r^2}\big(\B_l + \tfrac{1}{2}f_1\big) \\
 0 & 0 & 0 & 0 & 0 & \B_l\D_1 & 0 \\
 0 & 0 & 0 & 0 & 0 & 0 & \D_0
 \end{bmatrix}
 }_{\tilde{\sqtwo}_e}
 \begin{bmatrix}
 \phi_0 \\ \phi_1 \\ \phi'_0 \\ \phi_2 \\ \chi_0 \\ \chi_1 \\ \chi'_0
 \end{bmatrix}
 \\ \qquad
 {}= 0.\label{eq:lichev-equiv}
\end{gather}

\begin{rmk} \label{rmk:full-triang}
At this point, it is worth noting that equation~\eqref{eq:lichev-equiv}
is the first appearance of the full triangular decoupled
form~$\tilde{\sqtwo}_e$ in the literature. Previously, the original
investigations in~\cite{berndtson} had only produced the $(\phi_0,
\phi_1, \phi_2)$ subsystem, without studying how it couples to the
remaining degrees of freedom, or considering the maximal simplification
of the off-diagonal couplings as we have done here. A similar remark can
be made about~$\tilde{\sqtwo}_o$ from equation~\eqref{eq:lichod-equiv}.
\end{rmk}

While $\tilde{\sqtwo}_e$ is obviously not formally self-adjoint, its
equivalence with the formally self-adjo\-int~$\sqtwo_e$ survives in the
existence of an operator $\Sigma_{2e}$ effecting the equivalence between
$\tilde{\sqtwo}_e$ and $\tilde{\sqtwo}_{\rm e}^*$, namely $\tilde{\sqtwo}_e
\Sigma_{2e} = \Sigma_{2e} \tilde{\sqtwo}_{\rm e}^*$, where
\begin{equation}
 \Sigma_{2e} = \begin{bmatrix}
 0 & 0 & 0 & 0 & 0 & 0 & 1 \\
 0 & 0 & 0 & 0 & 0 & 1 & 0 \\
 0 & 0 & -1 & 0 & 0 & 0 & 0 \\
 0 & 0 & 0 & -1 & 0 & 0 & 0 \\
 0 & 0 & 0 & 0 & -1 & 0 & 0 \\
 0 & 1 & 0 & 0 & 0 & 0 & 0 \\
 1 & 0 & 0 & 0 & 0 & 0 & 0
 \end{bmatrix}\!.
\end{equation}

The precise equivalence identities take the form
\begin{equation} \label{eq:lich2e-equiv}
\scalebox{1.2}{\begin{tikzcd}[column sep=huge,row sep=huge]
 \bullet
 \ar[swap]{d}[description]{\sqtwo_e}
 \ar[shift left]{r}{k_{2e}}
 \&
 \bullet
 \ar[shift left]{l}{\bar{k}_{2e}}
 \ar{d}[description]{\tilde{\sqtwo}_e}
 \\
 \bullet
 \ar[shift left]{r}{k'_{2e}}
 \ar[dashed,bend left=40]{u}{h_{2e}}
 \&
 \bullet
 \ar[shift left]{l}{\bar{k}'_{2e}}
 \ar[dashed,bend right=40]{u}[swap]{\tilde{h}_{2e},}
\end{tikzcd}}
\end{equation}
where the (rather lengthy) explicit formulas for $\bar{k}_{2e}$,
$\bar{k}'_{2e}$ and $h_{2e}$, $\tilde{h}_{2e}$ can be found in
Appendix~\ref{sec:lichev-formulas},%
 \footnote{While undeniably lengthy, the formulas in
 Appendix~\ref{sec:lichev-formulas} are relatively compact. Analogous
 formulas in~\cite[Appendix~A]{berndtson} take up 10 pages and correspond
 only to the $\phi_0$, $\phi_1$ and $\phi_2$ columns of our
 $\bar{k}_{2e}$ operator.} %
while the remaining operators can be recovered from the identities
\begin{equation}
 k_{2e} = \Sigma_{2e} \big(\bar{k}'_{2e}\big)^*,
 \qquad
 k'_{2e} = \Sigma_{2e} \big(\bar{k}_{2e}\big)^*,
\end{equation}
where also
\begin{equation} \label{eq:h2e-adj}
 h_{2e}^* = h_{2e},
 \qquad
 \tilde{h}_{2e}^* = \Sigma_{2e} \tilde{h}_{2e} \Sigma_{2e}.
\end{equation}
When it comes to the properties of diagram~\eqref{eq:lich2e-equiv} with
respect to taking formal adjoints, Remark~\ref{rmk:twist-selfadj}
applies here equally well, with appropriate transposition of notation.

Just in the case of $\tilde{\sqtwo}_o$ (Section~\ref{sec:lichod}), the
above upper triangular form of $\tilde{\sqtwo}_e$ allows us to classify
all of its symmetries (or automorphisms in the sense of
Section~\ref{sec:formal}). Following the same logic, the most general
automorphism takes the form $\tilde{\sqtwo}_e A = A \tilde{\sqtwo}_e$,
with
\begin{equation} \label{eq:lich2e-sym}
 A = \begin{bmatrix}
 a_0 & 0 & c_0+d_0 & 0 & c_0-e_0 & 0 & g_0 \\
 0 & a_1 & 0 & 0 & 0 & b_1 & 0 \\
 0 & 0 & \tfrac{1}{2}(b_0+a_0) & 0 & \tfrac{1}{2}(b_0-a_0) & 0 & c_0-d_0 \\
 0 & 0 & 0 & a_2 & 0 & 0 & 0 \\
 0 & 0 & \tfrac{1}{2}(b_0-a_0) & 0 & \tfrac{1}{2}(b_0+a_0) & 0 & c_0+e_0 \\
 0 & 0 & 0 & 0 & 0 & a_1 & 0 \\
 0 & 0 & 0 & 0 & 0 & 0 & a_0
 \end{bmatrix}\!,
\end{equation}
parametrized by the $9$ constants $a_0$, $b_0$, $c_0$, $d_0$, $e_0$, $g_0$, $a_1$,
$b_1$, $a_2$. It is invertible when $a_0, b_0, a_1, a_2\allowbreak \ne 0$. Note that
it is allowed to deviate from upper triangular form by mixing the two
inner $\D_0$ blocks.

With this choice of $k_{2e}$ and $k'_{2e}$, up to homotopy, we have the
following equivalences of operators:
\begin{align}
 \frac{{\rm i}\omega r}{2}\, \tr \circ \bar{k}_{2e} &\sim
 \begin{bmatrix} 0 & 0 & 1 & 0 & 1 & 0 & 0 \end{bmatrix}\!, \\
 k_{2e} \circ D_{2e} \circ \bar{k}_{1e} &\sim
 \begin{bmatrix}
 1 & 0 & \tfrac{1}{16} (7\B_l+2) \\
 0 & 1 & 0 \\
 0 & 0 & 1 \\
 0 & 0 & 0 \\
 0 & 0 & 0 \\
 0 & 0 & 0 \\
 0 & 0 & 0
 \end{bmatrix}\!, \\
 k_{1e} \circ T_{2e} \circ \bar{k}_{2e} &\sim
 \begin{bmatrix}
 0 & 0 & 0 & 0 & 1 & 0 & \tfrac{1}{16} (7\B_l+2) \\
 0 & 0 & 0 & 0 & 0 & 1 & 0 \\
 0 & 0 & 0 & 0 & 0 & 0 & 1
 \end{bmatrix}\!.
\end{align}

Note that the adjointness properties of the equivalence
diagram~\eqref{eq:lich2e-equiv} do not uniquely fix the operators
$k_{2e}$ and $k'_{2e}$. It turns out that a subfamily
of~\eqref{eq:lich2e-sym}, the symmetries of~$\tilde{\sqtwo}_e$, respects
the adjointness properties of the equivalence when $A^{-1} = \Sigma_{2e}
A^* \Sigma_{2e}$. This constraint is satisfied for instance by the
family $a_0 = a_1 = b_0 = 1$, $d_0 = e_0 = 0$, and $g_0 = c_0^2$. The
remaining $c_0$ parameter may be uniquely fixed by requiring that the
operator equivalent to $\frac{{\rm i}\omega r}{2} \tr$ involves only the two
inner spin-$0$ components $\phi'_0$ and $\chi_0$, and the operator
equivalent to $D_{2e}$ has the right sign. No other ambiguity remains
after that.

\section{Discussion}\label{sec:discussion}

Inspired by previous more ad-hoc work~\cite{berndtson, rosa-dolan} on
decoupling the radial mode equations of the vector wave and Lichnerowicz
wave equations on the Schwarzschild spacetime, in~\cite{kh-vwtriang,
kh-rwtriang} we have initiated a research program of finding a
systematic approach to the problem, as well as simplifying and extending
the results. This work completes part of our program, which consists of
successfully applying the developed methods to both the odd and even
sectors of the Lichnerowicz wave equation on Schwarzschild. The results
include explicit formulas for reducing the complicated mode separated
radial equations to a triangular system of sparsely coupled spin-$s$
Regge--Wheeler equations, by separating the degrees of freedom into
\emph{pure gauge}, \emph{gauge invariant}, and \emph{constraint
violating} modes. Our systematically derived results extend those
of~\cite{berndtson, rosa-dolan} because these earlier works did not
attempt to uncover the full triangular structure of the decoupled
equations or their maximal simplification (in the context of rational
ODEs).

Our main results are summarized in Theorem~\ref{thm:rw-equiv} in
Section~\ref{sec:appl}, where we also list a number of quick corollaries
that indicate important applications of our results to the spectral
analysis of electromagnetic and gravitational perturbations of the
Schwarzschild black hole in harmonic gauges. One application that we
intend to pursue in future work, by leveraging the well-understood
spectral theory of Regge--Wheeler operators~\cite{dss}, is an explicit
mode-level construction of the retarded/advanced Green functions for
these perturbations. Such a construction is a first step toward
classical stability analysis and the construction of quantum field
theoretic propagators. The absence decoupling results, such as ours,
have so far constituted a great obstacle in these applications.

An open question is whether our decoupling strategy could be applied to
Kerr black holes, where so far the harmonic gauges have received very
little attention. The main obstacle to a direct generalization seems to
be the lack of a complete separation of the vector wave and Lichnerowicz
wave equations on the Kerr background. Other than that, the basic
ingredients that would allow the separation of the modes into pure
gauge, gauge invariant, and constraint violating already exists. After
all, metric perturbations are related to pure gauge and constraint
violating modes by the same covariant operators on any spacetime, while
the role of the Regge--Wheeler equations satisfied by gauge invariant
field combinations can be played by the Teukolsky
equations~\cite{chandrasekhar}. This question should be addressed in
future work.

\appendix

\section{Tensor spherical harmonics} \label{sec:spherical}

Here we compare our conventions and notation for mode decomposition into
spherical harmonics, summarized in Section~\ref{sec:schw}, to other closely
related work in the literature. For convenience of notation, we omit the
angular momentum spectral parameters $lm$, where it causes no confusion.

Any spherical harmonic decomposition of vectors and symmetric
$2$-tensors takes the form
\begin{equation}
 v_\mu = \sum_{lm} \sum_{i=1}^4 v^{(i)lm}(r,t) Y^{(i)lm}_{\mu}
 \qquad \text{and} \qquad
 p_{\mu\nu} = \sum_{lm} \sum_{i=1}^{10} p^{(i)lm}(r,t) Y^{(i)lm}_{\mu\nu},
\end{equation}
where the precise choice of the spherical harmonic bases $Y_\mu^{(i)}$
and $Y_{\mu\nu}^{(i)}$ depend on the convention. In this work our
notation in the vector case corresponds to
\begin{gather*}
 v^{(1)} = v_t,\qquad Y^{(1)}_\mu \to \begin{bmatrix} (\d t)_a Y \\ 0 \end{bmatrix}\!,
\\
 v^{(2)} = v_r,\qquad Y^{(2)}_\mu \to \begin{bmatrix} (\d r)_a Y \\ 0 \end{bmatrix}\!,\qquad
\\
 v^{(3)} = u,\qquad Y^{(3)}_\mu \to \begin{bmatrix} 0 \\ r Y_A \end{bmatrix}\!,\qquad
\\
 v^{(4)} = w, \qquad Y^{(4)}_\mu \to \begin{bmatrix} 0 \\ r X_A \end{bmatrix}\!,
\end{gather*}
and in the tensor case to
\begin{gather}
p^{(1)} = h_{tr},\qquad
Y^{(1)}_{\mu\nu} \to \begin{bmatrix} 2(\d t)_{(a}(\d r)_{b)} Y & 0 \\ 0 & 0 \end{bmatrix}\!,
\\
p^{(2)} = j_{t}, \qquad
Y^{(2)}_{\mu\nu}\to \begin{bmatrix} 0 & (\d t)_a r Y_B \\ (\d t)_b r Y_A & 0 \end{bmatrix}\!,
\\
p^{(3)} = h_{tt},\qquad
Y^{(3)}_{\mu\nu} \to \begin{bmatrix} (\d t)_a(\d t)_b Y & 0 \\ 0 & 0 \end{bmatrix}\!,
\\
p^{(4)} = h_{rr}, \qquad
Y^{(4)}_{\mu\nu} \to \begin{bmatrix} (\d r)_a(\d r)_b Y & 0 \\ 0 & 0 \end{bmatrix}\!,
\\
p^{(5)} = K,\qquad
Y^{(5)}_{\mu\nu} \to \begin{bmatrix} 0 & 0 \\ 0 & r^2 \Omega_{AB} Y \end{bmatrix}\!,
\\
p^{(6)} = j_{r}, \qquad
Y^{(6)}_{\mu\nu} \to \begin{bmatrix} 0 & (\d r)_a r Y_B \\ (\d r)_b r Y_A & 0 \end{bmatrix}\!,
\\
p^{(7)} = G, \qquad
Y^{(7)}_{\mu\nu} \to \begin{bmatrix} 0 & 0 \\ 0 & r^2 Y_{AB} \end{bmatrix}\!,
\\
p^{(8)} = h_{t}, \qquad
Y^{(8)}_{\mu\nu} \to \begin{bmatrix} 0 & (\d t)_a r X_B \\ (\d t)_b r X_A & 0 \end{bmatrix}\!,
\\
p^{(9)} = h_{r},\qquad
Y^{(9)}_{\mu\nu} \to \begin{bmatrix} 0 & (\d r)_a r X_B \\ (\d r)_b r X_A & 0 \end{bmatrix}\!,
\\
p^{(10)} = h_{2}, \qquad
Y^{(10)}_{\mu\nu} \to \begin{bmatrix} 0 & 0 \\ 0 & r^2 X_{AB} \end{bmatrix}\!.
\end{gather}

Our conventions on the definition of the spherical harmonics $Y$, $Y_A$,
$X_A$, $Y_{AB}$ and $X_{AB}$ are taken directly
from~\cite{martel-poisson} (MP), where they are compared to
\emph{spin-weighted} and \emph{pure spin} harmonics
in~\cite[Appendix~A]{martel-poisson}. However, our normalizations for the
coefficients differ by some factors of
$r$~\cite[equations~(4.1)--(4.5), (5.1)--(5.4)]{martel-poisson}:
\begin{alignat*}{5}
 &\xi^\MP_{t} = v_{t},\qquad&&
 \xi^\MP_{r} = v_{r},\qquad&&
 \xi^\MP_{\text{even}} = r u,\qquad&&
 \xi^\MP_{\text{odd}} = r w,&
 \\
& h^\MP_{tt} = h_{tt},\qquad&&
 h^\MP_{rr} = h_{rr},\qquad&&
 h^\MP_{tr} = h_{tr},&&&
 \\
& j^\MP_{t} = r j_{t},\qquad&&
 j^\MP_{r} = r j_{r},\qquad&&
 K^\MP = K,\qquad&&
 G^\MP = G,&
 \\
 &h^\MP_{t} = r h_{t},\qquad&&
 h^\MP_{r} = r h_{r},\qquad&&
 h^\MP_{2} = r^2 h_2.&&&
\end{alignat*}
The conventions for the spherical harmonics used
in~\cite[Appendix~A]{barack-lousto} (BL) are different. In~their
notation~\cite[equation~(8)]{barack-lousto}, the trace-reversed metric
perturbation $\bar{p}_{\mu\nu} = p_{\mu\nu} - \frac{1}{2}
p^\lambda_\lambda \; \gf_{\mu\nu}$ is expanded as
\begin{equation}
 \bar{p}_{\mu\nu}
 = \sum_{lm} \sum_{i=1}^{10}
 \frac{(\B_l-2)}{r} \; \bar{p}^{(i)lm}(r,t) \;
 a^{(i)l} \; {}^\BL Y^{(i)lm}_{\mu\nu},
\end{equation}
where they factored out some normalizing constants
$a^{(i)l}$~\cite[equations~(9,10)]{barack-lousto} from their basis. The relation
of their basis to ours is the following, which can be read off by
comparing the corresponding coordinate formulas
in~\cite[Appendix~A]{martel-poisson}
and~\cite[equations~(9,10),Appendix~A]{barack-lousto}:
\begin{gather}
 a^{(1)} \; {}^\BL Y^{(1)}_{\mu\nu} = \frac{1}{2} [Y^{(3)}_{\mu\nu} + f^{-2} Y^{(4)}_{\mu\nu}],\qquad
 a^{(2)} \; {}^\BL Y^{(2)}_{\mu\nu} = \frac{f^{-2}}{2} Y^{(1)}_{\mu\nu},
 \\
 a^{(3)} \; {}^\BL Y^{(3)}_{\mu\nu} = \frac{1}{2} [Y^{(3)}_{\mu\nu} - f^{-2} Y^{(4)}_{\mu\nu}],\qquad
 a^{(4)} \; {}^\BL Y^{(4)}_{\mu\nu} = \frac{1}{2\B_l} Y^{(2)}_{\mu\nu},
 \\
 a^{(5)} \; {}^\BL Y^{(5)}_{\mu\nu} = \frac{f^{-1}}{2\B_l} Y^{(6)}_{\mu\nu},\qquad
 a^{(6)} \; {}^\BL Y^{(6)}_{\mu\nu} = \frac{1}{2} Y^{(5)}_{\mu\nu}, \qquad
 a^{(7)} \; {}^\BL Y^{(7)}_{\mu\nu} = \frac{1}{2\A_l} 2 Y^{(7)}_{\mu\nu},
 \\
 a^{(8)} \; {}^\BL Y^{(9)}_{\mu\nu} = -\frac{1}{2\B_l} Y^{(8)}_{\mu\nu}, \qquad
 a^{(9)} \; {}^\BL Y^{(9)}_{\mu\nu} = -\frac{f^{-1}}{2\B_l} Y^{(9)}_{\mu\nu},\qquad
 a^{(10)} \; {}^\BL Y^{(10)}_{\mu\nu} = -\frac{1}{2\A_l} 2 Y^{(10)}_{\mu\nu}.
\end{gather}

\section{Operators for the Lichnerowicz even sector} \label{sec:lichev-formulas}

Below, we give explicit formulas for the morphism operators
$\bar{k}_{2e}$, $\bar{k}'_{2e}$ and homotopy operators~$h_{2e}$,
$\tilde{h}_{2e}$ that appear in the equivalence
diagram~\eqref{eq:lichev-equiv} in Section~\ref{sec:lichev}.
The morphism operators are first order. Hence, they can be written as
\begin{equation}
 \bar{k}_{2e} = \bar{k}_{2e}^{(1)} \frac{1}{{\rm i}\omega} \del_r + \bar{k}_{2e}^{(0)}
 \qquad \text{and} \qquad
 \bar{k}'_{2e} = \bar{k}_{2e}^{\prime(1)} \frac{1}{{\rm i}\omega} \del_r + \bar{k}_{2e}^{\prime(0)}.
\end{equation}
Let us denote by $\C$ and $\tilde{\C}$ the coefficient matrices of
$\del_r^2$ in $\sqtwo_e$ and $\tilde{\sqtwo}_e$, respectively. They are
clearly diagonal and invertible. By comparing the highest order
coefficients in the identity $\sqtwo_e \bar{k}_{2e} = \bar{k}'_{2e}
\tilde{\sqtwo}_e$, we find the relation
\begin{equation}
 \C \bar{k}_{2e}^{(1)} = \bar{k}_{2e}^{\prime(1)} \tilde{\C},
\end{equation}
which means that there exists a matrix $\K$ such that
\begin{equation} \label{eq:def-calK}
 \bar{k}_{2e}^{\prime(1)} = \K
 \qquad \text{and} \qquad
 \bar{k}_{2e}^{(1)} = \C^{-1} \K \tilde{\C}.
\end{equation}
More explicitly
\begin{gather}
 \C = r^2 \begin{bmatrix}[c@{~~}c@{~~}c@{~~}c@{~~}c@{~~}c@{~~}c]
 -2f & 0 & 0 & 0 & 0 & 0 & 0 \\
 0 & -2\B_l & 0 & 0 & 0 & 0 & 0 \\
 0 & 0 & f^{-1} & 0 & 0 & 0 & 0 \\
 0 & 0 & 0 & f^3 & 0 & 0 & 0 \\
 0 & 0 & 0 & 0 & 2f & 0 & 0 \\
 0 & 0 & 0 & 0 & 0 & 2\B_l f^2 & 0 \\
 0 & 0 & 0 & 0 & 0 & 0 & \tfrac{\A_l}{2} f
 \end{bmatrix}\!,
 \qquad
 \tilde{\C} = -\frac{2f}{\omega^2} \begin{bmatrix}[c@{~~}c@{~~}c@{~~}c@{~~}c@{~~}c@{~~}c]
 1 & 0 & 0 & 0 & 0 & 0 & 0 \\
 0 & \B_l & 0 & 0 & 0 & 0 & 0 \\
 0 & 0 & 1 & 0 & 0 & 0 & 0 \\
 0 & 0 & 0 & \alpha \A_l & 0 & 0 & 0 \\
 0 & 0 & 0 & 0 & 1 & 0 & 0 \\
 0 & 0 & 0 & 0 & 0 & \B_l & 0 \\
 0 & 0 & 0 & 0 & 0 & 0 & 1
 \end{bmatrix}\!.
\end{gather}
The expression for $\K$ is somewhat large, so we display it
split into groups of columns:
\begin{gather}
 \K \begin{bmatrix}[c@{~}c@{~}c@{~}c]
 1 & 0 & 0 & 0 \\ 0 & 1 & 0 & 0 \\ 0 & 0 & 1 & 0 \\ 0 & 0 & 0 & \alpha \\
 0 & 0 & 0 & 0 \\ 0 & 0 & 0 & 0 \\ 0 & 0 & 0 & 0
 \end{bmatrix}
 = \begin{bmatrix}[c@{\!}c@{\!}c@{\!\!\!\!}c]
 2{\rm i} r\omega & - 2{\rm i} r\omega  & -  \tfrac{{\rm i} (15 \B_l + 2) r\omega}{8}
  & \tfrac{{\rm i} (3 f - \B_l - 1) (3 f + \B_l - 3) r\omega}{6}
 \\[0.3ex]
 0 & -{\rm i} r\omega & -{\rm i} \B_l r\omega  & \tfrac{{\rm i} \B_l (3 f - \B_l - 1) r\omega}{3}
 \\[0.3ex]
 - \tfrac{f_1}{2 f} & \tfrac{f_1}{2 f}
  & - \tfrac{(8 f - 15 \B_l - 2) f_1}{32 f}
  & - \tfrac{- 12 f_1 r^2 \omega^2 - \B_l (6 f + \B_l - 2) f_1}{24 f}
 \\[0.3ex]
 - \tfrac{f (3 f + 1)}{2} & \tfrac{f (f + 1)}{2}
  & - \tfrac{\left(\!\!\scriptsize\begin{array}{l}f (24 f^2 - 13 \B_l f \\+ 2 f - 15 \B_l - 2)\end{array}\!\!\right)}{32}
  & - \tfrac{\left(\!\!\scriptsize\begin{array}{l}- 12 ff_1 r^2 \omega^2 - \B_l f (6 f^2 \\- 5 \B_l f + 4 f + \B_l - 2)\end{array}\!\!\right)}{24}
 \\[0.3ex]
 2 f & - f & \tfrac{f (8 f - 7 \B_l - 2)}{8}  & - \tfrac{\B_l f (3 f - \B_l - 1)}{6}
 \\[0.3ex]
 2 \B_l f & 2 f (f - \B_l)  & \tfrac{\B_l f (24 f - 15 \B_l - 2)}{8}
  & - \tfrac{\left(\!\!\scriptsize\begin{array}{l}- 12 ff_1 r^2 \omega^2 - \B_l f (6 f^2 \\- 2 \B_l f - 2 f - \B_l^2 + 2\B_l)\end{array}\!\!\right)}{6}
 \\[0.3ex]
 0 & - \tfrac{(\B_l - 2) f}{2} & - \tfrac{(\B_l - 2) \B_l f}{2}
  & - \tfrac{(\B_l - 2) \B_l f (3 f + 2 \B_l - 1) - 36 ff_1 r^2 \omega^2}{12}
 \end{bmatrix}\! ,
 \\
 \K \begin{bmatrix}[c@{~}c]
 0 & 0 \\ 0 & 0 \\ 0 & 0 \\ 0 & 0 \\ 1 & 0 \\ 0 & 1 \\ 0 & 0
 \end{bmatrix}
= \begin{bmatrix}
 \tfrac{{\rm i} (15 \B_l + 2) r\omega}{8} &  \tfrac{{\rm i} (3 f - 98
 \B_l^2 - 20 \B_l - 3) r\omega}{18 \B_l}
 \\[0.3ex]
{\rm i} \B_l r\omega & -  \tfrac{{\rm i} (15 f + 19 \B_l - 8) r\omega}{18}
 \\[0.3ex]
 \tfrac{(8 f - 15 \B_l - 2) f_1}{32 f} & \tfrac{- 6 f_1 r^2 \omega^2 -
 \B_l (15 f - 49 \B_l - 10) f_1}{36 \B_l f}
 \\[0.3ex]
 \tfrac{f (24 f^2 - 13 \B_l f + 2 f - 15 \B_l - 2)}{32} & \tfrac{- 6
 ff_1 r^2 \omega^2 - \B_l f (45 f^2 - 109 \B_l f - 25 f - 49 \B_l - 10
)}{36 \B_l}
\\[0.3ex]
 - \tfrac{f (8 f - 7 \B_l - 2)}{8} & \tfrac{f (30 f - 79 \B_l - 25)}{18}
 \\[0.3ex]
 - \tfrac{\B_l f (24 f - 15 \B_l - 2)}{8} & \tfrac{2 \B_l f (34
 \B_l f - 5 f - 49 \B_l^2 - 10 \B_l) - 3 ff_1 r^2 \omega^2}{18 \B_l}
 \\[0.3ex]
 \tfrac{(\B_l - 2) \B_l f}{2} & - \tfrac{(\B_l - 2) (19 \B_l - 5) f}{36}
 \end{bmatrix}\! ,
 \\
 \K \begin{bmatrix} 0 \\ 0 \\ 0 \\ 0 \\ 0 \\ 0 \\ 1 \end{bmatrix}
= \begin{bmatrix}
  \tfrac{{\rm i} (16 f^2 - 471 \B_l^2 - 84 \B_l + 20) r\omega}{128}
  \\[0.3ex]
 - \tfrac{{\rm i} \B_l (2 f + 15 \B_l + 4) r\omega}{8}
 \\[0.3ex]
 -\tfrac{- 64 f_1 r^2 \omega^2 - (32 f^2 + 16 \B_l f + 32 f + 471 \B_l^2 + 84 \B_l - 4) f_1}{512 f}
 \\[0.3ex]
 -\tfrac{- 64 ff_1 r^2 \omega^2 - f (96 f^3 + 240 \B_l f^2 + 128 f^2 + 453 \B_l^2 f
 + 76 \B_l f + 20 f + 471 \B_l^2 + 84 \B_l - 4)}{512}
 \\[0.3ex]
 -\tfrac{f (32 f^2 + 64 \B_l f + 32 f + 231 \B_l^2 + 36 \B_l - 4)}{128}
 \\[0.3ex]
 -\tfrac{\B_l f (128 f^2 - 464 \B_l f - 64 f + 471 \B_l^2 + 84 \B_l - 4)}{128}
 \\[0.3ex]
 \tfrac{3 (\B_l - 2) \B_l f (f - 5 \B_l - 1)}{16}
 \end{bmatrix}\! .
\end{gather}
Note that the 4th
column of $\K$ is proportional to $1/\alpha$, which is reflected in the
notation in a~special way.

It is easy to see (comparing the
coefficients of highest order terms of the relevant identities) that the
homotopy corrections $h_{2e}$ and $\tilde{h}_{2e}$ depend only on
$k_{2e}^{(1)}$ and $\bar{k}_{2e}^{(1)}$ coefficients. More explicitly,
plugging in their dependence on $\K$, we obtain
\begin{equation}
 h_{2e} = \C^{-1} \K \tilde{\C} \Sigma_{2e} \K^* \C^{-1},
 \qquad \text{and} \qquad
 \tilde{h}_{2e} = \Sigma_{2e} \K^* \C^{-1} \K.
\end{equation}
Since $\tilde{\C} \Sigma_{2e} = \Sigma_{2e} \tilde{\C}$, while $\C^* =
\C$, $\tilde{\C}^* = \tilde{\C}$, $\Sigma_{2e}^* = \Sigma_{2e}$ and
$\Sigma_{2e}^2 = \id$, the operators $h_{2e}$ and $\tilde{h}_{2e}$
manifestly satisfy the expected self-adjointness
properties~\eqref{eq:h2e-adj}.

The lower order terms $\bar{k}_{2e}^{(0)}$ and
$\bar{k}_{2e}^{\prime(0)}$ are also quite lengthy, so we present them
column by column:%
 \footnote{These formulas are copied from the output of computer algebra.}
\begin{gather}
 \bar{k}_{2e}^{(0)}
\begin{bmatrix}[c@{~\,}c@{~\,}c]
 1 & 0 & 0 \\ 0 & 1 & 0 \\ 0 & 0 & 1 \\ 0 & 0 & 0 \\ 0 & 0 & 0 \\ 0 & 0 & 0 \\ 0 & 0 & 0
 \end{bmatrix} \!=\! \frac{1}{({\rm i}\omega r)^3}\!
\begin{bmatrix}[c@{\,}c@{\,}c]
 \tfrac{{\rm i} (f + 1) r\omega}{f} & - \tfrac{{\rm i} \B_l r\omega}{f}
 & - \tfrac{{\rm i} (8 f^2 - \B_l f - 6 f + 15 \B_l + 2) r\omega}{16 f}
 \\[0.5ex]
 - 2{\rm i} r\omega & 2{\rm i} \B_l r\omega & \tfrac{{\rm i} (15 \B_l + 2) r\omega}{8}
 \\[0.5ex]
 ff_1 - 2 r^2 \omega^2 & 2 \B_l r^2 \omega^2
 & \tfrac{2 (8 f + 15 \B_l + 2) r^2 \omega^2 + (\B_l + 6) ff_1}{16}
 \\[0.5ex]
 - \tfrac{\left(\!\!\scriptsize\begin{array}{l}2 r^2 \omega^2 \\- f (3 f + 2 \B_l + 1)\end{array}\!\!\right)}{f^2} & \tfrac{2 \B_l r^2 \omega^2 - 2 \B_l^2 f}{f^2}
 & - \tfrac{\left(\!\!\scriptsize\begin{array}{l}2(8f - 15\B_l - 2) r^2 \omega^2\\ - f (19 \B_l f\!+\! 18 f \!-\! 30 \B_l^2\! -\! 3 \B_l\! +\! 6)\end{array}\!\!\right)}{16 f^2}
 \\[0.5ex]
 - 2 f - \B_l & \B_l^2 & - \tfrac{10 \B_l f + 12 f - 15 \B_l^2 - 2 \B_l}{16}
 \\[0.5ex]
 - 4 & \tfrac{r^2 \omega^2 + 2 \B_l f}{f} & \tfrac{4 r^2 \omega^2 - f (4 f - 7 \B_l + 2)}{4 f}
 \\[0.5ex]
 2 & - 2 \B_l & \tfrac{8 f - 15 \B_l - 2}{8}
\end{bmatrix}\!,
\\[1ex]
 \bar{k}_{2e}^{\prime(0)}
 \begin{bmatrix}[c@{~\,}c@{~\,}c]
 1 & 0 & 0 \\ 0 & 1 & 0 \\ 0 & 0 & 1 \\ 0 & 0 & 0 \\ 0 & 0 & 0 \\ 0 & 0 & 0 \\ 0 & 0 & 0
 \end{bmatrix}
 \!=\! \frac{1}{{\rm i}\omega r}\!
\begin{bmatrix}[c@{\!\!}c@{\!\!\!}c]
 \tfrac{{\rm i} (f - 3) r\omega}{f}
 & - \tfrac{{\rm i} (2 f - 3) r\omega}{f}
 & \tfrac{{\rm i} (8 f^2 - 31 \B_l f - 10 f + 45 \B_l + 6) r\omega}{16 f}
 \\[0.5ex]
 \tfrac{2{\rm i} \B_l r\omega}{f}
 & - \tfrac{2{\rm i} \B_l r\omega}{f}
 & - \tfrac{{\rm i} \B_l (15 \B_l + 2) r\omega}{8 f}
 \\[0.5ex]
 - \tfrac{2 r^2 \omega^2 - (4 f + 1) f_1}{2 f^2}
 & \tfrac{2 r^2 \omega^2 - (3 f + 1) f_1}{2 f^2}
 & \tfrac{\left(\!\!\!\scriptsize\begin{array}{l}2 (8 f + 15 \B_l + 2) r^2 \omega^2 \\+ (40 f^2 - 44 \B_l f - 8 f
 \\
 - 15 \B_l - 2) f_1\!\!\!\end{array}\right)}{32 f^2}
 \\[0.5ex]
 - \tfrac{\left(\!\!\!\scriptsize\begin{array}{l} 2 r^2 \omega^2 - 12 f^2 \\- 2 \B_l f + f - 1\!\!\!\end{array}\right)}{2}
 & \tfrac{2 r^2 \omega^2 - 3 f^2 - 2 \B_l f - 1}{2}
 & - \tfrac{\left(\!\!\!\scriptsize\begin{array}{l}
 2 (8 f - 15 \B_l - 2) r^2 \omega^2 - 120 f^3
 \\
 + 20 \B_l f^2 + 48 f^2 + 30 \B_l^2 f
 \\
 + 5 \B_l f - 2 f + 15 \B_l + 2\end{array}\!\!\!\right)}{32}
 \\[0.5ex]
 - 8 f - \B_l + 2
 & 3 f + \B_l - 1
 & - \tfrac{80 f^2 - 32 \B_l f - 48 f - 15 \B_l^2 + 12 \B_l + 4}{16}
 \\[0.5ex]
 - 8 \B_l f
 & r^2 \omega^2\! - \!2 f (4 f \!-\! 3 \B_l\! - \!2)
 & \tfrac{2 \B_l r^2 \omega^2 - \B_l f (26 f - 11 \B_l - 12)}{2}
 \\[0.5ex]
 \tfrac{(\B_l - 2) \B_l}{2}
 & \tfrac{(\B_l - 2) (3 f - \B_l - 1)}{2}
 & \tfrac{(\B_l - 2) \B_l (56 f - 15 \B_l - 18)}{32}
\end{bmatrix}\!,
\\[1ex]
 \bar{k}_{2e}^{(0)}
 \begin{bmatrix} 0 \\ 0 \\ 0 \\ 1 \\ 0 \\ 0 \\ 0 \end{bmatrix}
 = \frac{1}{({\rm i}\omega r)^3}
\begin{bmatrix}
 \tfrac{2{\rm i} \B_l^2 r\omega (- 9 f^3 + 24 f^2 - 9 f
 - 2) + 36{\rm i} \B_l f f_1^2 r\omega -
 2{\rm i} \B_l^3 (3 f^2 + 3 f - 2) r\omega +{\rm i} \B_l^4 (3 f
 - 1) r\omega}{12 f}\\[0.5ex]
 \tfrac{- 12{\rm i} \B_l f f_1 r\omega - 2{\rm i} \B_l^2 (3 f^2 - 3 f - 2) r\omega +{\rm i} \B_l^4
 r\omega - 4{\rm i} \B_l^3 r\omega}{6}\\[0.5ex]
 \tfrac{12 \B_l^2 f f_1^2 + 6 \B_l^3 f f_1 (f-2) + 3 \B_l^4 f f_1 }{12}\\[0.5ex]
 \tfrac{12 \B_l^2 f f_1 (1+f) - 2 \B_l^3 (- 3 f^3 - 3 f^2 + 10 f) + \B_l^4
 (- 3 f^2 + 11 f) - 2 \B_l^5 f}{12 f^2}\\[0.5ex]
 \tfrac{- 12 \B_l^2 f_1 f - 2
 \B_l^3 (3 f^2 - 3 f - 2) + \B_l^5 - 4 \B_l^4}{12}\\[0.5ex]
 \tfrac{ 6 \B_l f^2 f_1 - \B_l^2 (- 3 f^3 + 3 f^2 + 4 f) +
 4 f \B_l^3 - \B_l^4 f}{3 f}\\[0.5ex]
 \tfrac{2 \B_l^2 (- 3 f^2 + 9 f - 2) - 12 \B_l f f_1 - 2 \B_l^3
 (3 f - 2) - \B_l^4}{6}
\end{bmatrix}
\\ \hphantom{ \bar{k}_{2e}^{(0)}
 \begin{bmatrix} 0 \\ 0 \\ 0 \\ 1 \\ 0 \\ 0 \\ 0 \end{bmatrix} =}
{} - \frac{1}{{\rm i}\omega r}
\begin{bmatrix}
 \tfrac{2{\rm i} \B_l^2 r\omega (6 f - 6) + 24{\rm i} \B_l f_1 r\omega}{12 f}\\[0.5ex]
 \tfrac{24{\rm i} f_1 r \omega - 12{\rm i} \B_l f_1 r\omega}{6}\\[0.5ex]
 \tfrac{- 2 \B_l^2 (3 f^2 - 6 f - 1) - 8 \B_l^3
 + 2 \B_l^4 + 12 \B_l f_1^2}{12}\\[0.5ex]
 \tfrac{- 2 \B_l^2 (15 f^2 - 18 f - 1)
 - 8 \B_l^3 + 2 \B_l^4 - 12 \B_l (5 f - 1) f_1 }{12 f^2}\\[0.5ex]
 \tfrac{- 12 \B_l^2 f_1 + 24 \B_l f_1 }{12}\\[0.5ex]
 \tfrac{- 2 \B_l (3 f^2 - 6 f + 1) - \B_l^2 (3 f + 1) + \B_l^3 - 12 ff_1 }{3 f}\\[0.5ex]
 \tfrac{12 \B_l^2 - 24 \B_l f + 24 (3 f - 2) f_1 }{6}
\end{bmatrix}\!,
\end{gather}
\begin{gather}
 \bar{k}_{2e}^{\prime(0)}
 \begin{bmatrix} 0 \\ 0 \\ 0 \\ 1 \\ 0 \\ 0 \\ 0 \end{bmatrix}
 = \frac{1}{\alpha} \frac{1}{{\rm i}\omega r}
\begin{bmatrix}
 - \tfrac{{\rm i} (36 f^3 - 6 \B_l f^2 - 42 f^2 + 5 \B_l^2 f - 4 \B_l f - 3
 \B_l^2 + 6 \B_l + 6) r\omega - 12{\rm i} f_1 r^3 \omega^3}{12 f}
 \\[1.2ex]
 - \tfrac{{\rm i} \B_l (6 f^2 - 6 f + \B_l^2 - 2 \B_l) r\omega - 12
{\rm i} f_1 r^3 \omega^3}{6 f}
\\[1.2ex]
 - \tfrac{2 (9 f^2 - 18 f - \B_l^2 + 2 \B_l + 9) r^2 \omega^2 + \B_l
 (24 f^2 - 6 f + \B_l - 2) f_1}{24 f^2}
 \\[1.2ex]
 - \tfrac{2 (21 f^2 - 30 f - \B_l^2 + 2 \B_l + 9) r^2 \omega^2 + \B_l
 (24 f^3 - 12 \B_l f^2 - 6 f^2 + 2 \B_l^2 f - \B_l f + \B_l - 2)}{24}\\[1.2ex]
 \tfrac{\B_l (24 f^2 - 6 \B_l f - 18 f + \B_l^2 + 2) - 12 f_1 r^2
 \omega^2}{12}
 \\[1.2ex]
 - \tfrac{(6 f^2 + 3 \B_l f - 6 f - \B_l^2 - \B_l) r^2 \omega^2 + \B_l f
 (15 f^2 - 4 \B_l f - 13 f + 2 \B_l + 2)}{3}
 \\[1.2ex]
 \tfrac{\left(\!\!\!\scriptsize\begin{array}{l}12 (3 f^2 - 2 \B_l f - 8 f + \B_l^2 + 5) r^2 \omega^2 + 24 \B_l^2
 f^2 - 48 \B_l f^2 \\+ 6 \B_l^3 f - 30 \B_l^2 f + 36 \B_l f - 4 \B_l^3 + 10 \B_l^2 -
 4 \B_l - \alpha\!\!\!\end{array}\right)}{24}
\end{bmatrix}\!,
\\[3ex]
 \bar{k}_{2e}^{(0)}
 \begin{bmatrix} 0 \\ 0 \\ 0 \\ 0 \\ 1 \\ 0 \\ 0 \end{bmatrix}
 = \tfrac{1}{({\rm i}\omega r)^3}
\begin{bmatrix}
 \tfrac{{\rm i} (8 f^2 - \B_l f - 6 f + 15 \B_l + 2) r\omega}{16 f}\\[1.2ex]
 - \tfrac{{\rm i} (15 \B_l + 2) r\omega}{8}\\[1.2ex]
 - \tfrac{2 (15 \B_l + 2) r^2 \omega^2 + (\B_l + 6)
 ff_1}{16}\\[1.2ex]
 - \tfrac{2 (15 \B_l + 2) r^2 \omega^2 + f (19 \B_l f + 18 f - 30
 \B_l^2 - 3 \B_l + 6)}{16 f^2}\\[1.2ex]
 - \tfrac{16 r^2 \omega^2 - 10 \B_l f - 12 f + 15 \B_l^2 + 2 \B_l}{16}\\[1.2ex]
 - \tfrac{4 r^2 \omega^2 - f (4 f - 7 \B_l + 2)}{4 f}\\[1.2ex]
 - \tfrac{8 f - 15 \B_l - 2}{8}
\end{bmatrix}\!,
\\[3ex]
 \bar{k}_{2e}^{\prime(0)}
 \begin{bmatrix} 0 \\ 0 \\ 0 \\ 0 \\ 1 \\ 0 \\ 0 \end{bmatrix}
 = \frac{1}{{\rm i}\omega r}
\begin{bmatrix}
 - \tfrac{{\rm i} (8 f^2 - 31 \B_l f - 10 f + 45 \B_l + 6) r\omega}{16 f}
 \\[1.2ex]
 \tfrac{{\rm i} \B_l (15 \B_l + 2) r\omega}{8 f}\\[1.2ex]
 - \tfrac{2 (15 \B_l + 2) r^2 \omega^2 + (40 f^2 - 44 \B_l f - 8 f
 - 15 \B_l - 2) f_1}{32 f^2}\\[1.2ex]
 - \tfrac{2 (15 \B_l + 2) r^2 \omega^2 + 120 f^3 - 20 \B_l f^2 - 48 f^2
 - 30 \B_l^2 f - 5 \B_l f + 2 f - 15 \B_l - 2}{32}\\[1.2ex]
 - \tfrac{16 r^2 \omega^2 - 80 f^2 + 32 \B_l f + 48 f + 15 \B_l^2 - 12 \B_l - 4}{16}\\[1.2ex]
 - \tfrac{2 \B_l r^2 \omega^2 - \B_l f (26 f - 11 \B_l - 12)}{2}\\[1.2ex]
 - \tfrac{(\B_l - 2) \B_l (56 f - 15 \B_l - 18)}{32}
\end{bmatrix}\!,
\\[3ex]
 \bar{k}_{2e}^{(0)}
 \begin{bmatrix} 0 \\ 0 \\ 0 \\ 0 \\ 0 \\ 1 \\ 0 \end{bmatrix}
 = \frac{1}{({\rm i}\omega r)^3}
\begin{bmatrix}
 - \tfrac{{\rm i} \B_l (30 \B_l f + 3 f + 49 \B_l + 10) r\omega - 6
{\rm i} f_1 r^3 \omega^3}{18 f}\\[1.2ex]
 - \tfrac{- 3{\rm i} f_1 r^3 \omega^3 - 2{\rm i} \B_l^2 (49 \B_l + 10)
 r\omega}{18 \B_l}\\[1.2ex]
 \tfrac{(3 f^2 + 18 \B_l f - 6 f + 98 \B_l^2 + 20 \B_l + 3) r^2 \omega^2
 - 30 \B_l^2 ff_1}{18}\\[1.2ex]
 - \tfrac{(3 f^2 + 18 \B_l f - 98 \B_l^2 - 20 \B_l - 3) r^2 \omega^2 + 2
 \B_l^2 f (30 f + 49 \B_l + 25)}{18 f^2}\\[1.2ex]
 \tfrac{3 (f + 6 \B_l - 1) r^2 \omega^2 + \B_l^2 (45 f + 49 \B_l +
 10)}{18}\\[1.2ex]
 - \tfrac{(6 f^2 + 3 \B_l f - 6 f - 19 \B_l^2 + 8 \B_l) r^2 \omega^2 + 2
 \B_l^2 f (15 f - 79 \B_l - 10)}{18 \B_l f}\\[1.2ex]
 \tfrac{3 (f + 2 \B_l - 1) r^2 \omega^2 + \B_l^2 (15 f - 49 \B_l -
 10)}{9 \B_l}
\end{bmatrix}\!,
\end{gather}
\begin{gather}
 \bar{k}_{2e}^{\prime(0)}
 \begin{bmatrix} 0 \\ 0 \\ 0 \\ 0 \\ 0 \\ 1 \\ 0 \end{bmatrix}
 = \frac{1}{{\rm i}\omega r}
\begin{bmatrix}
 \tfrac{- 6{\rm i} f_1 r^3 \omega^3 -{\rm i} (3 f^2 + 68 \B_l^2 f + 17 \B_l f -
 147 \B_l^2 - 30 \B_l - 3) r\omega}{18 \B_l f}
 \\[0.9ex]
 \tfrac{2{\rm i} \B_l (15 f^2 - 15 f - 49 \B_l^2 - 10 \B_l) r\omega - 3
{\rm i} f_1 r^3 \omega^3}{18 \B_l f}
\\[0.9ex]
 \tfrac{(9 f^2 + 18 \B_l f - 18 f + 98 \B_l^2 + 20 \B_l + 9) r^2 \omega^2
 + \B_l (75 f^2 - 177 \B_l f - 45 f - 49 \B_l - 10) f_1}{36 \B_l f^2}
 \\[0.9ex]
 \tfrac{\left(\!\!\!\scriptsize\begin{array}{l}(3 f^2 - 18 \B_l f - 12 f + 98 \B_l^2 + 20 \B_l + 9) r^2 \omega^2 \\
 + \B_l (225 f^3 - 387 \B_l f^2 - 210 f^2 - 98 \B_l^2 f + 10 \B_l f + 15 f
 - 49 \B_l - 10)\end{array}\!\!\!\right)}{36 \B_l}
 \\ [0.9ex]
 \tfrac{3 (f + 6 \B_l - 1) r^2 \omega^2 - \B_l (150 f^2 - 282 \B_l
 f - 165 f - 49 \B_l^2 + 69 \B_l + 25)}{18 \B_l}
 \\[0.9ex]
 - \tfrac{(6 f^2 + 3 \B_l f - 6 f - 19 \B_l^2 + 8 \B_l) r^2 \omega^2 + 2
 \B_l f (151 \B_l f - 20 f - 177 \B_l^2 - 98 \B_l + 10)}{18 \B_l}
 \\[0.9ex]
 \tfrac{\left(\!\!\!\scriptsize\begin{array}{l}3 (588 f^2 + \B_l f - 1178 f + 2 \B_l^2 - 5 \B_l + 590) r^2 \omega^2 + 72 \B_l^3 f \\ - 159 \B_l^2 f + 30 \B_l f - 127 \B_l^3 + 259 \B_l^2 - 10 \B_l -
 49 \alpha\end{array}\!\!\!\right)}{36 \B_l}
\end{bmatrix}\!,
\\[3ex]
 \bar{k}_{2e}^{(0)}
 \begin{bmatrix} 0 \\0 \\ 0 \\ 0 \\ 0 \\ 0 \\ 1 \end{bmatrix}
 = \frac{1}{({\rm i}\omega r)^3}
\begin{bmatrix}
 \tfrac{{\rm i} (80 \B_l f^2 + 64 f^2 + 9 \B_l^2 f + 28 \B_l f + 4 f - 471
 \B_l^2 - 84 \B_l + 4) r\omega - 64{\rm i} f_1 r^3 \omega^3}{256 f}
 \\[0.9ex]
 - \tfrac{{\rm i} (64 f^2 - 32 f - 471 \B_l^2 - 84 \B_l + 4) r\omega}{128}
 \\[0.9ex]
 \tfrac{2 (16 f^2 + 8 \B_l f + 16 f + 471 \B_l^2 + 84 \B_l - 20) r^2
 \omega^2 + f (32 f^2 - 96 \B_l f - 64 f + 9 \B_l^2 - 4 \B_l - 28) f_1}{256}
 \\[0.9ex]
 \tfrac{\left(\!\!\!\scriptsize\begin{array}{l}2 (80 f^2 - 8 \B_l f - 16 f + 471 \B_l^2 + 84 \B_l - 20) r^2 \omega^2 + f (96 f^3 - 160 \B_l f^2 \\- 160 f^2 - 5 \B_l^2 f - 364 \B_l f - 148
 f - 942 \B_l^3 - 159 \B_l^2 + 4 \B_l - 28)\end{array}\!\!\!\right)}{256 f^2}
 \\[0.9ex]
 - \tfrac{16 (4 f - \B_l - 2) r^2 \omega^2 + 64 f^3 - 128 \B_l f^2 - 128
 f^2 + 2 \B_l^2 f - 136 \B_l f - 56 f - 471 \B_l^3 - 84 \B_l^2 + 4 \B_l}{256}
 \\[0.9ex]
 - \tfrac{8 (6 f - 15 \B_l - 4) r^2 \omega^2 + f (48 f^2 - 56 \B_l
 f - 96 f - 231 \B_l^2 - 44 \B_l - 12)}{64 f}
 \\[0.9ex]
 \tfrac{64 r^2 \omega^2 + 64 f^2 - 16 \B_l f - 128 f - 471 \B_l^2 - 84 \B_l + 4}{128}
\end{bmatrix}\!,
\\[3ex]
 \bar{k}_{2e}^{\prime(0)}
 \begin{bmatrix} 0 \\ 0 \\ 0 \\ 0 \\ 0 \\ 0 \\ 1 \end{bmatrix}
 = \frac{1}{{\rm i}\omega r}
\begin{bmatrix}
 - \tfrac{{\rm i} (96 f^3 + 80 \B_l f^2 - 32 f^2 + 951 \B_l^2 f + 196 \B_l f
 - 36 f - 1413 \B_l^2 - 252 \B_l + 44) r\omega }{256 f}
 \\[0.9ex]
 \tfrac{{\rm i} \B_l (128 f^2 - 96 f - 471 \B_l^2 - 84 \B_l + 4) r\omega}{128 f}
 \\[0.9ex]
 - \tfrac{(192 f^3 + 176 \B_l f^2 + 128 f^2 + 1404 \B_l^2 f + 240 \B_l f
 - 16 f + 471 \B_l^2 + 84 \B_l - 4) f_1}{512 f^2}\\[0.9ex]
 \tfrac{\left(\!\!\!\scriptsize\begin{array}{l} - 576 f^4 - 1360 \B_l f^3 - 320 f^3 - 1364 \B_l^2 f^2 + 128 \B_l f^2 + 176 f^2 \\- 942 \B_l^3 f - 177 \B_l^2 f - 4 \B_l f - 4 f - 471 \B_l^2 - 84 \B_l +
 4\end{array}\!\!\!\right)}{512}
 \\[0.9ex]
 - \tfrac{- 384 f^3 - 768 \B_l f^2 -
 128 f^2 - 1384 \B_l^2 f + 32 \B_l f + 160 f - 471 \B_l^3 + 378 \B_l^2 + 76 \B_l -
 8}{256}
 \\[0.9ex]
 - \tfrac{ - \B_l f (168
 f^2 - 436 \B_l f - 144 f + 351 \B_l^2 + 296 \B_l + 36)}{32}\\[0.9ex]
 \tfrac{- 416 \B_l^2 f^2 + 832 \B_l f^2 + 1424 \B_l^3 f - 2400 \B_l^2 f - 896 \B_l f -
 1506 \B_l^3 + 2920 \B_l^2 + 184 \B_l - 471 \alpha}{512}
\end{bmatrix}
\\[3ex] \hphantom{ \bar{k}_{2e}^{\prime(0)}
 \begin{bmatrix} 0 \\ 0 \\ 0 \\ 0 \\ 0 \\ 0 \\ 1 \end{bmatrix} =}
 - {\rm i}\omega r
\begin{bmatrix}
 \tfrac{ 64{\rm i} f_1 r \omega}{256 f}\\[0.9ex]
 0\\[0.9ex]
 - \tfrac{2 (16 f^2 - 8 \B_l f - 80 f - 471 \B_l^2 - 84 \B_l + 52) }{512 f^2}\\[0.9ex]
 \tfrac{2 (48 f^2 - 8 \B_l f + 48 f + 471 \B_l^2 + 84 \B_l - 52) }{512}\\[0.9ex]
 - \tfrac{16 (4 f - \B_l - 2) }{256}\\[0.9ex]
 - \tfrac{4 \B_l (6 f - 15 \B_l - 4) }{32}\\[0.9ex]
 \tfrac{4 (4239 f^2 - 8478 f + 16 \B_l^2 - 32 \B_l + 4239) }{512}
\end{bmatrix}\!.
\end{gather}

\subsection*{Acknowledgments}
Research of the author was partially supported by the Praemium Academiae
of M.~Markl, GA\v{C}R project GA18-07776S and RVO:~67985840. Early stages
of this work were completed while the author was affiliated with the
University of Rome II (Tor Vergata) and with the University of Milan
(Statale). Thanks to Claudio Dappiaggi, Felix Finster, Christian
G\'erard, Dietrich H\"afner, Peter Hintz, Dmitry Jakobson, Niky Kamran,
Vesselin Petkov, Iosif Polterovich, Ko Sanders, Werner Seiler and Artur
Sergyeyev for interesting discussions. Thanks to Stefanos Aretakis and
Mihalis Dafermos for their hospitality during a visit to Princeton
University, where an early version of this work was first presented in
the Spring of 2016. The author also thanks Francesco Bussola for
checking the mode decompositions~\eqref{eq:radial-ev}
and~\eqref{eq:radial-od} as part of his MSc thesis~\cite{bussola}.
And final thanks go to Roman O.~Popovych and to anonymous referees,
whose comments greatly improved the presentation of the manuscript.

\pdfbookmark[1]{References}{ref}
\LastPageEnding

\end{document}